\documentclass[journal,10pt,twocolumn,twoside]{IEEEtran} 
\IEEEoverridecommandlockouts
\usepackage{amsmath}
\usepackage{amsfonts}
\usepackage{cases}
\usepackage{setspace}
\usepackage{fancybox}
\usepackage{subfigure}
\usepackage{epsfig}
\usepackage{graphicx}
\usepackage{epstopdf}
\usepackage{float}
\usepackage{multirow}
\usepackage{color}%
\usepackage{amsmath}
\usepackage{multirow}

\usepackage{indentfirst}
\usepackage{dsfont}
\usepackage{amsfonts}
\usepackage{times,amsmath,color,amssymb,epsfig,cite,subfigure,algorithm,algorithmic}

\newenvironment{proof}[1][Proof]{\begin{trivlist}
\item[\hskip \labelsep {\bfseries #1}]}{\end{trivlist}}
\newenvironment{definition}[1][Definition]{\begin{trivlist}
\item[\hskip \labelsep {\bfseries #1}]}{\end{trivlist}}

\newcommand{\qed}{\nobreak \ifvmode \relax \else
      \ifdim\lastskip<1.5em \hskip-\lastskip
      \hskip1.5em plus0em minus0.5em \fi \nobreak
      \vrule height0.40em width0.6em depth0.25em\fi}

\newtheorem{lemma}{Lemma}
\newtheorem{theorem}{Theorem}
\newtheorem{fact}{Fact}

\def\x{{\mathbf x}}
\def\W{{\mathbf W}}
\def\A{{\mathbf A}}
\def\I{{\mathbf I}}
\def\v{\mathbf{v}}
\def\u{{\mathbf u}}
\def\e{{\mathbf e}}
\def\p{{\mathbf p}}
\def\z{{\mathbf z}}
\def\r{{\mathbf r}}
\def\y{{\mathbf y}}
\def\b{{\mathbf b}}

\def\F{\mathcal{F}}
\def\D{\mathcal{D}}
\def\P{\mathcal{P}}

\begin{document}
%
\title{Proximal-ADMM Decoder for Nonbinary LDPC Codes}

\author{Yongchao~Wang,~\IEEEmembership{Member,~IEEE,
        Jing~Bai}
\thanks{(\emph{Corresponding author: Jing Bai.} e-mail: ychwang@mail.xidian.edu.cn, j.bai@stu.xidian.edu.cn)}}

\markboth{}%
{}

\maketitle

\begin{abstract}
In this paper, we develop an efficient decoder via the proximal alternating direction method of multipliers (proximal-ADMM) technique
for nonbinary linear block codes in the Galois field. Its main contents are as follows: first, exploiting the decomposition technique based on the three-variables check equation, we formulate the maximum likelihood (ML) decoding problem approximately to a non-convex quadratic program;
second, an efficient algorithm based on the proximal-ADMM technique is proposed to solve the formulated QP problem.
Exploiting the QP problem's inherent structures, its variables can be updated in parallel;
third, we prove that the proposed decoding algorithm can converge to some stationary point of the formulated QP problem. Moreover, we also show, for nonbinary low-density parity-check (LDPC) codes, its computational complexity in each proximal-ADMM iteration scales linearly with block length and the size of the considered Galois field.
Simulation results demonstrate that the proposed proximal-ADMM decoder outperforms state-of-the-art nonbinary LDPC decoders in terms of either error correction performance or computational complexity.
\end{abstract}

\begin{IEEEkeywords}
Nonbinary Low-density Parity-check (LDPC) codes, Galois Field, Proximal Alternating Direction Method of Multipliers (Proximal-ADMM), Quadratic Programming (QP).
\end{IEEEkeywords}

\IEEEpeerreviewmaketitle

\section{Introduction}\label{introduction}

Nonbinary linear block codes, such as nonbinary low-density parity-check (LDPC) codes \cite{Davey-nonBP}, are favorable in high-data-rate communication systems and storage systems \cite{optical-communication,Underwater-Acoustic,Flash-Memories} since they possess many desirable merits from the viewpoints of practical applications.
For example, nonbinary LDPC codes have greater ability to eliminate short cycles (especially 4-cycles) and display better error correction performance \cite{better-performance}.
Moreover, nonbinary LDPC codes have good ability to resist burst errors by combining multiple burst bit errors into fewer nonbinary symbol errors \cite{burst-error}.
Furthermore, nonbinary LDPC codes can provide a higher data transmission rate and spectral efficiency when combined with a higher-order modulation scheme \cite{combined-highmodu}.

Typical decoding algorithms for nonbinary LDPC codes, such as sum-product \cite{Davey-nonBP}\cite{FHT-nonBP}, are based on the belief propagation (BP) strategy.
In \cite{Davey-nonBP}, authors first investigated nonbinary LDPC codes and the corresponding nonbinary BP algorithm.
Moreover, authors in \cite{FHT-nonBP} optimized nonbinary BP algorithms by reducing the computational complexity of the check-node processing.
However, nonbinary BP-like decoding algorithms are heuristic from a theoretical viewpoint since their theoretical performance, such as convergence, cannot be guaranteed. Meanwhile, analyzing the behavior of the nonbinary BP-like decoding algorithms is often difficult and the corresponding results are very limited.

In recent years, mathematical programming (MP) techniques, such as linear programming (LP) and quadratic programming (QP), are proposed to decode LDPC codes. Due to their theoretically-guaranteed decoding performance, MP decoding techniques have attracted significant attention from researchers in the error correction coding/decoding field.
The first MP decoding technique was proposed by Feldman et al. \cite{FeldmanLP}, who relaxed the maximum-likelihood (ML) decoding problem to a linear program for binary LDPC codes. In comparison with the classical BP decoder, two issues must be considered: general LP solving algorithms' complexity, such as interior point method \cite{interior-point} and simplex method \cite{revised-simplex}, are prohibitive in practical applications and the other is its inferior error correction performance in low SNR regions. For the first issue, many new MP decoders for binary LDPC codes are proposed.
In \cite{Barman-ADMM}, authors applied the alternating direction method of the multipliers (ADMM) technique to solve the original LP decoding problem \cite{FeldmanLP}. Its main concern is that each ADMM iteration involves expensive check-polytope  projection operations.
Later, authors in \cite{efficient-projection1} and \cite{efficient-projection2} independently optimized the check-polytope projection algorithm.
Based on the works in \cite{efficient-projection1} and \cite{efficient-projection2}, authors in \cite{look-up-SPL} proposed a hardware-compatible projection algorithm.
In \cite{efficient-projection3}, authors proposed an iterative check-polytope projection algorithm to reduce the complexity of LP decoding.
Moreover, authors in \cite{jiao-zhang} investigated a projection reduction technique to reduce the number of Euclidean projections onto the check polytope.
Authors in \cite{my-wcl} proposed an efficient ADMM-based LP decoding algorithm via the three-variables check-equations decomposition technique.
For the second issue, authors in \cite{adaptive-LP} proposed an adaptive LP decoder to improve LP decoding performance by adaptively adding necessary parity-check constraints. Authors in \cite{cut-plane-algorithm,separation-cut,Adaptive-cut} designed different cut-generating algorithms to eliminate unexpected pseudo-codewords and improve the error correction performance of LP decoding.
In addition, authors in \cite{penalty-decoder} and \cite{bai-admm-qp-binary} independently proposed improved ADMM-based penalized decoding algorithms to enhance the error correction performance of LP decoding in low SNR regions.

In comparison with MP decoding techniques for binary LDPC codes, the corresponding approaches for nonbinary cases are limited.
Of particular relevance is the work \cite{Flanagan} where LP decoding was first generalized to nonbinary LDPC codes.
However, nonbinary LP decoding encounters a similar computational complexity problem as the binary case when using general LP solvers. To overcome this problem, authors in \cite{ILP-nonbinary} and \cite{LCLP-nonbinary} independently extended works in  \cite{ILP-binary} and \cite{LCLP-binary} and applied the coordinate ascent method to solve an approximated dual problem of the original LP decoding problem.
Authors in \cite{Trellis-LP-nonbinary} were inspired by the binary LP decoding idea of \cite{Trellis-LP-binary} and proposed a trellis-based algorithm for check node processing to reduce the complexity of nonbinary LP decoding.
In addition, another nonbinary LP decoding scheme was introduced in \cite{Fast-LP-nonbinary} by using constant-weight binary vectors to represent elements in the Galois field of characteristic two but no efficient algorithm was developed to solve the resulting LP problem.
More recently, authors in \cite{Liu-nonbinary-journal}\cite{Liu-nonbinary-ISIT} extended the binary LP decoding idea of \cite{Barman-ADMM} to the nonbinary case and developed an LP decoding algorithm based on the ADMM technique for nonbinary LDPC codes in the Galois field of characteristic two.
However, the proposed nonbinary ADMM decoder involves time-consuming Euclidean projections onto check polytopes.

In this paper, we focus on designing a convergence-guaranteed decoder with low complexity and favourable error correction performance for nonbinary LDPC codes.
Specifically, the main contributions of this paper are summarized as follows.
\begin{itemize}
  \item Based on the check-equation decomposition method, we first decompose a general multi-variables check equation into a set of three-variables check equations.
        Then, by exploiting the equivalent binary parity-check formulation of every three-variables check equation, we transform the nonbinary ML decoding problem to an equivalent linear integer program (LIP).
        Finally, by adding the quadratic penalty term into the objective of the LIP model, relaxing binary constraints to box constraints, and introducing extra linear constraints, a new quadratic programming (QP) decoding model is established for nonbinary LDPC codes.
  \item We develop a proximal-ADMM algorithm to solve the resulting QP decoding problem. By exploiting the inherent structures of the QP problem, variables in one ADMM update step are updated by blocks, but all of the blocks can be calculated in parallel.
        Meanwhile, variables in other ADMM update steps are computed in parallel.
        Moreover, the complexity of our proposed QP decoding algorithm is cheap in comparison with the state-of-the-art nonbinary ADMM-based decoders since it eliminates time-consuming Euclidean projections onto check polytopes.
  \item We theoretically prove that the proposed proximal-ADMM decoding algorithm converges to a stationary point of the formulated QP decoding problem. In addition, a complexity analysis
        shows that the proposed algorithm scales linearly with block length and the size of the considered Galois field.
\end{itemize}

The rest of this paper is organized as follows. In Section \ref{nonbinary-ML-model}, we briefly introduce the formulation of the ML decoding problem for nonbinary linear block codes.
In Section \ref{Problem formulation}, we present the equivalent binary parity-check formulation of the three-variables check equation and establish a relaxed QP decoding problem for nonbinary linear block codes.
Moreover, an efficient proximal-ADMM algorithm for solving the formulated QP problem is presented in Section \ref{admm-qp-section}.
Section \ref{Analysis-admm-qp-decoding} shows the convergence and complexity analyses of the proposed proximal-ADMM decoding algorithm.
Simulation results demonstrate the effectiveness of our proposed QP decoder in Section \ref{simulation-result}.
Finally, Section \ref{Conclusion} concludes this paper.

\emph{Notations:} in this paper, we just focus on LDPC codes in the Galois field of characteristic two, denoted by $\mathbb{F}_{2^{q}}$.
$\mathbb{R}$ represents real numbers.
Bold lowercase and uppercase letters denote vectors and matrices respectively.
$(\cdot)^{T}$ and $\|\cdot\|_{2}$ symbolize the transpose operation and the 2-norm operator respectively.
$\textrm{diag}(\cdot)$ indicates the operator of element diagonalization.
Let $\delta_{\mathbf{A}}$ denote the spectral norm of matrix $\mathbf{A}$.
The Euclidean projection operator onto a set $\mathcal{X}$ is denoted by $\underset{\mathcal{X}}\Pi$.
$\nabla_{\x} f$ is the derivative of function $f$ with respect to variable $\x$.
$\lambda_{\min}(\A^T\!\A)$ represents the minimum eigenvalue of matrix $\A^T\A$.

\section{ML decoding problem formulation}\label{nonbinary-ML-model}
   Consider a nonbinary linear block codeword defined by an $m$-by-$n$ check matrix $\mathbf{H}$. Its feasible code set is denoted by $\mathcal{C}\in\{0,1,\dotsb,2^{q-1}\}^n$. All the elements in $\mathbf{H}$ belong to set $\{0, 1, \dotsb,\ 2^q-1\}$. Let $\mathbf{h}_j^T$, $j\in\mathcal{J}=\{1,2,\dotsb,m\}$, denote the $j$th row vector of the check matrix $\mathbf{H}$.

   Assume that codeword $\mathbf{u}$ is transmitted through an additional white Gaussian noise (AWGN) channel and its corresponding output is denoted as $\mathbf{r}$. In the receiver, the aim of ML decoding is to determine which codeword has the largest {\it a priori} probability $p(\mathbf{r}|\mathbf{u})$ throughout the feasible codeword set $\mathcal{C}$. So the ML decoding problem can be formulated as
    {\setlength\abovedisplayskip{-0.5pt}
 \setlength\belowdisplayskip{-0.5pt}
  \setlength\jot{1pt}
    \begin{equation}\label{ML}
                \mathbf{u}^* = \underset{\mathbf{u}\in\mathcal{C}}{\rm argmax} \hspace{0.2cm} p(\mathbf{r}|\mathbf{u}),
    \end{equation}
  where} codeword set $\mathcal{C}$ can be expressed as
  {\setlength\abovedisplayskip{-0.5pt}
 \setlength\belowdisplayskip{-0.5pt}
  \setlength\jot{1pt}
  \begin{equation}\label{C}
    \mathcal{C}\!=\!\bigg\{\!\mathbf{u}|\big(\mathbf{h}_j^T\mathbf{u}\big) \oplus 2^q=0, j\in\mathcal{J},\!                               \mathbf{u}\in\!\{0,1,\!\dotsb,\!2^q-1\}^n\!\bigg\}.
  \end{equation}
  According} to the mapping rule \eqref{Mq2b}, any nonzero element $u_i\in \mathbf{u}$ can be one-to-one correspondent to a $(2^q-1)$-length binary vector $\mathbf{x}_i=[x_{i,1},\dotsb,
  x_{i,\sigma},\dotsb, x_{i,2^q-1}]^T$, where
   {\setlength\abovedisplayskip{-0.5pt}
 \setlength\belowdisplayskip{-0.5pt}
  \setlength\jot{1pt}
    \begin{equation}\label{Mq2b}
       x_{i,\sigma}= \begin{cases}1,&\ \sigma=u_i, \\0,& \ \sigma\neq u_i. \end{cases}
    \end{equation}
 Besides,} we map the zero element to a $(2^q-1)$-length all-zeros vector. Then, codeword $\u$ can be mapped one-to-one correspondent to a binary vector $\mathbf{x}=[\mathbf{x}_1^T,\dotsb,\mathbf{x}_n^T]^T\in\{0,1\}^{n(2^q-1)}$. Here, we call $\mathbf{x}$ as the equivalent binary codeword to the nonbinary codeword $\mathbf{u}$. Let $\mathcal{X}$ denote the set consisting of all of the equivalent binary codewords. Then, the ML decoding problem \eqref{ML} is equivalent to
{\setlength\abovedisplayskip{-0.5pt}
 \setlength\belowdisplayskip{-0.5pt}
  \setlength\jot{1pt}
     \begin{equation}\label{ML_X}
                \mathbf{x}^* = \underset{\mathbf{x}\in\mathcal{X}}{\rm argmax} \hspace{0.2cm} p(\mathbf{r}|\mathbf{x}),
    \end{equation}
 which} can be further derived as
{\setlength\abovedisplayskip{-2.5pt}
 \setlength\belowdisplayskip{-2.5pt}
  \setlength\jot{-1pt}
 \begin{equation}\label{ML 2}
  \begin{split}
    \underset{\mathbf{x}\in\mathcal{X}}{\rm argmax} \hspace{0.2cm} p(\mathbf{r}|\mathbf{x}) & = \underset{\mathbf{x}\in\mathcal{X}}{\rm argmax} \hspace{0.1cm} \prod_{i=1}^n\prod_{\sigma=1}^{2^q-1} p(r_i|x_{i,\sigma}) \\
    & = \underset{\mathbf{x}\in\mathcal{X}}{\rm argmin} \hspace{0.1cm} \sum_{i=1}^n\sum_{\sigma=1}^{2^q-1} -p(r_i|x_{i,\sigma}).
    \end{split}
 \end{equation}
 Plugging} constant $\displaystyle\sum_{i=1}^n\sum_{\sigma=1}^{2^q-1} \log p(r_i|x_{i,\sigma}=0)$ into \eqref{ML 2}, we have the following derivations
{\setlength\abovedisplayskip{-2.5pt}
 \setlength\belowdisplayskip{-2.5pt}
  \setlength\jot{-1pt}
   \begin{equation}\label{ML_3}
        \begin{split}
            \underset{\mathbf{x}\in\mathcal{X}}{\rm argmax}\ p(\mathbf{r}|\mathbf{x})\!\! &= \underset{\mathbf{x}\in\mathcal{X}}{\rm argmin} \displaystyle\sum_{i=1}^n\sum_{\sigma=1}^{2^q-1} \log\frac{p(r_i|x_{i,\sigma}=0)}{p(r_{i}|x_{i,\sigma})} \\
                                   \!\!\! \hspace{-0.1cm} &=\!\! \underset{\mathbf{x}\in\mathcal{X}}{\rm argmin}\! \displaystyle\sum_{i=1}^n\sum_{\sigma=1}^{2^q-1} \!\! x_{i,\sigma}\!\log\frac{p(r_{i}|x_{i,\sigma}=0)}{p(r_{i}|x_{i,\sigma}=1)} \\
                                     &=  \underset{\mathbf{x}\in \mathcal{X}}{\rm argmin}\ \boldsymbol\gamma^T\mathbf{x},
        \end{split}
    \end{equation}
    where} $\boldsymbol{\gamma}\in\mathbf{R}^{n(2^q-1)}$ is called a cost vector defined by
    \begin{equation}\label{cost vector}
        \begin{split}
        {\boldsymbol{\gamma}} =\bigg[&
                                    \log\frac{p(r_1|x_{1,1}=0)}{p(r_1|x_{1,1}=1)}, \dotsb, \log\frac{p(r_1|x_{1,2^q-1}=0)}{p(r_{1}|x_{1,2^q-1}=1)},\\
             \dotsb,&
                                   \log\frac{p(r_n|x_{n,1}=0)}{p(r_n|x_{n,1}=1)},
                                    \dotsb, \log\frac{p(r_n|x_{n,2^q-1}=0)}{p(r_n|x_{n,2^q-1}=1)}
                                \bigg].
        \end{split}
     \end{equation}
    Then, the ML decoding problem can be cast as the following standard optimization model
    \begin{subequations}\label{ML_x}
        \begin{align}
                &\underset{\mathbf{x}}{\rm min} \hspace{0.31cm} \boldsymbol\gamma^T\mathbf{x},
                 \label{ML_x_a} \\
                & \hspace{0.1cm} {\rm s.t.} \hspace{0.355cm}     \mathbf{x}\in\mathcal{X}. \label{ML_x_b}
        \end{align}
    \end{subequations}
    Problem \eqref{ML_x} is a binary integer program. The difficulty of solving it lies in how to formulate and handle the constraint $\mathbf{x}\in\mathcal{X}$. In the following sections, we first decompose check equations in \eqref{C} into a series of three-variables check equations and second transform the resulting constraints to binary constraints and show analytical expressions for $\mathbf{x}\in\mathcal{X}$, which leads \eqref{ML_x} to a nonconvex quadratic continuous optimization problem. Moreover, an efficient algorithm , named as proximal-ADMM, is proposed to solve the formulated optimization problem. The analyses on convergence and computational complexity of the designed proximal-ADMM algorithm are also presented.

\section{Problem Relaxation}\label{Problem formulation}

\subsection{Three-variables check equation and its equivalent binary parity-check formulation}\label{embed-section}
First, we consider the following three-variables $2^q$-ary check equation in $\mathbb{F}_{2^{q}}$
\begin{equation}\label{three-variabls check euqation}
  \bigg(\displaystyle\sum_{k=1}^3 h_ku_k\bigg)\oplus 2^q =0,
\end{equation}
where $h_k$ is some nonzero constant, $u_k$ is a variable and the addition ``$\oplus$'' and multiplication  are in $\mathbb{F}_{2^{q}}$. Since $h_ku_k$ can be expressed exactly as
\begin{equation}\label{c}
   h_ku_k = \displaystyle\sum_{i=1}^q c_{i,k}2^{i-1},
\end{equation}
where $c_{i,k}\in\{0,1\}$, then the $2^q$-ary check equation \eqref{three-variabls check euqation} is equivalent to the following $q$ three-variables parity-check equations
\begin{equation}\label{three-variable pairty-check}
  \begin{split}
    \bigg(\sum_{k=1}^3 c_{i,k}\bigg)\oplus2=0,\ i=1,2,\dotsb,q.
  \end{split}
\end{equation}
Moreover, the above parity-check equations can be further equivalent to
 \begin{equation}\label{four-inequa}
   \begin{split}
  & c_{i,1} \leq c_{i,2}+c_{i,3},~~ c_{i,2} \leq c_{i,1}+c_{i,3}, \\
  & c_{i,3} \leq c_{i,1}+c_{i,2},~~ c_{i,1}+c_{i,2}+c_{i,3} \leq 2,\\
  & c_{i,1},c_{i,2},c_{i,3} \in \{0,1\},~ i=1,2,\dotsb,q.
   \end{split}
 \end{equation}
Letting
   \begin{equation}\label{t F matrix}
      \begin{split}
         \mathbf{t}=\begin{bmatrix}\ 0\ \\ \ 0\ \\ \ 0\ \\ \ 2\ \end{bmatrix}, \ \
         \mathbf{P} = \begin{bmatrix}
                           ~~1  &-1  & -1 \\
                          -1  &~~1 & -1 \\
                          -1  &-1  &~~1 \\
                          ~~1  &~~1 &~~1
                       \end{bmatrix},
      \end{split}
   \end{equation}
we can rewrite \eqref{four-inequa} as
\begin{equation}\label{three-variable-checks-matrix-form}
  \begin{split}
   \mathbf{P}\mathbf{c}_{i} \preceq \mathbf{t},\ \mathbf{c}_i \in \{0,1\}^3,~ i=1,2,\ldots,q.
  \end{split}
\end{equation}
where $\mathbf{c}_i=[c_{i,1},c_{i,2},c_{i,3}]$.

For any nonzero constant $h_k \in \mathbb{F}_{2^{q}}$, it can be mapped one-to-one to a $(2^q-1)\times(2^q-1)$ matrix , denoted by $\mathbf{D}(2^q,h_k)$, whose entries are determined by
\begin{equation}\label{permu-matrix}
D(2^q,h_k)_{ij}=
\begin{cases}
1,& \textrm{if}~ i=jh_k\oplus 2^q,\\
0,& \textrm{otherwise}.
\end{cases}
\end{equation}
Then, we have the following lemma.
\begin{lemma}\label{hu-dD-u-define}
   For $u_k\in\mathbb{F}_{2^q}$, let $\mathbf{x}_k$ be its corresponding binary vector codeword according to mapping rule \eqref{Mq2b}. Then, $\mathbf{D}(2^q,h_k)\mathbf{x}_k$ is the corresponding binary equivalent codeword for $h_ku_k$.
\end{lemma}
\begin{proof}
  See proof in Appendix \ref{hu-dD-u-define-proof}.
\end{proof}
Moreover, we map any nonzero element $\alpha\in\mathbb{F}_{2^q}$ to a $q$-length binary vector according to rule \eqref{c} and formulate the following $q$-by-($2^q-1$) matrix
\begin{equation}\label{B}
   \mathbf{B}=\begin{bmatrix}
           0&       0&   \dotsb & 1&       1 \\
      \vdots&  \vdots&   \vdots & \vdots& \vdots \\
           0&       1&   \dotsb & 1&       1  \\
           1&       0&   \dotsb & 0&       1
   \end{bmatrix}.
\end{equation}
It is easy to see that $\alpha$ equals $\mathbf{B}$'s column index.
Letting $\mathbf{b}_i^T$ denote the $i$th row vector of matrix $\mathbf{B}$, then we have
\begin{equation}\label{cik}
  c_{i,k} = \mathbf{b}_i^T\mathbf{D}(2^q,h_k)\mathbf{x}_k,  i=1,2,\ldots,q,\ k=1,2,3.
\end{equation}
Plugging it into \eqref{three-variable-checks-matrix-form}, we can rewrite it as
\begin{equation}\label{sum T-matrix-form-2}
   \begin{split}
 &  \mathbf{P}\mathbf{T}_{i}\mathbf{D}\mathbf{x} \preceq \mathbf{t}, ~ \forall i=1,2,\ldots,q,
  \end{split}
\end{equation}
where
\begin{subequations}\label{t F matrix}
\begin{align}
&\mathbf{T}_{i}={\rm{diag}}(\mathbf{b}^T_{i},\mathbf{b}^T_{i},\mathbf{b}^T_{i}), \label{Ti}\\
&\mathbf{D} ={\rm diag}(\mathbf{D}(2^q,h_1),\mathbf{D}(2^q,h_2),\mathbf{D}(2^q,h_3)),\label{D}\\
&\mathbf{x}=[\mathbf{x}_{1}; \mathbf{x}_{2}; \mathbf{x}_{3}]. \label{x}
\end{align}
\end{subequations}

Moreover, let
\begin{subequations}\label{Ww}
 \begin{align}
    &\mathbf{W}=[\mathbf{P}\mathbf{T}_{1}\mathbf{D};\ldots;\mathbf{P}\mathbf{T}_{q}\mathbf{D}], \label{W}\\
    &\mathbf{w}=\mathbf{1} \otimes\mathbf{t}, \label{w}
 \end{align}
 \end{subequations}
where ``$\mathbf{1}$'' is a length-$q$ all-ones vector.
Then, based on \eqref{sum T-matrix-form-2},
the three-variables parity-check equations \eqref{three-variable pairty-check} can be equivalent to
\begin{equation}\label{Ww}
  \mathbf{Wx}\preceq\mathbf{w}, \ \mathbf{E}\mathbf{x}\preceq \mathbf{1}, \ \mathbf{x}\in\{0,1\}^{3(2^q-1)},
\end{equation}
in the sense their solutions are one-to-one correspondent, where $\mathbf{1}=[1;1;1]$ and $\mathbf{E}={\rm diag}(\mathbf{1}^T, \mathbf{1}^T, \mathbf{1}^T)$ in which ``$\mathbf{1}$'' is a length-$(2^q-1)$ all-ones vector.

\subsection{Equivalent ML decoding problem}

Consider the $j$th check equation in \eqref{C}. Without loss
of generality, we assume it involves $d_i\geq3$ variables, which are denoted by $u_{\sigma_1},\dotsb, u_{\sigma_{d_j}}$ and their corresponding coefficients are $h_{\sigma_1},\dotsb, h_{\sigma_{d_j}}$. Then, the check equation can be decomposed equivalently to $d_j-2$ three-variables check equations by introducing $d_j-3$ auxiliary variables. The detailed decomposing procedure is presented as the following three steps:

Step 1: for the first two variables $u_{\sigma_1}$ and $u_{\sigma_2}$, we introduce an auxiliary variable $v_1$ and let them satisfy the three-variables check equation \eqref{three-variabls check euqation}, i.e.,
\begin{equation}\label{decom-1}
 (h_{\sigma_1}u_{\sigma_1}+h_{\sigma_2}u_{\sigma_2}+g_1)\oplus 2^q=0.
\end{equation}

Step 2: for the variables in set $\{u_{\sigma_{3}},\dotsb,u_{\sigma_{d_j-2}}\}$, we introduce a corresponding auxiliary variables set $\{g_{2},\dotsb, g_{d_j-3}\}$ and let them satisfy the following three-variables check equations
\begin{equation}\label{decom-2}
(g_{t-1}+h_{\sigma_{t+1}}u_{\sigma_{t+1}}+g_{t})\oplus 2^q=0,~~ t=2, \dotsb, d_j-4.
\end{equation}

Step 3: let auxiliary variable $g_{d_i-3}$ and the last two variables $u_{\sigma_{d_j-1}}$ and $u_{\sigma_{d_j}}$ satisfy \eqref{decom-3}
\begin{equation}\label{decom-3}
   (g_{d_i-3}+h_{\sigma_{d_j-1}}u_{\sigma_{d_j-1}}+h_{\sigma_{d_j}}u_{\sigma_{d_j}})\oplus2^q=0.
\end{equation}

Applying the above decomposing procedure to all the check equations in \eqref{C}, one can find the total numbers of the three-variables check equations and the introduced $2^q$-ary auxiliary variables are
{\setlength\abovedisplayskip{1pt}
\setlength\belowdisplayskip{1pt}
\setlength\jot{1pt}
\begin{equation}\label{gamma_a_c}
  \begin{split}
   &\Gamma_{c} = \sum_{j=1}^{m}(d_j-2),  ~~ \Gamma_{a} = \sum_{j=1}^{m}(d_j-3),
  \end{split}
\end{equation} respectively.}

Based on the discussion of the three-variables check equation in the previous subsection, we define
\begin{equation}\label{v}
  \mathbf{v}=[\mathbf{x}; \mathbf{s}],
\end{equation}
  where auxiliary variable $\mathbf{s}=[\mathbf{s}_1; \dotsb; \mathbf{s}_i;\dotsb; \mathbf{s}_{\Gamma_a}]$ and $\mathbf{s}_i\in\{0,1\}^{2^q-1}$ correspond to auxiliary variable $g_{i}$. Define a variable-selecting matrix $\mathbf{Q}_{\tau}\in\{0,1\}^{3\times(n+\Gamma_a)}$ corresponding to the $\tau$th three-variables check equation. Its every row includes only one ``1'', whose index corresponds to the variable in the check equation. It is easy to see that $\mathbf{Q}_{\tau}\mathbf{v}$ are the variables in the $\tau th$ three-variables check equation, where $\tau=1,\dotsb,\Gamma_c$. Moreover, we define
\begin{subequations}\label{Abq}
     \begin{align}
     &\boldsymbol\lambda=[\pmb{\gamma}; \mathbf{0}],  \label{Abq-a}\\
     &\mathbf{F}\!=\![\mathbf{W}_1(\mathbf{Q}_1\!\otimes\!\mathbf{I}); \!\cdots\!;\!\mathbf{W}_{\tau} (\!\mathbf{Q}_\tau\!\otimes\!\mathbf{I})\!;\!\cdots\!;\!\mathbf{W}_{\Gamma_c}(\!\mathbf{Q}_{\Gamma_c}\!\otimes\!\mathbf{I})], \label{Abq-b}\\
     &\mathbf{f}=\mathbf{1} \otimes {\mathbf{w}}.                \label{Abq-d}
    \end{align}
\end{subequations}
where $\mathbf{I}$ is a $(2^q-1) \times (2^q-1)$ identity matrix, symbols ``$\mathbf{1}$'' and ``$\mathbf{0}$'' are length-$\Gamma_c$ all-ones vector and length-$(2^q-1)\Gamma_a$ all-zeros vector, respectively, and {$\mathbf{W}\!_{\tau}=[\mathbf{P}\mathbf{T}_{1}\mathbf{D}\!_{\tau};\ldots;\mathbf{P}\mathbf{T}_{q}\mathbf{D}\!_{\tau}]$ and $\mathbf{D}\!_{\tau}$ have the same expression as \eqref{t F matrix}.}
Another observation on the equivalent binary codeword of the nonbinary symbol is that it includes at most one 1 (see \eqref{Mq2b}). To exploit this structure, we define $(n+\Gamma_a)\times(n+\Gamma_a)(2^q-1)$ matrix $\mathbf{S}={\rm diag}(\mathbf{1}^T,\dotsb,\mathbf{1}^T)$ where ``$\mathbf{1}$'' is a length-$(2^q-1)$ all-ones vector. Then, we have
\begin{equation}\label{Sv}
 \mathbf{Sv}\preceq \mathbf{1}.
\end{equation}
where ``$\mathbf{1}$'' is a length-$n+\Gamma_{a}$ all-ones vector. Then, the ML decoding problem \eqref{ML_x} is equivalent to the following linear integer program
{\setlength\abovedisplayskip{0pt}
 \setlength\belowdisplayskip{0pt}
  \setlength\jot{1pt}
\begin{subequations}\label{ML-decoding}
\begin{align}
&\underset{\mathbf{v}}{\rm min} \hspace{0.35cm} {\boldsymbol\lambda}^{T}\mathbf{v}\\
&\hspace{0.1cm} \rm{s. t.} \hspace{0.25cm} {\mathbf{F}}\mathbf{v} \preceq \mathbf{f},~ \mathbf{Sv}\preceq \mathbf{1},\label{ML-decoding-b}\\
& \hspace{0.9cm} \mathbf{v}\in \{0,1\}^{(2^q-1)(n+\Gamma_a)}. \label{ML-decoding-c}
\end{align}
\end{subequations}}

Due to the binary constraints \eqref{ML-decoding-c}, the above linear integral program \eqref{ML-decoding} is NP-hard, i.e., its computational complexity scales exponentially with the number of variables. Therefore, it is prohibitive to solve problem \eqref{ML-decoding} directly. In the following, we exploit relaxation and tightness techniques to formulate a tractable model.

\subsection{Relaxation and Tightness}

The typical way to handle the binary constraint \eqref{ML-decoding-c} is to relax it to the box constraint $\mathbf{0}\preceq\mathbf{v}\preceq\mathbf{1}$, which can simplify the NP-hard problem \eqref{ML-decoding} to a convex one. However, the resulting optimization problem's optimal solution could be fractional especially when the decoder works in low SNR regions. To overcome this drawback, we deploy the following two techniques to tighten the relaxation.

One is to add a quadratic penalty term into the objective, i.e., $\pmb{\lambda}^{T}\mathbf{v}-\frac{\alpha}{2}\|\mathbf{v}-0.5\|_{2}^{2}$, where $\alpha >0$ is a preset constant. Intuitively, the quadratic penalty
can make the optimal integer solutions more favorable.

The other is to introduce extra linear constraints to cut possible fractional solutions from feasible space. For \eqref{three-variable pairty-check}, defining $\mathbf{c}_k=[c_{1,k}, \dotsb, c_{q,k}]^T$, we can rewrite it as $\bigg(\displaystyle\sum_{k=1}^3\mathbf{c}_k\bigg)\oplus2=\mathbf{0}$.
Letting $\mathbf{B}^T$ (see \eqref{B}) multiply the left side of the above equation, we can obtain
{\setlength\abovedisplayskip{0pt}
 \setlength\belowdisplayskip{1pt}
  \setlength\jot{1pt}
\begin{equation}\label{Bxik}
    \bigg(\mathbf{B}^T\displaystyle\sum_{k=1}^3\mathbf{c}_k\bigg)\oplus2=\mathbf{0},
\end{equation}
where} ``$\mathbf{0}$'' is a $(2^q-1)$-length all-zeros vector. Similar derivations to \eqref{four-inequa}--\eqref{Ww}, we can obtain
\begin{equation}\label{ineq redundant}
\begin{split}
  & \hat{\mathbf{W}}\mathbf{x}\preceq \hat{\mathbf{w}}, \ \  \mathbf{x}\in\{0,1\}^{3(2^q-1)},
\end{split}
\end{equation}
where
{\setlength\abovedisplayskip{0pt}
 \setlength\belowdisplayskip{0pt}
  \setlength\jot{1pt}
\begin{subequations}\label{wW}
  \begin{align}
  &\hat{\mathbf{w}}=\mathbf{1} \otimes \mathbf{t}, \label{wW_w}\\ &\hat{\mathbf{W}}=\left((\mathbf{B}^T\otimes\mathbf{I})\mathbf{W}\right)\oplus 2.\label{wW_W}
  \end{align}
\end{subequations}
where} ``$\mathbf{1}$'' is a length-$2^q-1$ all-ones vector and ``$\mathbf{I}$'' is a $4 \times 4$ identity matrix.
Since $\hat{\mathbf{w}}$ is a $4(2^q-1)$-length vector and $\hat{\mathbf{W}}$ is a $4(2^q-1)$-by-$(2^q-1)$ matrix, one can find that \eqref{ineq redundant} consists of $4(2^q-1)$ inequalities. Moreover, besides $4q$ inequalities in \eqref{Ww}, other inequalities in \eqref{ineq redundant} can be cast as redundant ones since they are combined by some inequalities in \eqref{Ww}. However, when the binary constraint $\x \in\{0,1\}^{3(2^q-1)}$ is relaxed  to the box constraint $\x \in[0,1]^{3(2^q-1)}$, these redundant inequalities can play a role in tightening the relaxation.

Based on the above tightness techniques and through a similar formulation procedure \eqref{decom-1}--\eqref{ML-decoding}, the MP decoding problem \eqref{ML-decoding} can be relaxed to the following optimization model
 {\setlength\abovedisplayskip{1pt}
 \setlength\belowdisplayskip{1pt}
  \setlength\jot{1pt}
\begin{subequations}\label{ML-decoding-all}
\begin{align}
&\underset{\mathbf{v}}{\rm min} \hspace{0.3cm} {\boldsymbol\lambda}^{T}\mathbf{v}-\frac{\alpha}{2}\|\mathbf{v}-0.5\|_{2}^{2},\\
& \hspace{0.1cm} \rm{s.t.} \hspace{0.24cm} {\mathbf{A}}\mathbf{v} \preceq {\mathbf{b}}, \label{ML-decoding-all-b}\\
&\hspace{0.9cm} \mathbf{0}\preceq\mathbf{v}\preceq\mathbf{1}, \label{ML-decoding-all-c}
\end{align}
\end{subequations}
where} symbols ``$\mathbf{1}$'' and ``$\mathbf{0}$'' in \eqref{ML-decoding-all-c} are length-$(2^q-1)(n+\Gamma_a)$ all-ones vector and length-$(2^q-1)(n+\Gamma_a)$ all-zeros vector, respectively, and
\begin{subequations}\label{Ab-construct}
\begin{align}
&\hspace{-7pt}{\mathbf{A}}\!\!=\!\![\hat{\mathbf{W}}_1(\mathbf{Q}_1\!\otimes\!\mathbf{I}); \!\cdots\!;\!\hat{\mathbf{W}}_{\tau} \!(\mathbf{Q}_\tau\!\otimes\!\mathbf{I})\!;\!\cdots\!;\!\hat{\mathbf{W}}_{\Gamma_c}(\mathbf{Q}_{\Gamma_c}\!\otimes\!\mathbf{I})\!;\mathbf{S}],\label{Ab-construct_A} \\
&\hspace{-7pt}{\mathbf{b}}=[\mathbf{1} \otimes \hat{\mathbf{w}}; \mathbf{1}]. \label{Ab-construct_b}
\end{align}
\end{subequations}
where symbols ``$\mathbf{1}$''s in \eqref{Ab-construct_b} are length-$\Gamma_{c}$ and length-$(n+\Gamma_a)$ all-ones vectors, respectively.

In sequel, we will present an efficient solving algorithm via the proximal-ADMM technique for the above optimization problem \eqref{ML-decoding-all}.
Moreover, we also prove that the proposed proximal-ADMM decoding algorithm converges to some stationary point of problem \eqref{ML-decoding-all} in theory. Furthermore, by exploiting inherent structures of problem \eqref{ML-decoding-all}, we show that the computational complexity of the proposed proximal-ADMM algorithm is linear to the length of the LDPC codes.

\section{Proximal-ADMM solving algorithm}\label{admm-qp-section}

\subsection{Proximal-ADMM algorithm framework}

By introducing two auxiliary variables, $\mathbf{e}_{1}$ and $\mathbf{e}_2$, we transform the decoding problem \eqref{ML-decoding-all} to
\begin{subequations}\label{pADMM-frame-problem}
\begin{align}
& \hspace{0.0cm} \underset{\mathbf{v},\mathbf{e}_1,\mathbf{e}_2}{\min} \hspace{0.25cm}  \pmb{\lambda}^{T}\mathbf{v}-\frac{\alpha}{2}\|\mathbf{v}-0.5\|_{2}^{2} \\
&\ \ {\rm s.\ t.} \hspace{0.28cm} \mathbf{A}\mathbf{v}+\mathbf{e}_1 = \mathbf{b}, \ \mathbf{e}_{1} \succeq \mathbf{0}, \label{pADMM-frame-b}\\
& \hspace{1.2cm} \mathbf{v}=\mathbf{e}_{2}, ~~ \mathbf{0}\preceq\mathbf{e}_{2}\preceq\mathbf{1}. \label{pADMM-frame-c}
\end{align}
\end{subequations}
where symbol ``$\mathbf{0}$'' in \eqref{pADMM-frame-b} is a length-$(4(2^q-1)\Gamma_{c}+n+\Gamma_{a})$ all-zeros vector, and symbols ``$\mathbf{1}$'' and ``$\mathbf{0}$'' in \eqref{pADMM-frame-c} are length-$(2^q-1)(n+\Gamma_a)$ all-ones vector and length-$(2^q-1)(n+\Gamma_a)$ all-zeros vector, respectively.
The augmented Lagrangian function for problem \eqref{pADMM-frame-problem} can be written as
\begin{equation}\label{aug-Lagrangian}
\begin{split}
\mathcal{L}_{\mu}\!(\mathbf{ v},\!\mathbf{e}_{1},\!\mathbf{e}_{2},\!\mathbf{y}_{1},\!\mathbf{y}_{2}) \!\!= \pmb{\lambda}^{T}\!\mathbf{v}\!\!-\!\!\frac{\alpha}{2}\|\mathbf{v}\!\!-\!\!0.5\|_{2}^{2} \!+\!\mathbf{y}_{1}^{T}\!(\mathbf{A}\mathbf{v} \!\!+\!\!\mathbf{e}_{1}\! \!-\!\! \mathbf{b})\! \\
+ \mathbf{y}_{2}^{T}(\mathbf{v}-\mathbf{e}_{2})+ \frac{\mu}{2}\|\mathbf{A}\mathbf{v}  +\mathbf{e}_{1}- \mathbf{b}\|_{2}^{2} + \frac{\mu}{2} \|\mathbf{v}-\mathbf{e}_{2}\|_{2}^{2},
\end{split}
\end{equation}
where $\mathbf{y}_{1}$ and $\mathbf{y}_{2}$ are Lagrangian multipliers corresponding to the two equality constraints in \eqref{pADMM-frame-problem} respectively and $\mu>0$ is a penalty parameter.
Based on \eqref{aug-Lagrangian}, the proximal-ADMM iteration algorithm for solving \eqref{pADMM-frame-problem} can be described as follows
\begin{subequations}\label{proximal-ADMM update_LP}
\begin{align}
& \v^{k+1} \!=\! \mathop{\arg \min}\limits_{\v} \mathcal{L}_{\mu}(\v,\e_{1}^{k},\e_{2}^{k},\y_{1}^{k},\y_{2}^{k})+\frac{\rho}{2}\|\v\!-\!\mathbf{p}^{k}\|_2^2, \label{proximal-x-update}  \\
& \e_{1}^{k+1}\!\! =\! \mathop{\arg \min}\limits_{\e_{1}\succeq\mathbf{0}} \mathcal{L}_{\mu}(\!\v^{k\!+\!1}\!,\e_{1},\!\e_{2}^{k},\y_{1}^{k},\y_{2}^{k})\!\!+\!\!\frac{\rho}{2}\|\e_1\!-\!\z_1^{k}\|_2^2, \label{proximal-v1-update} \\
& \e_{2}^{k+1} \!\!=\! \mathop{\arg \min}\limits_{\mathbf{0}\preceq\mathbf{e}_{2}\preceq\mathbf{1}} \mathcal{L}_{\mu}(\!\v^{k\!+\!1},\e_{1}^{k},\e_{2},\y_{1}^{k},\y_{2}^{k})\!\!+\!\!\frac{\rho}{2}\|\e_2\!\!-\!\!\z_2^{k}\|_2^2, \label{proximal-v2-update} \\
& \mathbf{p}^{k+1}=\mathbf{p}^{k}+\beta(\v^{k+1}-\mathbf{p}^{k}), \nonumber\\
& \z_{1}^{k+1}=\z_{1}^{k}+\beta(\e_{1}^{k+1}-\z_{1}^{k}), \label{proximal-z-update}\\
& \z_{2}^{k+1}=\z_{2}^{k}+\beta(\e_{2}^{k+1}-\z_{2}^{k}), \nonumber\\
& \y_1^{k+1} = \y_1^{k} + \mu(\A\v^{k+1}+ \e_{1}^{k+1}- \mathbf{b}), \nonumber\\
& \y_2^{k+1} = \y_2^{k} + \mu(\v^{k+1}-\e_{2}^{k+1}), \label{proximal-lamda-update}
\end{align}
\end{subequations}
where $k$ is the iteration number, $\|\v-\mathbf{p}^{k}\|_2^2$, $\|\e_1\!-\!\z_1^{k}\|_2^2$, and $\|\e_2-\z_2^{k}\|_2^2$ are so-called  proximal terms, $\rho$ is the corresponding penalty parameter, and $\beta$ belongs to $(0,1]$. In the following, we show that subproblems \eqref{proximal-x-update}-\eqref{proximal-v2-update} can be solved efficiently by exploiting their inherent structures.

\subsection{Solving subproblem \eqref{proximal-x-update}}

Obviously, choosing their values of parameters $\rho$, $\mu$, and $\alpha$ properly, problem \eqref{proximal-x-update} can be reduced to a strongly quadratic convex one with respect to $\mathbf{v}$.
In this case, its solving procedure can be described as follows:
by setting the gradient of the function $\mathcal{L}_{\mu}(\v,\u_{1}^{k},\u_{2}^{k},\y_{1}^{k},\y_{2}^{k})+\frac{\rho}{2}\|\v-\r^{k}\|_2^2$ to be zero and solving the corresponding linear equation,
we update $\mathbf{v}$ as
\begin{equation}\label{lp-x-update-solution}
\begin{split}
& \hspace{0.0cm} \mathbf{v}^{k+1} = \big(\A^{T}\A+\epsilon\I\big)^{-1}\pmb{\varphi}^{k}\\
\end{split}
\end{equation}
where ``$\mathbf{I}$'' is a $(2^q-1)(n+\Gamma_{a}) \times (2^q-1)(n+\Gamma_{a})$ identity matrix, and
\begin{equation}\label{epsilon_varphi}
\begin{split}
&\epsilon=1+\frac{\rho}{\mu}-\frac{\alpha}{\mu}, \\
&\pmb{\varphi}^{k}\!\!=\!\! \A^{T}\!\Big(\b\!-\!\e_{1}^{k}\!-\!\frac{\y_1^{k}}{\mu}\!\Big)\!
+\!\big(\!\e_{2}^{k}\!-\!\frac{\y_2^{k}}{\mu}\!\big)\!+\!\frac{\rho}{\mu}\mathbf{p}^{k}\!-\!\frac{\pmb{\lambda}\!+\!0.5\alpha}{\mu}.
\end{split}
\end{equation}

Note that $\big(\A^{T}\A+\epsilon\I\big)^{-1}$ is fixed for a given code. Thus, it only needs to be calculated only once throughout the ADMM iterations. Therefore, the main computational cost lies in $\big(\A^{T}\A+\epsilon\I\big)^{-1}\pmb{\varphi}^{k}$. It is easy to see that calculating it directly requires $\mathcal{O}\big((2^q-1)^2(n+\Gamma_{a})^2\big)$ complexity,
which is prohibitive for large-scale problems. In the following, we show a much more efficient way to perform the computational procedure. Firstly, we present the following lemma.
\begin{lemma}\label{w-inverse-lemma}
Matrix $\big(\A^{T}\A+\epsilon\I\big)^{-1}$ is a block diagonal. Specifically, it can be denoted by
\begin{equation}\label{w-I-inverse-defi}
\begin{split}
& \big(\A^{T}\A+\epsilon\I\big)^{-1} \\
= & {\rm diag}\Big(\!\big({\A}_1^{T}\A_1\!\!+\!\epsilon\I\big)\!^{-1}\!,\ldots,\!\big(\A_{n+\Gamma_{a}}^{T}\A_{n+\Gamma_{a}}\!\!+\!\epsilon\I\big)\!^{-1}\!\Big),
\end{split}
\end{equation}
where sub-matrix $\mathbf{A}_{i}$, $i=1,\ldots, n+\Gamma_{a}$, is formed by  column vectors indexed from $\big((2^q-1)(i-1)+1\big)$ to $(2^q-1)i$ in matrix $\mathbf{A}$ and ``$\mathbf{I}$'' is a $(2^q-1) \times (2^q-1)$ identity matrix. Moreover,
\begin{equation}\label{dv-phi-I-inverse1}
\big(\A_i^{T}\A_i+\epsilon\I\big)^{-1} =\left[
  \begin{array}{cccc}
    \theta_{i} & \omega_{i} & \cdots & \omega_{i} \\
    \omega_{i} & \theta_{i} & \cdots & \omega_{i} \\
    \vdots & \vdots & \vdots & \vdots \\
    \omega_{i} & \omega_{i} & \cdots & \theta_{i} \\
  \end{array}
\right],
\end{equation}
where
\begin{equation}\label{omega_theta}
\begin{split}
  &\theta_{i}\!=\!\omega_{i}+\frac{1}{2^{m}d_{i}+\epsilon}, \\
  &\displaystyle\omega_{i}\!=\!\frac{\!-\!2^{m}d_{i}}{(2^{m}d_{i}\!+\!\epsilon)[(2^{m+1}d_{i}\!+\!\epsilon\!+\!1)\!+\!(2^{m}d_{i}\!+\!1)(2^{m}\!-\!2)]}.
  \end{split}
\end{equation}
Here, $d_{i}$ denotes the degree of the $i$th information symbol, i.e., the number of check equations which it participates in.
\end{lemma}

\begin{proof}
See Appendix \ref{w-inverse-lemma-proof}.
\end{proof}

\eqref{w-I-inverse-defi} indicates that $\mathbf{v}^{k+1}$ in \eqref{lp-x-update-solution} can be obtained through the following implementations in parallel
\begin{equation}\label{proximal-xi-calcu}
\begin{split}
& \hspace{0.0cm} \v_{i}^{k+1} = \big(\A_i^{T}\A_i+\epsilon\I\big)^{-1}\pmb{\varphi}_i^{k},\  i=1,2,\dotsb,n+\Gamma_a,
\end{split}
\end{equation}
where $\pmb{\varphi}_i^{k}$ is the corresponding $(2^q-1)$-length sub-vector in $\pmb{\varphi}^{k}$. Moreover, based on \eqref{dv-phi-I-inverse1}, \eqref{proximal-xi-calcu} can be further derived as follows
{\setlength\abovedisplayskip{2pt}
 \setlength\belowdisplayskip{2pt}
  \setlength\jot{2pt}
\begin{equation}\label{proximal-xi-calcu-simp}
\begin{split}
 \mathbf{v}_{i}^{k+1}
& =  \left[\!\!\!\!
  \begin{array}{cccc}
    \theta_{i}\!\!-\!\!\omega_i \!\!\!& 0 &\! \cdots \!&\!\!\! 0 \\
    0\!\!\!&\!\!\! \theta_{i}\!\!-\!\!\omega_i \!\!\!&\! \cdots \!&\! \!\!0 \\
    \vdots \!&\! \vdots \!&\! \vdots \!&\!\!\! \vdots \\
    0\!\!\!& 0 &\! \cdots \!&\!\!\! \theta_{i}\!\!-\!\!\omega_i \\
  \end{array}
\!\!\!\!\right]\!\!\pmb{\varphi}_i^{k} \!+\!\!
 \left[\!\!\!
  \begin{array}{cccc}
    \omega_{i}\!\! &\!\! \omega_{i} \!\!&\!\! \cdots \!\!&\!\! \omega_{i} \\
    \omega_{i}\!\! &\!\! \omega_{i} \!\!&\!\! \cdots \!\!&\!\! \omega_{i} \\
    \vdots\!\! &\!\! \vdots \!\!&\!\! \vdots \!\!&\!\! \vdots \\
    \omega_{i}\!\! &\!\! \omega_{i} \!\!&\!\! \cdots \!\!&\!\! \omega_{i} \\
  \end{array}
\!\!\!\right]\!\!\pmb{\varphi}_i^{k} \\
&=(\theta_{i}-\omega_{i})\pmb{\varphi}_i^{k}+\omega_{i}\sum_{\iota=1}^{q-1}\varphi_{i,\iota}^{k},
\end{split}
\end{equation}}
where $\varphi_{i,\iota}^{k}$ denotes the $\iota th$ entry in $\pmb{\varphi}_i^{k}$.

\subsection{Solving subproblems \eqref{proximal-v1-update} and \eqref{proximal-v2-update}}

Solving \eqref{proximal-v1-update} is equivalent to solving the following $M$ subproblems in parallel
\begin{equation}\label{proximal-v1-min}
\begin{split}
& \hspace{0cm} \mathop {\min }\limits_{e_{1,j}} \hspace{0.1cm} y_{1, j}^{k} e_{1, j}
\!+\!\frac{\mu}{2}\!\left(\mathbf{a}_{j}^{T} \mathbf{v}^{k+1}\!\!+\!e_{1, j}\!-\!b_{j}\right)^{2}\!\!\!+\!\frac{\rho}{2}\!\left(e_{1, j}\!\!-\!\!z_{1, j}^{k}\right)\!^{2}  \\
& \hspace{0.15cm} \textrm{s.t.} \hspace{0.3cm} e_{1, j} \geq 0, \ j=1,\dotsb,M.
\end{split}
\end{equation}
where $M=4(2^q-1)\Gamma_{c}+n+\Gamma_{a}$, $\mathbf{a}_{j}^{T}$ denotes the $jth$ row vector of matrix $\A$.
Obviously, the optimal solution of the above problem \eqref{proximal-v1-min}
can be obtained by setting the gradient of its objective function to zero and then projecting the solution
of the corresponding equation to region $[0,+\infty]$. Then, we can obtain
\begin{equation}\label{v1-solution-component}
e_{1,j}^{k+1} = \underset{[0,+\infty]}\Pi \frac{\mu}{\rho+\mu} \big(b_{j}-\mathbf{a}_{j}^{T}\mathbf{v}^{k+1}-\frac{y_{1,j}^{k}}{\mu}+\frac{\rho}{\mu}z_{1,j}^{k}\big).
\end{equation}

Similar to \eqref{proximal-v1-update}, problem \eqref{proximal-v2-update} can be separated into the following $N$ independent subproblems
\begin{equation}\label{v2-min}
\begin{split}
& \hspace{0cm} \mathop {\min }\limits_{e_{2,\ell}} \hspace{0.1cm} -\mathrm{y}_{2, \ell}^{k} e_{2, \ell}+\frac{\mu}{2}\left(v_{\ell}^{k}-e_{2, \ell}\right)^{2}
+\frac{\rho}{2}\left(e_{2, \ell}-z_{2, \ell}^{k}\right)^{2},  \\
& \hspace{0.15cm} \textrm{s.t.} \hspace{0.3cm} 0 \leq e_{2, \ell} \leq 1, \ \ell=1,\dotsb,N.
\end{split}
\end{equation}
where $N=(2^q-1)(n+\Gamma_{a})$. Their optimal solutions can be expressed as
\begin{equation}\label{v2-solution-component}
e_{2,\ell}^{k+1} = \underset{[0,1]}\Pi \frac{\mu}{\rho+\mu}\big(v_{\ell}^{k+1}+\frac{y_{2,\ell}^{k}}{\mu}+\frac{\rho}{\mu}z_{2,\ell}^{k}\big).
\end{equation}

In \emph{Algorithm \ref{proxiaml-ADMM-QP}}, we summarize the proposed proximal-ADMM decoding algorithm for nonbinary LDPC codes in $\mathbb{F}_{2^{m}}$. In the next section, we will discuss its convergence and computational complexity.

\begin{algorithm}[t]
\caption{Proximal-ADMM decoding algorithm}
\label{proxiaml-ADMM-QP}
\begin{algorithmic}[1]
\STATE Initializations: decompose each check equation into three-variables check equations based on the parity-check matrix $\mathbf{H}$.
       Then, construct matrix $\mathbf{A}$ and vector $\mathbf{b}$ based on \eqref{Ab-construct}. Let $\mu>0$, $\rho>\alpha$ and $0 <\beta \leq1$. For all $i \in \{1,\dotsb, n+\Gamma_{a}\}$, compute $\theta_{i}$ and $\omega_{i}$ via \eqref{omega_theta}. Initialize variables $\{\mathbf{v}, \mathbf{e}_1, \mathbf{e}_2, \mathbf{p}, \mathbf{z}_1, \mathbf{z}_2, \mathbf{y}_1, \mathbf{y}_2\}$ as all-zeros vectors. \\
\STATE  \textbf{Repeat}  \\
\STATE \hspace{0.3cm} Compute $\pmb{\varphi}^{k}$ via \eqref{epsilon_varphi}. Then, update $\v_i^{k+1}$, $i\in \{1,\ldots,n+\Gamma_{a}\}$, in parallel. \\
\STATE \hspace{0.3cm} Update $e_{1,j}^{k+1}$, $j\in\{1,\ldots,M\}$, via \eqref{v1-solution-component} in parallel.
\STATE \hspace{0.3cm} Update $e_{2,\ell}^{k+1}$, $\ell\in\{1,\ldots,N\}$, via \eqref{v2-solution-component} in parallel.
\STATE \hspace{0.3cm} Update $\mathbf{p}^{k+1}$, $\mathbf{z}_{1}^{k+1}$, $\mathbf{z}_{2}^{k+1}$, $\y_1^{k+1}$ and $\y_2^{k+1}$ in parallel via \eqref{proximal-z-update} and \eqref{proximal-lamda-update} respectively.
\STATE \hspace{0.2cm} $k \leftarrow k+1$. \\
\STATE \textbf{Until} some preset conditions are satisfied. \\
\end{algorithmic}
\end{algorithm}

\section{Performance Analysis}\label{Analysis-admm-qp-decoding}

\subsection{Convergence}

Before presenting the main convergence result, we have the following lemma to show that the gradient of the considered augmented Lagrangian $\mathcal{L}_{\mu}(\cdot)$ is Lipschitz continuous.

\begin{lemma}\label{x-Lipschitz-constant}
Suppose $\alpha>0$ and $\mu>0$ and let $X:=\{\v|\A\v \preceq \b, \mathbf{0} \preceq \v \preceq \mathbf{1}\}$. Then, the gradient of the augmented Lagrangian $\mathcal{L}_{\mu}$  with respect to variable $\v$ is Lipschitz continuous,
i.e., for any $\v,\v^{\prime}\in X$,
\begin{equation}\label{x-Lipschitz-defi}
\begin{split}
& \|\nabla_{\v} \mathcal{L}_{\mu}(\v,\! \e_{1}, \! \e_{2},\! \y_{1}, \! \y_{2})\!-\!\nabla_{\v} \mathcal{L}_{\mu}(\v^{\prime},\!\e_{1},\!\e_{2},\!\y_{1},\!\y_{2})\|_{2} \\
\leq &L\|\v-\v^{\prime}\|_{2},
\end{split}
\end{equation}
where $L\geq \alpha+\mu+\mu\delta_{\A}^2$ and ``$\delta_{\A}$'' denotes the spectral norm of matrix $\A$.
\end{lemma}

\begin{proof}
See Appendix \ref{Lipschitz-continuous}.
\end{proof}

Based on Lemma \ref{x-Lipschitz-constant}, we have the following theorem to characterize the convergence property of the proposed proximal-ADMM decoding algorithm.
\begin{theorem}\label{converge-proof-theorem}
Assume that $\rho>\alpha>0$ and
$\frac{\alpha}{\lambda_{\min}(\A^T\A)}\leq\mu\leq \frac{\rho(\rho-\alpha)^{2}}{4\delta_{\A\I}^2(\rho+L+2)^2-(\rho-\alpha)^2}$,
where $\delta_{\A\I}$ is the spectral norm of matrix
$\left[\!\!\!\!\begin{array}{ccc}{\A} \!\!&\!\! {\I_{M}} \!\!\!&\!\!\! {\mathbf{0}} \\ {~\I_{N}} \!\!&\!\! {\mathbf{0}} \!\!\!&\!\!\! {-\I_{N}}\end{array}\!\!\!\!\right]$ and symbol ``$\mathbf{I}_{\kappa}$'' denotes  a $\kappa$-by-$\kappa$ identity matrix.
Let $\{\v^{k},\e_{1}^{k},\e_{2}^{k},\mathbf{p}^{k},\z_{1}^{k},\z_{2}^{k},\y_1^{k},\y_2^{k}\}$ be the tuples generated by \emph{Algorithm \ref{proxiaml-ADMM-QP}}.
Then, we have the following convergence results
\begin{equation}\label{converge-limit}
\begin{split}
& \mathop {\lim }\limits_{k\rightarrow +\infty} \v^{k} = \v^{*},~\mathop {\lim }\limits_{k\rightarrow +\infty}\e_{1}^{k} = \e_{1}^{*},
  ~\mathop {\lim }\limits_{k\rightarrow +\infty} \e_{2}^{k} = \e_{2}^{*}, \\
& \mathop {\lim }\limits_{k\rightarrow +\infty} \p^{k} = \p^{*},~\mathop {\lim }\limits_{k\rightarrow +\infty}\z_{1}^{k} = \z_{1}^{*},
  ~~\mathop {\lim }\limits_{k\rightarrow +\infty} \z_{2}^{k} = \z_{2}^{*}, \\
& \mathop {\lim }\limits_{k\rightarrow +\infty} \y_{1}^{k} = \y_{1}^{*},\mathop {\lim }\limits_{k\rightarrow +\infty} \y_{2}^{k} = \y_{2}^{*},~~ \A\v^{*}\!+\!\e_1^{*}\!-\!\b \!=\!\mathbf{0} \\
&\ \ \v^{*}=\e_2^{*},~~ \v^{*}=\p^{*}, ~~ \e_1^{*}=\z_1^{*}, ~~ \e_2^{*}=\z_2^{*}.
\end{split}
\end{equation}
Moreover, $\v^{*}$ is some stationary point of the original problem \eqref{ML-decoding-all}, i.e.,
\begin{equation}\label{stationary-point}
(\v-\v^{*})^T\nabla_{\v} g(\v^*) \geq 0, ~~~\forall \v \in X,
\end{equation}
where $g(\v)=\pmb{\lambda}^{T}\v-\frac{\alpha}{2}\|\v-0.5\|_2^2$.
\end{theorem}

\begin{proof}
 See Appendix \ref{converge-proof}.
\end{proof}

\subsection{Computational Complexity}

Before analyzing the complexity of \emph{Algorithm \ref{proxiaml-ADMM-QP}}, we show matrix $\A$ has the following property.

\begin{fact}\label{W-property}
 The elements in matrix $\A$ are 0, 1 or -1.
\end{fact}

\begin{proof}
See Appendix \ref{W-property-proof}.
\end{proof}

Based on the above fact of matrix $\mathbf{A}$,
we can see that all multiplications with regard to $\mathbf{A}$, such as $\mathbf{A}^{T}\mathbf{b}$, can be performed via additions.
Moreover, $(\pmb{\lambda}+0.5\alpha)/\mu$, $\theta_{i}$, $\omega_{i}$ and $\frac{\mu}{\rho+\mu}$ can be calculated in advance before we start the ADMM iterations.
Observing \eqref{proximal-xi-calcu-simp}, we can find that computing each $\v_{i}^{k+1}$ requires $2^{q+1}-1$ multiplications.
This implies that the complexity of the $\v^{k+1}$-update is $\mathcal{O}((n+\Gamma_{a})2^q)$.
The properties of matrix $\A$ can also be applied to updating $\e_{1}^{k+1}$.
From \eqref{v1-solution-component}, we observe that each $e_{1,j}^{k+1}$ can be updated only via two multiplication operations and thus the complexity of the $\e_{1}^{k+1}$-update is $\mathcal{O}(M)$.
Similarly, observing \eqref{v2-solution-component}, we can find that computing $e_{2,\ell}^{k+1}$ requires two multiplications.
As a result, the $\e_{2}^{k+1}$-update has $\mathcal{O}(N)$ complexity.
From \eqref{proximal-z-update}, we easily observe that the update of $\mathbf{p}^{k+1}$, $\z_1^{k+1}$ and $\z_2^{k+1}$ requires $N$, $M$ and $N$ multiplications respectively.
Hence, the complexities of computing $\mathbf{p}^{k+1}$, $\z_1^{k+1}$ and $\z_2^{k+1}$ are $\mathcal{O}(N)$, $\mathcal{O}(M)$ and $\mathcal{O}(N)$ respectively.
In addition, observing variables $\y_1$ and $\y_2$ in \eqref{proximal-xi-calcu-simp} \eqref{v1-solution-component} \eqref{v2-solution-component},
one can find that if their scaled forms $\frac{\y_1}{\mu}$ and $\frac{\y_2}{\mu}$ are updated, then corresponding multiplications are not necessary,
i.e., calculating $\frac{\y_1^{k+1}}{\mu}$ and $\frac{\y_2^{k+1}}{\mu}$ only requires some addition operations.
From the above analysis, we can see that the overall computational complexity of \emph{Algorithm \ref{proxiaml-ADMM-QP}} in each iteration is roughly $\mathcal{O}\big((n+\Gamma_{a})2^q+3N+2M\big)$.
Furthermore, based on $M=4(2^q-1)\Gamma_{c}+n+\Gamma_{a}$ and $N=(2^q-1)(n+\Gamma_{a})$, we obtain
\begin{equation*}
\begin{split}
(n\!+\!\Gamma_{a})2^q\!+\!3N\!+\!2M
\!=\!(2^{q+2}\!-\!1)(n\!+\!\Gamma_{a})\!+\!8(2^q\!-\!1)\Gamma_{c}.
\end{split}
\end{equation*}
Moreover, observing \eqref{gamma_a_c}, we have $\Gamma_{a}\leq m(d-3)=n(1-R)(d-3)$ and $\Gamma_{c}\leq m(d-2)=n(1-R)(d-2)$ where $R$ denotes the code rate and $d$ is the largest check node degree.
This implies that $\Gamma_{a}$ and $\Gamma_{c}$ is proportional to the code length $n$ since $d \ll n$ in the case of the LDPC code.
Therefore, we can conclude that the total computational complexity of \emph{Algorithm \ref{proxiaml-ADMM-QP}} in every proximal-ADMM iteration is roughly $\mathcal{O}\big(2^qn\big)$.

\section{Simulation results}\label{simulation-result}
In this section, several numerical results are presented for the proposed proximal-ADMM decoder in \emph{Algorithm \ref{proxiaml-ADMM-QP}}. First, we show its error-correction performance (frame error rate (FER) and symbol error rate (SER)) and decoding efficiency, which is compared with several state-of-the-art nonbinary LDPC decoders. Second, we present how to select the proper value of the parameters in the proximal-ADMM decoder, which can improve the proximal-ADMM decoder's error-correction performance and convergence rate.

\subsection{Performance of the proposed proximal-ADMM decoder}
We consider two codes, named $\mathcal{C}_{1}$ and $\mathcal{C}_{2}$, which are  Tanner (1055,424) LDPC code and Tanner [155,64] LDPC code respectively \cite{Tanner-code}.
For $\mathcal{C}_{1}$, we use the same parity-check matrix as the binary case but each binary nonzero check value is replaced by $1 \in \mathbb{F}_{4}$.
For $\mathcal{C}_{2}$, its parity-check matrix is the same as the binary case with each nonzero check value $1 \in \mathbb{F}_{16}$. The corresponding modulations are quadrature phase shift keying (QPSK) and sixteen quadrature amplitude modulation (16QAM) respectively.
The modulated symbols are transmitted over the additive white Gaussian noise (AWGN) channel.
The considered decoders include the proposed proximal-ADMM algorithm,  logarithm-domain fast-fourier-transforms-based Q-ary sum-product algorithm (Log-FFT-QSPA) \cite{burst-error}, and the nonbinary ADMM-based LP  (ADMM-LP) decoding algorithm \cite{Liu-nonbinary-journal}.
The parameters of \emph{Algorithm \ref{proxiaml-ADMM-QP}} are set as follows:
penalty parameter $\mu$ is chosen as 0.8 and 0.6 for codes $\mathcal{C}_{1}$ and $\mathcal{C}_{2}$ respectively; parameters $\alpha$, $\rho$ and $\beta$ are set to be 0.5, 0.52 and 0.9 respectively for both of the two codes; we stop the iteration when both $\|\mathbf{A}\mathbf{v}^{k}+\mathbf{e}_{1}^{k}-\mathbf{b}\|_{2}^{2}\leq 10^{-5}$ and $\|\mathbf{v}^{k}-\mathbf{e}_{2}^{k}\|_{2}^{2}\leq 10^{-5}$ are satisfied, or the maximum iteration number $t_{\rm max}=500$ is reached.

 \begin{figure}[tp]
    \subfigure[Tanner (1055,424) code $\mathcal{C}_{1}$ in $\mathbb{F}_4$ with QPSK modulation.]{
    \begin{minipage}{8.5cm}
    \centering
        \includegraphics[width=3.5in,height=2.5in]{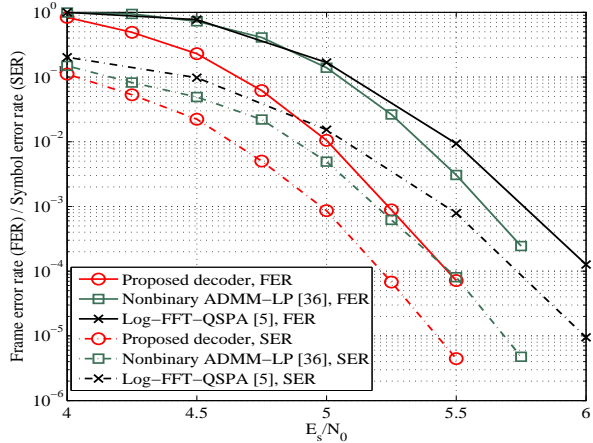}
            \label{fer1055}
    \end{minipage}%
    }\
    \subfigure[Tanner (155,64) code $\mathcal{C}_{2}$ in $\mathbb{F}_{16}$ with 16QAM modulation.]{
    \begin{minipage}{8.5cm}
    \centering
        \includegraphics[width=3.5in,height=2.5in]{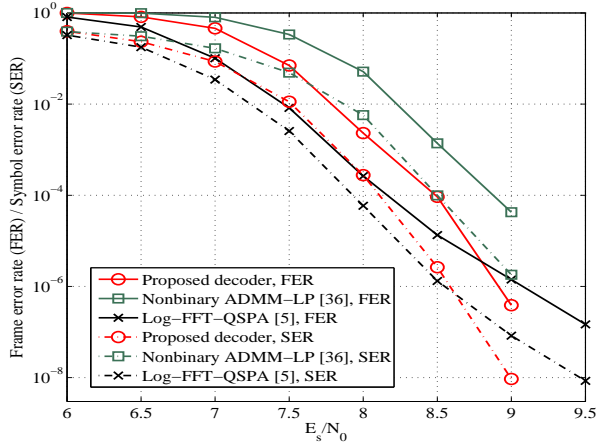}
            \label{fer155}
    \end{minipage}%
    }
     \centering
    \caption{Comparisons of FER/SER performance for two nonbinary LDPC codes in different Galois fields with different modulations, where $\mathcal{C}_{1}$ and $\mathcal{C}_{2}$ denote the Tanner (1055,424) code and the Tanner (155,64) code from \cite{Tanner-code}, respectively.}
    \label{fer_ser}
 \end{figure}

 \begin{figure*}[tp]
    \subfigure[Tanner (1055,424) code $\mathcal{C}_{1}$ in $\mathbb{F}_4$ with QPSK modulation.]{
    \begin{minipage}{8.5cm}
    \centering
        \includegraphics[width=3.5in,height=2.5in]{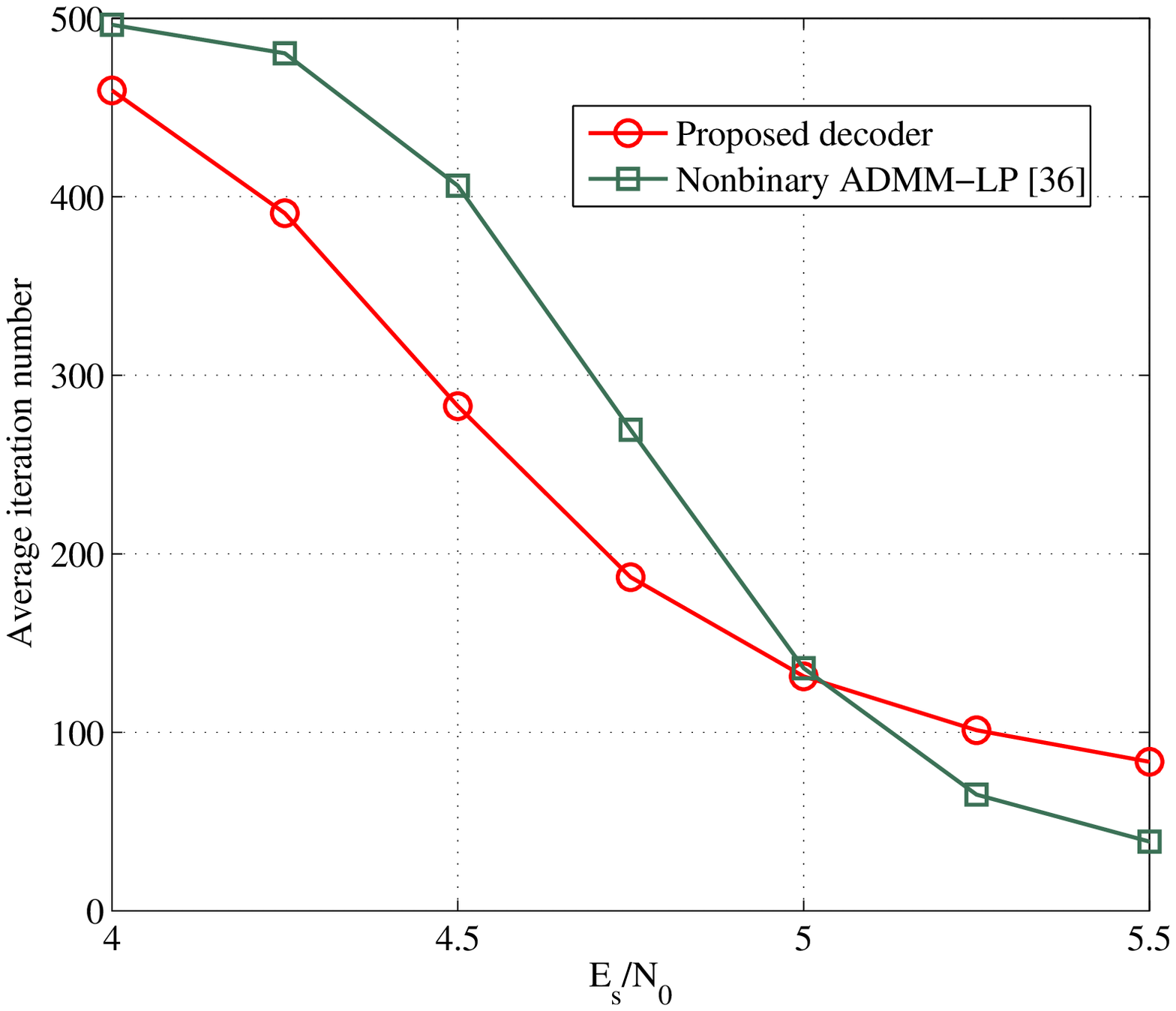}
            \label{iter504}
    \end{minipage}%
    }\
    \subfigure[Tanner (155,64) code $\mathcal{C}_{2}$ in $\mathbb{F}_{16}$ with 16QAM modulation.]{
    \begin{minipage}{8.5cm}
    \centering
        \includegraphics[width=3.5in,height=2.5in]{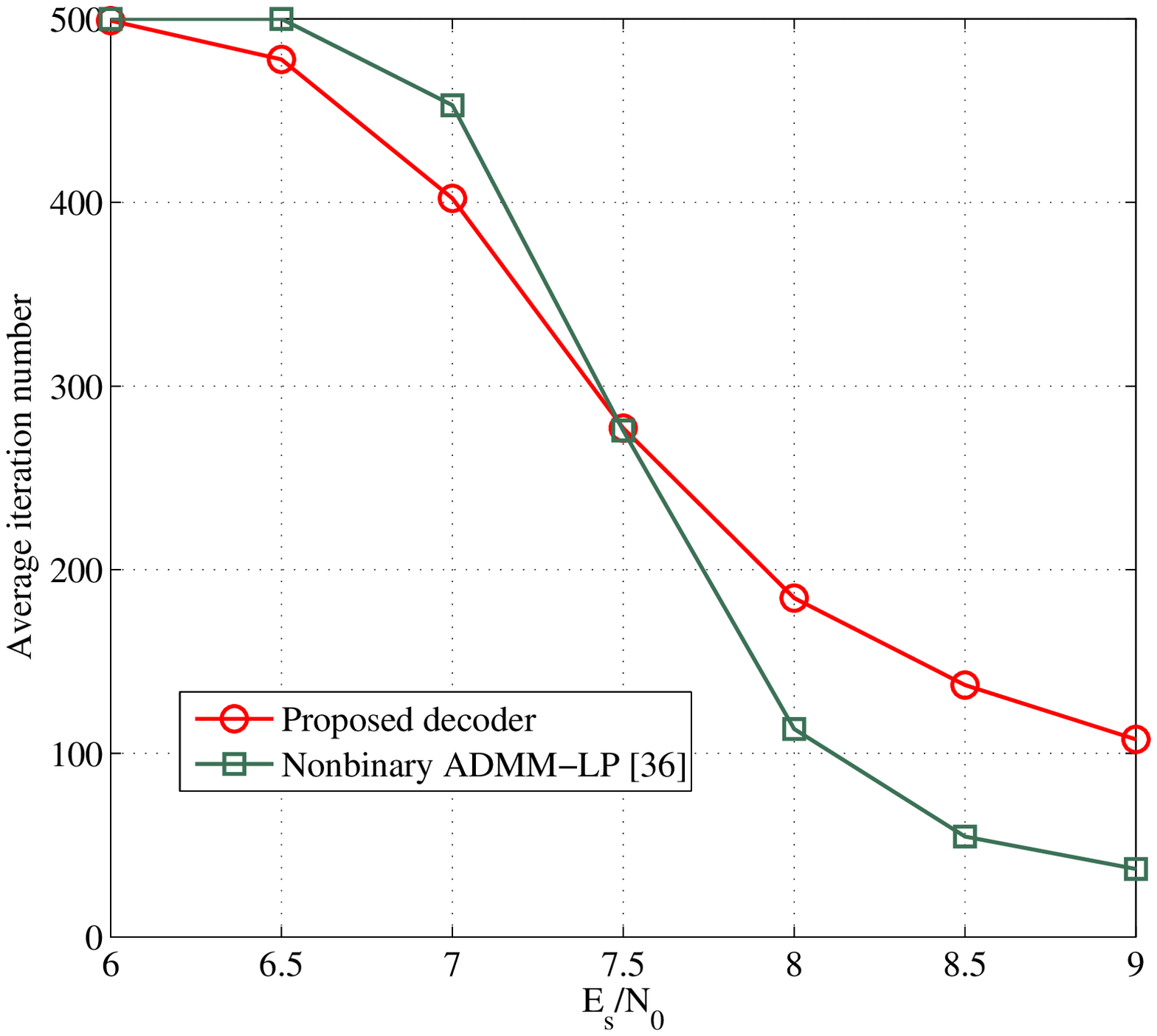}
            \label{iter155}
    \end{minipage}%
    }
     \centering
    \caption{Comparisons of the average number of iterations for two nonbinary LDPC codes in different Galois fields with different modulations, where $\mathcal{C}_{1}$ and $\mathcal{C}_{2}$ denote the Tanner (1055,424) code and the Tanner (155,64) code from \cite{Tanner-code}, respectively.}
    \label{iter2}
 \end{figure*}
Fig.\,\ref{fer_ser} shows the FER/SER performance of code $\mathcal{C}_{1}$ in $\mathbb{F}_4$ and code $\mathcal{C}_{2}$ in $\mathbb{F}_{16}$ when different decoders are applied.
From Fig.\,\ref{fer1055}, it can be seen that our proposed decoder attains better FER/SER performance than Log-FFT-QSPA \cite{FHT-nonBP} and the nonbinary ADMM-LP decoder \cite{Liu-nonbinary-journal}.
Observing Fig.\,\ref{fer155}, it can be seen that the proposed decoder outperforms the ADMM-LP decoder \cite{Liu-nonbinary-journal} in terms of either FER or SER.
We can also see that the proposed decoder displays comparable FER/SER performance to the Log-FFT-QSPA in low SNR regions and performs superiorly to the Log-FFT-QSPA in terms of FER and BER at $E_s/N_0$ =9dB, where the FER/SER curves of the proposed proximal-ADMM decoder continue to drop in a waterfall manner while the FER/SER curves of the Log-FFT-QSPA drop slowly and the corresponding error-correction performance is inferior to the proximal-ADMM decoder.
Therefore, we can conclude that the proposed proximal-ADMM decoder outperforms Log-FFT-QSPA \cite{burst-error} and the ADMM-LP decoder \cite{Liu-nonbinary-journal} in terms of error-correction performance for the considered nonbinary LDPC codes.

Fig.\,\ref{iter2} shows the average iteration number of the proximal-ADMM decoder and the ADMM-LP decoder in \cite{Liu-nonbinary-journal}. From the figures, one can find that the proposed proximal-ADMM decoder requires fewer iterations than the ADMM-LP decoder \cite{Liu-nonbinary-journal} in low SNR regions and they are comparable in high SNR regions.
Moreover, it can also be observed that the average number of iterations required by the proposed proximal-ADMM decoder is less than 200 in high SNR regions for both of the codes $\mathcal{C}_{1}$ and $\mathcal{C}_{2}$.
Furthermore, compared with the competing ADMM-LP decoder \cite{Liu-nonbinary-journal}, since the expensive Euclidean projections onto the check/simplex polytopes are replaced by the simple Euclidean projections onto the positive quadrant, the proposed proximal-ADMM decoder outperforms the nonbinary ADMM-LP decoder \cite{Liu-nonbinary-journal} in terms of decoding efficiency.

\begin{figure}[tp]
  \centering
  \centerline{\psfig{figure=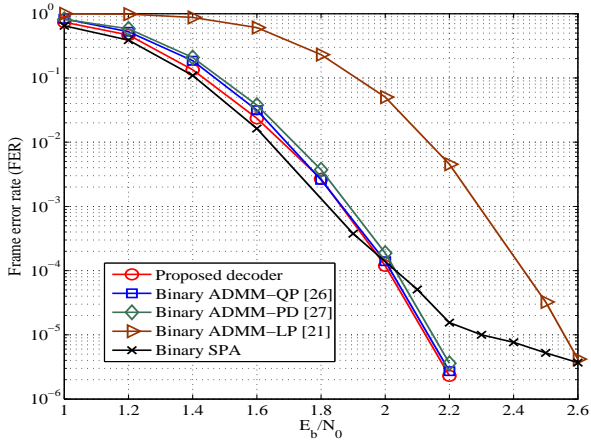,width=8.8cm,height=6.3cm}}
  \caption{FER comparison for the [2640,1320] ``Margulis'' binary LDPC code $\mathcal{C}_{3}$ from \cite{Mackay-code}.}
  \label{fer2640}
\end{figure}
\begin{figure}[tp]
  \centering
  \centerline{\psfig{figure=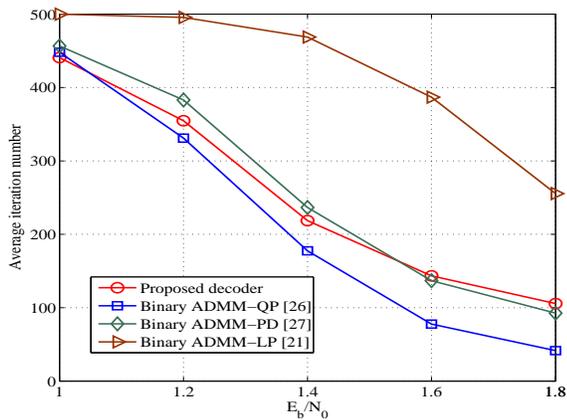,width=8.5cm,height=6cm}}
  \caption{Comparison of the average number of iterations for the [2640,1320] ``Margulis'' binary LDPC code $\mathcal{C}_{3}$ from \cite{Mackay-code}.}
  \label{iter2640}
\end{figure}

In addition, it should be noted that the proposed proximal-ADMM decoder is also suitable for decoding binary LDPC codes. Here, we presented several simulation results for binary LDPC codes, such as the rate-0.5 (3,6)-regular [2640,1320] ``Margulis'' binary LDPC code $\mathcal{C}_3$ from \cite{Mackay-code}.
The information bits of $\mathcal{C}_3$ are modulated by binary phase shift keying (BPSK) and then transmitted over the AWGN channel.
Fig.\,\ref{fer2640} and Fig.\,\ref{iter2640} present the FER performance and number of iterations for code $\mathcal{C}_3$ under our proposed decoder and other competing binary LDPC decoders, respectively.
From Fig.\,\ref{fer2640}, we observe that our proposed decoder displays similar FER to the traditional binary SPA at low SNRs, but achieves better error-correction performance in high SNR regions. Specifically, the SPA decoder displays error floor behavior while our proposed decoder continuously drops in a waterfall manner.
Moreover, one can also find that the proposed proximal-ADMM decoding algorithm attains comparable FER performance to state-of-the-art binary ADMM-based penalized decoding (ADMM-PD) algorithms \cite{penalty-decoder} \cite{bai-admm-qp-binary} and outperforms the binary ADMM-LP decoding algorithm \cite{my-wcl} in terms of error-correction performance.
Fig.\,\ref{iter2640} shows that the average iteration numbers required to realize decoding are comparable between the proposed proximal-ADMM decoder and the ADMM-PD methods \cite{penalty-decoder} \cite{bai-admm-qp-binary}.
At last, we should also note that the proposed proximal-ADMM decoder is more efficient than the ADMM-PD decoder \cite{penalty-decoder} since it  eliminates time-consuming check-polytope projections.

\subsection{Parameter choices of the proposed proximal-ADMM decoder}
There are several parameters in the proposed proximal-ADMM decoding \emph{Algorithm \ref{proxiaml-ADMM-QP}}, including the penalty parameter $\mu$, the 2-norm penalty parameter $\alpha$, the ending tolerance $\xi$, the maximum number of iterations $t_{\rm max}$, and parameters $\rho$ and $\beta$.
Proper parameters can make \emph{Algorithm \ref{proxiaml-ADMM-QP}} achieve favourable error-correction performance and reduce the iteration number.
However, an exhaustive search over all possible parameters is not practical and necessary.
First, it is easy to see that a sufficiently large $t_{\rm max}$ and sufficiently small $\xi$ can lead to good error-correction performance for \emph{Algorithm \ref{proxiaml-ADMM-QP}}.
Thus, we fix the ending tolerance $\xi=10^{-5}$ and the maximum number of iterations $t_{max}=500$ in the simulations.
Moreover, a large $\beta$ parameter is favorable because it can make the updates for variables $\mathbf{v}$, $\mathbf{e}_1$, and $\mathbf{e}_2$ in every proximal-ADMM iteration not deviate too much from the stabilized iterate $\mathbf{p}$, $\mathbf{z}_1$, and $\mathbf{z}_2$ respectively (c.f.\cite{proximal-admm}). Therefore, we set $\beta$ to be 0.9 in the simulations.
Moreover, \emph{Algorithm \ref{proxiaml-ADMM-QP}} is sensitive to the settings of parameters $\mu$, $\alpha$ and $\rho$.
In order to guarantee that subproblem \eqref{proximal-x-update} is strongly convex with respect to variable $\x$, we let $\rho>\alpha$.

\begin{figure}[tp]
  \centering
  \centerline{\psfig{figure=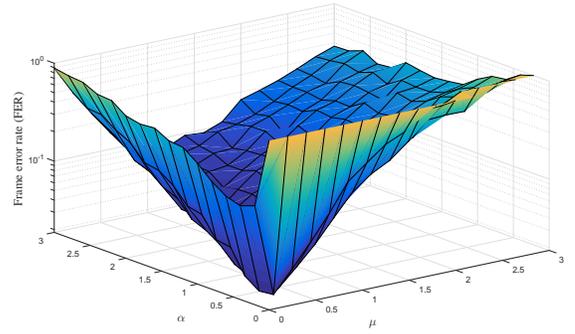,width=8.5cm,height=5cm}}
  \caption{FER performance plotted as a function of $\mu$ and $\alpha$ for the Tanner [1055,424] code $\mathcal{C}_{1}$ in $\mathbb{F}_{4}$.}
  \label{miu-alpha-fer}
\end{figure}
\begin{figure}[tp]
  \centering
  \centerline{\psfig{figure=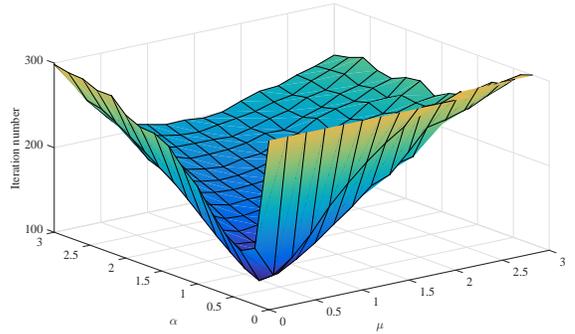,width=8.5cm,height=5cm}}
  \caption{Number of iterations plotted as a function of $\mu$ and $\alpha$ for the Tanner [1055,424] code $\mathcal{C}_{1}$ in $\mathbb{F}_{4}$.}
  \label{miu-alpha-iter}
\end{figure}

Next, we focus on how to choose parameters $\mu$ and $\alpha$. Fig.\,\ref{miu-alpha-fer} and Fig.\,\ref{miu-alpha-iter} plots FER performance and iteration numbers for code $\mathcal{C}_{1}$ as a function of parameters $\mu$ and $\alpha$ at $E_s/N_0=5$\,dB, respectively.
In the figures, we set $\rho=\alpha+0.2$ to ensure that $\rho>\alpha$ always holds.
Observing Fig.\,\ref{miu-alpha-fer}, one can find that \emph{Algorithm \ref{proxiaml-ADMM-QP}} achieves better FER performance as long as $\mu \in [0.3,1]$ and $\alpha \in [0.2,0.5]$.
Moreover, from Fig.\,\ref{miu-alpha-iter}, one can see that fewer iterations are required when $\mu \in [0.3,1]$ and $\alpha \in [0.2,0.5]$.
This means that $\mu\in[0.3,1]$ and $\alpha\in[0.2,0.5]$ are good choices in terms of error-correction performance and decoding efficiency of \emph{Algorithm \ref{proxiaml-ADMM-QP}}.

\section{Conclusion} \label{Conclusion}
In this paper, we propose an efficient proximal-ADMM decoding algorithm for nonbinary LDPC codes in $\mathbb{F}_{2^q}$.
We prove that the proposed decoding algorithm is guaranteed to converge to some stationary point of the formulated QP decoding problem.
Moreover, it is shown that the complexity of the proposed decoder in each iteration is roughly $\mathcal{O}(2^qn)$.
Simulation results demonstrate that the proposed proximal-ADMM decoding algorithm attains better error-correction performance and decoding efficiency than state-of-the-art nonbinary LDPC decoders.

\appendices

\section{Proof of Lemma \ref{hu-dD-u-define}}\label{hu-dD-u-define-proof}
\begin{proof}
Let $\mathbb{I}(\cdot)$ is an indicator function defined in $\mathbb{F}_{2^q}$, i.e., $\mathbb{I}(u, v)=1$ when $u=v$, otherwise $\mathbb{I}(u, v)=0$. Then, from the mapping rule defined in \eqref{Mq2b}, binary vector codeword $\x_k$, corresponding to $u_k \in \mathbb{F}_{2^q}$, can be expressed by
\begin{equation}\label{x_k-express}
\x_k =[\mathbb{I}(1, u_k);\ldots;\mathbb{I}(j, u_k);\ldots;\mathbb{I}(2^q-1, u_k)].
\end{equation}
Assuming $j=u_k$, one can find that $\mathbb{I}(j, u_k)=1$ and other elements in $\x_k$ are zeros.
Moreover, based on the definition of $D(2^q,h_k)_{ij}$ in \eqref{permu-matrix}, matrix $\mathbf{D}(2^q,h_k)$ can be rewritten as
 {\setlength\abovedisplayskip{4pt}
 \setlength\belowdisplayskip{4pt}
  \setlength\jot{1pt}
\begin{equation}\label{permu-matrix-express}
\begin{split}
&\mathbf{D}(2^q,h_k) \\
=&\!\!
\left[\!\!\!\!
  \begin{array}{ccccc}
    \mathbb{I}(1, h_k) \!\!\!&\!\!\! \cdots \!\!&\!\! \mathbb{I}(1,jh_k ) \!\!&\!\! \cdots \!\!&\!\! \mathbb{I}(1, (2^q\!-\!1)h_k) \\
    \vdots \!\!&\!\! \vdots & \vdots & \vdots & \vdots \\
    \mathbb{I}(i, h_k) \!\!\!&\!\!\! \cdots \!\!&\!\! \mathbb{I}(i, jh_k) \!\!&\!\! \cdots \!\!&\!\! \mathbb{I}(i, (2^q\!-\!1)h_k) \\
    \vdots \!\!\!&\!\!\! \vdots \!\!&\!\! \vdots \!\!&\!\! \vdots \!\!&\!\! \vdots \\
    \mathbb{I}(2^q\!\!-\!\!1\!, \!h_k\!) \!\!\!&\!\!\!\! \cdots \!\!&\!\!\!\! \mathbb{I}(2^q\!\!-\!\!1\!, \!jh_k) \!\!\!&\!\! \cdots \!\!\!&\!\!\! \mathbb{I}(2^q\!\!-\!\!1\!, \!(2^q\!\!-\!\!1)h_k\! ) \\
  \end{array}
\!\!\!\!\right].
\end{split}
\end{equation}
Then,} we have
 {\setlength\abovedisplayskip{3pt}
 \setlength\belowdisplayskip{3pt}
  \setlength\jot{1pt}
\begin{equation}\label{Dx-express}
\mathbf{D}(2^q,h_k)\x_k =\left[\!\!
                       \begin{array}{c}
                         \mathbb{I}(1,jh_k)\mathbb{I}(j,u_k) \\
                         \vdots \\
                         \mathbb{I}(i,jh_k)\mathbb{I}(j,u_k) \\
                         \vdots \\
                         \mathbb{I}(2^q\!\!-\!\!1,jh_k)\mathbb{I}(j,u_k) \\
                       \end{array}
                     \!\!\right]
\end{equation}
Since} $j=u_k$, $\mathbf{D}(2^q,h_k)\x_k$ can be further derived as follows when $i=h_ku_k\oplus 2^q$
\begin{equation}\label{Dx-express}
\mathbf{D}(2^q,h_k)\x_k\!=\!\!\left[\!\!\!\!
                       \begin{array}{c}
                         \mathbb{I}(1,h_ku_k) \\
                         \vdots \\
                         \mathbb{I}(i,h_ku_k) \\
                         \vdots \\
                         \mathbb{I}(2^q-1,h_ku_k) \\
                       \end{array}
                     \!\!\!\!\right]
\!\!=\!\!\left[\!\!\!\! \begin{array}{c}0\\ \vdots \\ \mathbb{I}(i,h_ku_k)\\ \vdots \\ 0\end{array}\!\!\!\right].
\end{equation}
Notice that $\mathbb{I}(i,h_ku_k)=1$, which completes the proof.
\end{proof}

\section{Proof of Lemma \ref{w-inverse-lemma}}\label{w-inverse-lemma-proof}
\begin{proof}
According to the definition of matrix $\A$ in \eqref{Ab-construct}, we have
\begin{equation}\label{ATA-derive1}
\begin{split}
&\A^T\A \\
=&\bigg(\sum_{\tau=1}^{\Gamma_{c}}\left(\hat{\mathbf{W}}_{\tau}(\mathbf{Q}_{\tau}\otimes \mathbf{I})\right)^{T}\left(\hat{\mathbf{W}}_{\tau}(\mathbf{Q}_{\tau}\otimes \mathbf{I})\right)\bigg)+\mathbf{S}^T\mathbf{S} \\
=&\bigg(\sum_{\tau=1}^{\Gamma_{c}}(\!\mathbf{Q}_{\tau}\otimes \mathbf{I})^T\hat{\mathbf{W}}_{\tau}^T\hat{\mathbf{W}}_{\tau}(\mathbf{Q}_{\tau}\otimes \mathbf{I})\bigg)+\mathbf{S}^T\mathbf{S}.
\end{split}
\end{equation}
Plugging \eqref{B} and \eqref{W} into \eqref{wW_W}, we can derive $\hat{\mathbf{W}}_{\tau}$ as follows
\begin{equation}\label{W-hat-equal}
\begin{split}
\hat{\mathbf{W}}_{\tau}=&(\mathbf{B}^T\otimes\mathbf{I})\mathbf{W}\!_{\tau} \!\oplus\! 2  \\
\!\!=&\begin{bmatrix} \mathbf{P}\displaystyle\bigg(\sum_{i\in\mathcal{K}_1}\mathbf{T}_{i}\bigg)\mathbf{D}_\tau;\dotsb; \mathbf{P}\displaystyle\bigg(\!\sum_{i\in\mathcal{K}_{2^q-1}}\!\!\!\!\mathbf{T}_{i}\bigg)\mathbf{D}_\tau
\end{bmatrix}\!\!\oplus\! 2,
\end{split}
\end{equation}
where $\mathcal{K}_1=\{1\}$, $\mathcal{K}_2=\{2\}$, $\mathcal{K}_3=\{1, 2\},\dotsb, \mathcal{K}_{2^q-1}=\{1,\dotsb,q\}$. Then, we have the following derivations
 {\setlength\abovedisplayskip{1pt}
 \setlength\belowdisplayskip{2pt}
  \setlength\jot{1pt}
\begin{equation}\label{wTw-derive1}
\begin{split}
\hat{\mathbf{W}}_{\tau}^T\hat{\mathbf{W}}_{\tau}
\!\!&=\!\!\sum_{\ell=1}^{2^q-1}\mathbf{D}_{\tau}^T\bigg(\!\Big(\!\sum_{i \in \mathcal{K}_{\ell}}\!\!\mathbf{T}_{i}\!\Big)^T\!\mathbf{P}^T\!\mathbf{P}\Big(\!\sum_{i \in \mathcal{K}_{\ell}}\!\!\mathbf{T}_{i}\!\Big)\oplus\!2\bigg)\mathbf{D}_{\tau}\!\!\\
&=4\!\sum_{\ell=1}^{2^q-1}\mathbf{D}_{\tau}^T\bigg(\Big(\!\sum_{i \in \mathcal{K}_{\ell}}\!\!\mathbf{T}_{i}\!\Big)^T\Big(\!\sum_{i \in \mathcal{K}_{\ell}}\!\!\mathbf{T}_{i}\!\Big)\!\oplus\!2\bigg)\mathbf{D}_{\tau} \\
&=4\mathbf{D}_{\tau}^T\sum_{\ell=1}^{2^q-1}\bigg(\Big(\!\sum_{i \in \mathcal{K}_{\ell}}\!\!\mathbf{T}_{i}\!\Big)^T\Big(\!\sum_{i \in \mathcal{K}_{\ell}}\!\!\mathbf{T}_{i}\!\Big)\!\oplus\!2\bigg)\mathbf{D}_{\tau},
\end{split}
\end{equation}
where} the second equality holds since $\mathbf{P}^T\mathbf{P}=4\mathbf{I}$. Moreover, since $\mathbf{T}_{i}={\rm diag}(\mathbf{b}_i^T,\mathbf{b}_i^T,\mathbf{b}_i^T)$ (see \eqref{t F matrix}), we can define
\begin{equation}\label{Phi-define}
\pmb{\Phi}=\sum_{\ell=1}^{2^q-1}\!\!\bigg(\Big(\!\sum_{i \in \mathcal{K}_{\ell}}\mathbf{b}_{i}\!\Big)\Big(\!\sum_{i \in \mathcal{K}_{\ell}}\mathbf{b}_{i}\!\Big)^T\oplus 2\bigg),
\end{equation}
and derive \eqref{wTw-derive1} as
\begin{equation}\label{wTw-derive2}
\hat{\mathbf{W}}_{\tau}^T\hat{\mathbf{W}}_{\tau}
=4\mathbf{D}_{\tau}^T\textrm{diag}(\pmb{\Phi},\pmb{\Phi},\pmb{\Phi})\mathbf{D}_{\tau}.
\end{equation}
Furthermore, since $\mathbf{b}_i^T$ is the $i$th row vector in matrix $\mathbf{B}$ (see \eqref{B}), $\pmb{\Phi}$ can be expressed as
\begin{equation}\label{phi-value}
\pmb{\Phi} = \left[
               \begin{array}{cccc}
                 2^{q-1} & 2^{q-2} & \cdots & 2^{q-2} \\
                 2^{q-2} & 2^{q-1} & \cdots & 2^{q-2} \\
                 \vdots & \vdots & \ddots & \vdots \\
                 2^{q-2} & 2^{q-2} & \cdots & 2^{q-1} \\
               \end{array}
             \right].
\end{equation}
From \eqref{permu-matrix}, we can see that either row vector or column vector in matrix $\mathbf{D}(2^q,h_{\tau_k})$ only includes one ``1''. Therefore, we can get
\begin{equation}\label{DphiD-phi}
\mathbf{D}(2^q,h_{\tau_k})^{T}\pmb{\Phi}\mathbf{D}(2^q,h_{\tau_k})=\pmb{\Phi},
\end{equation}
which implies that \eqref{wTw-derive2} can be further simplified to
\begin{equation}\label{wTw-derive3}
\hat{\mathbf{W}}_{\tau}^T\hat{\mathbf{W}}_{\tau}
=4\textrm{diag}(\pmb{\Phi},\pmb{\Phi},\pmb{\Phi}).
\end{equation}
Plugging \eqref{wTw-derive3} into \eqref{ATA-derive1} and noting $\mathbf{S}={\rm diag}(\mathbf{1}^T,\dotsb,\mathbf{1}^T)$, we have the following derivations
 {\setlength\abovedisplayskip{1pt}
 \setlength\belowdisplayskip{1pt}
  \setlength\jot{1pt}
\begin{equation}\label{ATA-derive2}
\begin{split}
& \A^T\A \\
=&4\bigg(\sum_{\tau=1}^{\Gamma_{c}}\!(\mathbf{Q}_{\tau}\!\otimes\! \mathbf{I})^T\textrm{diag}(\pmb{\Phi},\pmb{\Phi},\pmb{\Phi})\!(\mathbf{Q}_{\tau}\!\otimes\! \mathbf{I})\bigg)\!+\!\mathbf{S}^T\!\mathbf{S} \\
=&\textrm{diag}(4d_{1}\pmb{\Phi}\!\!+\!\!\mathbf{1}^T\mathbf{1},\dotsb,4d_{n+\Gamma_{a}}\pmb{\Phi}\!\!+\!\!\mathbf{1}^T\mathbf{1}),
\end{split}
\end{equation}
where} $d_{i}$, $i=1,\dotsb,n+\Gamma_a$, denotes the number of three-variables check equations that the $i$th information symbol participates in. From \eqref{ATA-derive2}, one can conclude that matrix $\big(\A^{T}\A+\epsilon\I\big)^{-1}$ is block diagonal. Letting
\begin{equation}\label{AAe}
  \mathbf{A}_{i}^T\mathbf{A}_i+\epsilon\mathbf{I}=4d_{i}\pmb{\Phi}+\mathbf{1}^T\!\mathbf{1}+\epsilon\mathbf{I},
\end{equation}
we can obtain \eqref{w-I-inverse-defi}, i.e.,
\[
\begin{split}
& \big(\A^{T}\A+\epsilon\I\big)^{-1} \\
= & {\rm diag}\Big(\big({\A}_1^{T}\A_1+\!\epsilon\I\big)^{-1},\ldots,\!\big(\A_{n+\Gamma_{a}}^{T}\A_{n+\Gamma_{a}}+\epsilon\I\big)^{-1}\Big).
\end{split}
\]

Plugging \eqref{phi-value} into \eqref{AAe}, we have
\begin{equation}\label{AAematrix}
  \begin{split}
  &\mathbf{A}_{i}^T\mathbf{A}_i+\epsilon\mathbf{I}\\
   =&\!\! \left[\!\!\!\!
               \begin{array}{cccc}
                 4d_{i}2^{q-1}\!\!+\!1\!+\!\epsilon \!\!&\!\! 4d_{i}2^{q-2}\!\!+\!1 \!\!&\!\!\! \cdots \!\!\!&\!\! 4d_{i}2^{q-2}\!\!+\!1 \\
                 4d_{i}2^{q-2}\!\!+\!1 \!\!&\!\! 4d_{i}2^{q-1}\!\!+\!1\!+\!\epsilon \!\!&\!\!\! \cdots \!\!\!&\!\! 4d_{i}2^{q-2}\!\!+\!1 \\
                 \vdots & \vdots & \ddots & \vdots \\
                 4d_{i}2^{q-2}\!\!+\!1 \!\!&\!\! 4d_{i}2^{q-2}\!\!+\!1 \!\!&\!\!\! \cdots \!\!\!&\!\! 4d_{i}2^{q-1}\!\!+\!1\!+\!\epsilon \\
               \end{array}
            \! \!\!\!\right]
   \end{split}.
\end{equation}
Observing \eqref{AAematrix}, it is easy to see that the diagonal elements in the matrix are the same and other elements are also the same. So its inverse matrix can be written as \cite{Matrix-analysis}
 {\setlength\abovedisplayskip{1pt}
 \setlength\belowdisplayskip{1pt}
  \setlength\jot{1pt}
\begin{equation}\label{AAematrix inv}
(\mathbf{A}_{i}^T\mathbf{A}_i+\epsilon\mathbf{I})^{-1} =\left[
  \begin{array}{cccc}
     \theta_{i} & \omega_{i} & \cdots & \omega_{i} \\
    \omega_{i} & \theta_{i} & \cdots & \omega_{i} \\
    \vdots & \vdots & \vdots & \vdots \\
    \omega_{i} & \omega_{i} & \cdots & \theta_{i} \\
  \end{array}
\right].
\end{equation}
Multiplying} the right sides of \eqref{AAematrix} and \eqref{AAematrix inv}, the corresponding equality leads to the following equations
\begin{equation*}\label{cal-a-b-equations}
\begin{split}
&(4d_{i}2^{q-1}\!\!+\!1\!+\!\epsilon)\theta_{i} + (4d_{i}2^{q-2}+1)(2^{q}-2)\omega_{i}=1, \\
&(\!4d_{i}2^{q\!-\!1}\!\!+\!1\!+\!\epsilon\!)\omega_{i}\!\! +\!\! (\!4d_{i}2^{q\!-\!2}\!\!+\!1\!)\theta_{i} \!+\! (\!4d_{i}2^{q\!-\!2}\!\!+\!1\!)(\!2^{q}\!\!-\!3\!)\omega_{i}\!=\!0.
\end{split}
\end{equation*}
Solving the above equations, we can obtain
 {\setlength\abovedisplayskip{1pt}
 \setlength\belowdisplayskip{1pt}
  \setlength\jot{1pt}
\[
\begin{split}
& \theta_{i}\!=\!\omega_{i}+\frac{1}{2^{q}d_{i}+\epsilon}, \\
& \omega_{i}\!=\!\frac{\!-\!2^{q}d_{i}}{(2^{q}d_{i}\!+\!\epsilon)[(2^{q+1}d_{i}\!+\!\epsilon\!+\!1)\!+\!(2^{q}d_{i}\!+\!1)(2^{q}\!-\!2)]}.
\end{split}
\]
This} completes the proof.
\end{proof}

\section{Proof of Lemma \ref{x-Lipschitz-constant}}\label{Lipschitz-continuous}
Based on the definition of $\mathcal{L}_{\mu}$ in \eqref{aug-Lagrangian}, its gradient, with respect to variable $\v$, can be calculated by
\begin{equation*}
\begin{split}
 \nabla_{\v} \mathcal{L}_{\mu}(\v,\e_{1}, \e_{2},&\y_{1}, \y_{2})\! =\! (\pmb{\lambda}-\alpha(\v-0.5)) +\A^{T}\y_1 +\y_2\\
&+\mu\A^T(\A\v+\e_1-\b)+\mu(\v-\e_2).
\end{split}
\end{equation*}
Then, we have the following derivations
\begin{equation*}
\begin{split}
& \|\nabla_{\v} \mathcal{L}_{\mu}(\v, \e_{1},  \e_{2}, \y_{1},  \y_{2})\!-\!\nabla_{\v} \mathcal{L}_{\mu}(\v^{\prime},\e_{1},\e_{2},\y_{1},\y_{2})\|_{2}\\
=&\|-\alpha(\v-\v^{\prime})+\mu(\v-\v^{\prime})+\mu\A^T\A(\v-\v^{\prime})\|_2\\
\leq & (\alpha+\mu+\mu\delta_{\A}^2)\|\v-\v^{\prime}\|_2, \\
\end{split}
\end{equation*}
where ``$\delta_{\A}$'' is the spectral norm of matrix $\A$. Letting $L=\alpha+\mu+\mu\delta_{\A}^2$, we obtain \eqref{x-Lipschitz-defi}.

\section{Proof of Theorem \ref{converge-proof-theorem}}\label{converge-proof}
Before we show its proof, we give one definition and two lemmas that are used to establish \emph{Theorem \ref{converge-proof-theorem}}.

\begin{definition}\label{F-D-P-define}
Define the following local functions
 {\setlength\abovedisplayskip{1pt}
 \setlength\belowdisplayskip{1pt}
  \setlength\jot{1pt}
\begin{equation}\label{D-define}
\begin{split}
& \mathcal{D}(\p,\! \z_{1},\! \z_{2}, \!\y_{1}, \!\y_{2})
\!\! =\!\!\min\limits_{\v, \e_{1} \succeq \mathbf{0}, \atop \mathbf{0} \preceq \e_{2} \preceq \mathbf{1}}\!\!\!\mathcal{F}(\v,\!\e_{1},\! \e_{2}, \!\p, \!\z_{1}, \!\z_{2},\! \y_{1},\! \y_{2}),
\end{split}
\end{equation}
\begin{equation}\label{x-D-define}
\begin{split}
& \hspace{-0.45cm} \left[\!\!\!\!\begin{array}{l}{\v(\p,\!\z_{1},\!\z_{2},\!\y_{1},\! \y_{2})} \\ {\e_{1}(\p,\! \z_{1},\! \z_{2},\! \y_{1},\! \y_{2})} \\ {\e_{2}(\p,\! \z_{1},\! \z_{2},\! \y_{1},\! \y_{2})}\end{array}\!\!\!\!\right]
\!\!\!=\!\underset{\v, \e_{1} \succeq \mathbf{0}, \atop \mathbf{0} \preceq \e_{2} \preceq \mathbf{1}}{\arg \min }~\!\mathcal{F}\!(\v,\! \e_{1},\! \e_{2},\! \p,\! \z_{1},\! \z_{2},\! \y_{1},\! \y_{2}\!),
\end{split}
\end{equation}
\begin{equation}\label{P-define}
\begin{split}
 & \mathcal{P}(\p,\! \z_{1},\! \z_{2})
\! =\!\!\!\min\limits_{\A\v+\e_1=\b, \v=\e_2, \atop \e_{1} \succeq \mathbf{0},  \mathbf{0} \preceq \e_{2} \preceq \mathbf{1}}\!\!\mathcal{F}(\v,\!\e_{1},\! \e_{2}, \!\p, \!\z_{1}, \!\z_{2},\! \y_{1},\! \y_{2}),
\end{split}
\end{equation}
\begin{equation}\label{x-P-define}
\begin{split}
& \left[\!\!\!\!\begin{array}{l}{\v(\p,\!\z_{1},\!\z_{2})\!} \\ {\e_{1}(\p,\! \z_{1},\! \z_{2})} \\ {\e_{2}(\p,\! \z_{1},\! \z_{2})}\end{array}\!\!\!\!\right]
\!\!\! =\!\!\!\!\underset{\A\v+\e_1=\b, \v=\e_2, \atop \e_{1} \succeq \mathbf{0},  \mathbf{0} \preceq \e_{2} \preceq \mathbf{1}}{\arg \min }\!\!\!\!\mathcal{F}(\v,\!\e_{1},\! \e_{2}, \!\p,\z_{1}, \!\z_{2},\! \y_{1},\! \y_{2}),
\end{split}
\end{equation}
where} function $\mathcal{F}(\cdot)$ is expressed by
 {\setlength\abovedisplayskip{2pt}
 \setlength\belowdisplayskip{2pt}
  \setlength\jot{1pt}
\begin{equation}\label{F-define}
\hspace{-0.2cm}\begin{split}
\mathcal{F}\!(\!\v,\!\e_{1},&\e_{2}\!,\!\p, \!\z_{1}, \!\z_{2},\! \y_{1},\! \y_{2})\!\!=\!\!\mathcal{L}_{\mu}\!(\v,\!\e_{1},\! \e_{2}, \!\p, \!\z_{1}, \!\z_{2},\! \y_{1},\! \y_{2}) \\
& \!\!+\!\frac{\rho}{2}\|\v\!-\!\p\|_{2}^{2}\!+\!\frac{\rho}{2}\left\|\e_{1}-\z_{1}\right\|_{2}^{2}+\frac{\rho}{2}\left\|\e_{2}-\z_{2}\right\|_{2}^{2}.
\end{split}
\end{equation}}
\end{definition}

Based on \emph{Definition 1}, the following inequalities hold.

\begin{lemma}\label{error-bound-lemma}
Suppose $\rho>\alpha>0$. then we have
\begin{equation}\label{error-bound1}
\begin{split}
& \|\!\!\!\left[\!\!\!\begin{array}{l}{\v^{k}-\v^{k+1}} \\ {\e_{1}^{k}-\e_{1}^{k+1}} \\ {\e_{2}^{k}-\e_{2}^{k+1}}\end{array}\!\!\!\right]\!\!\!\|_{2}^{2}
\geq \varepsilon_{1}\|\!\!\!\left[\!\!\!\begin{array}{c}{\v^{k}\! -\!\v\left(\p^{k}, \z_{1}^{k}, \z_{2}^{k}, \y_{1}^{k}, \y_{2}^{k}\right)}
 \\ {\e_{1}^{k}\!-\!\e_{1}\left(\p^{k}, \z_{1}^{k}, \z_{2}^{k}, \y_{1}^{k}, \y_{2}^{k}\right)}
 \\ {\e_{2}^{k}\!-\!\e_{2}\left(\p^{k}, \z_{1}^{k}, \z_{2}^{k}, \y_{1}^{k}, \y_{2}^{k}\right)}\end{array}\!\!\!\right]\!\!\!\|_{2}^{2},
\end{split}
\end{equation}
\begin{equation}\label{error-bound2}
\begin{split}
& \|\!\!\!\left[\!\!\!\!\begin{array}{l}{\v^{k}\!-\!\v^{k+1}} \\ {\e_{1}^{k}\!-\!\e_{1}^{k+1}} \\ {\e_{2}^{k}\!-\!\e_{2}^{k+1}}\end{array}\!\!\!\right]\!\!\!\|_{2}^{2}
\geq \!\varepsilon_{2}\|\!\!\!\left[\!\!\!\!\begin{array}{c}{\v^{k+1}\!\! -\!\v\left(\p^{k}, \z_{1}^{k}, \z_{2}^{k}, \y_{1}^{k}, \y_{2}^{k}\right)}
 \\ {\e_{1}^{k+1}\!\!-\!\e_{1}\left(\p^{k}, \z_{1}^{k}, \z_{2}^{k}, \y_{1}^{k}, \y_{2}^{k}\right)}
 \\ {\e_{2}^{k+1}\!\!-\!\e_{2}\left(\p^{k}, \z_{1}^{k}, \z_{2}^{k}, \y_{1}^{k}, \y_{2}^{k}\right)}\end{array}\!\!\!\right]\!\!\!\|_{2}^{2},
\end{split}
\end{equation}
\begin{equation}\label{error-bound3}
\begin{split}
&\hspace{-0.45cm}\|\!\!\left[\!\!\!\begin{array}{l}{\y_{1}\!\!-\!\y_{1}^{\prime}} \\ {\y_{2}\!\!-\!\y_{2}^{\prime}}\end{array} \!\!\!\right]\!\!\! \|_{2}^{2}
\!\geq\! \varepsilon_{3}\!\|\!\!\!\left[\!\!\!\!\begin{array}{l}{\v\!\left(\!\p, \z_{1}\!, \z_{2}\!, \y_{1}\!, \y_{2}\!\right)\!\!-\!\!\v\!\left(\!\p, \z_{1}\!, \z_{2}\!, \y_{1}^{\prime}\!, \y_{2}^{\prime}\!\right)}
\\ {\e_{1}\!\!\left(\!\p, \z_{1}\!, \z_{2}\!, \y_{1}\!, \y_{2}\!\right)\!\!-\!\!\e_{1}\!\!\left(\!\p, \z_{1}\!, \z_{2}\!, \y_{1}^{\prime}\!, \y_{2}^{\prime}\!\right)}
\\ {\e_{2}\!\!\left(\!\p, \z_{1}\!, \z_{2}\!, \y_{1}\!, \y_{2}\!\right)\!\!-\!\!\e_{2}\!\!\left(\!\p, \z_{1}\!, \z_{2}\!, \y_{1}^{\prime}\!, \y_{2}^{\prime}\!\right)}\end{array}\!\!\!\!\right]\!\!\!\|_2^2,
\end{split}
\end{equation}
\begin{equation}\label{error-bound4}
\begin{split}
&\hspace{-0.45cm}\|\!\!\!\left[\!\!\!\!\begin{array}{l}{\p^{k}\!\!-\!\p^{k+1}} \\ {\z_{1}^{k}\!\!-\!\z_{1}^{k+1}} \\ {\z_{2}^{k}\!\!-\!\z_{2}^{k+1}}\end{array}\!\!\!\!\right]\!\!\!\|_{2}^{2}
\!\geq\! \varepsilon_{4}\! \|\!\!\!\left[\!\!\!\!\begin{array}{l}{\v\!\left(\!\p^{k+1}\!, \z_{1}^{k+1}\!, \z_{2}^{k+1}\!\right)\!\!-\!\!\v\!\left(\!\p^{k}\!, \z_{1}^{k}\!, \z_{2}^{k}\!\right)}
\\ {\e_{1}\!\!\left(\!\p^{k+1}\!, \z_{1}^{k+1}\!, \z_{2}^{k+1}\!\right)\!\!-\!\!\e_{1}\!\!\left(\!\p^{k}\!, \z_{1}^{k}\!, \z_{2}^{k}\!\right)}
\\ {\e_{2}\!\!\left(\!\p^{k+1}\!, \z_{1}^{k+1}\!, \z_{2}^{k+1}\!\right)\!\!-\!\!\e_{2}\!\!\left(\!\p^{k}\!, \z_{1}^{k}\!, \z_{2}^{k}\!\right)}\end{array}\!\!\!\!\right]\!\!\!\|_{2}^{2},
\end{split}
\end{equation}
\begin{equation}\label{error-bound5}
\begin{split}
 \|\!\!\!\left[\!\!\!\!\begin{array}{l}{\p^{k}\!-\!\p^{k+1}} \\ {\z_{1}^{k}\!-\!\z_{1}^{k+1}} \\ {\z_{2}^{k}\!-\!\z_{2}^{k+1}}\end{array}\!\!\!\right]\!\!\!\|_{2}^{2}
& \!\geq \! \varepsilon_{5} \|\!\!\!\left[\!\!\!\!\begin{array}{l}{\v\!\left(\p^{k+1}\!, \z_{1}^{k+1}\!, \z_{2}^{k+1}\!, \y_{1}^{k+1}\!, \y_{2}^{k+1}\!\right)}
\\ {\e_{1}\!\!\left(\p^{k+1}\!, \z_{1}^{k+1}\!, \z_{2}^{k+1}\!, \y_{1}^{k+1}\!, \y_{2}^{k+1}\!\right)}
\\ {\e_{2}\!\!\left(\p^{k+1}\!, \z_{1}^{k+1}\!, \z_{2}^{k+1}\!, \y_{1}^{k+1}\!, \y_{2}^{k+1}\!\right)}\end{array}\right. \\
&~~~~ \left.\begin{array}{l}{-\v\!\left(\u^{k}, \z_{1}^{k}, \z_{2}^{k}, \y_{1}^{k+1}, \y_{2}^{k+1}\right)}
\\ {-\e_{1}\!\!\left(\p^{k}, \z_{1}^{k}, \z_{2}^{k}, \y_{1}^{k+1}, \y_{2}^{k+1}\right)}
\\ {-\e_{2}\!\!\left(\p^{k}, \z_{1}^{k}, \z_{2}^{k}, \y_{1}^{k+1}, \y_{2}^{k+1}\right)}\end{array}\!\!\!\!\right]\!\!\!\|_2^2,
\end{split}
\end{equation}
where
 {\setlength\abovedisplayskip{1pt}
 \setlength\belowdisplayskip{1pt}
  \setlength\jot{1pt}
\begin{equation}\label{varepsilon}
\begin{split}
&\varepsilon_1=\frac{(\rho-\alpha)^2}{(\rho+L+2)^2}, \  \varepsilon_2=\frac{(\rho-\alpha)^2}{(2\rho+L+2-\alpha)^2}, \\ &\varepsilon_3=\frac{(\rho-\alpha)^2}{\delta_{\A\I}^{2}}, \
\varepsilon_4=\frac{(\rho-\alpha)^2}{\rho^2}, \ \varepsilon_5=\frac{(\rho-\alpha)^2}{\rho^2}.
\end{split}
\end{equation}
Moreover,} if
$$\|\!\!\left[\!\!\!\begin{array}{c}{\A \v\!\left(\p,\! \z_{1},\! \z_{2},\! \y_{1},\! \y_{2}\right)\!\!+\!\!\e_{1}\!\left(\p,\! \z_{1},\! \z_{2},\! \y_{1},\! \y_{2}\right)\!\!-\!\!\b}
\\ {\v\!\left(\p,\! \z_{1},\! \z_{2},\! \y_{1},\! \y_{2}\right)\!-\!\e_{2}\!\left(\p,\! \z_{1},\! \z_{2},\! \y_{1},\! \y_{2}\right)}\end{array}\!\!\!\right]\!\!\|_2 \!\leq\! \Delta,$$
and
$$\|\!\!\left[\!\!\!\!\begin{array}{l}{\v-\p} \\ {\e_{1}-\z_1}
\\ {\e_{2}-\z_2}\end{array}\!\!\!\!\right]\!\!\|_2 \leq \Delta,$$
where $\Delta>0$ is some constant, then there exists $\varepsilon_{6}>0$ such that
 {\setlength\abovedisplayskip{1pt}
 \setlength\belowdisplayskip{1pt}
  \setlength\jot{1pt}
\begin{equation}\label{error-bound6}
\begin{split}
& \|\!\!\left[\!\!\begin{array}{l}{\y_{1}}-\y_1^{*}\left(\p, \z_{1}, \z_{2}\right) \\ {\y_{2}- \y_2^{*}\left(\u, \z_{1}, \z_{2}\right)}\end{array}\right]\|_2^2 \\
\leq &  \varepsilon_{6}\|\!\!\left[\!\!\!\!\begin{array}{c}{\A \v\!\left(\!\p, \z_{1}\!, \z_{2}\!, \y_{1}\!, \y_{2}\!\right)\!+\!\e_{1}\left(\!\p, \z_{1}\!, \z_{2}\!, \y_{1}\!, \y_{2}\!\right)\!-\!\b}
\\ {\v\!\left(\p, \z_{1}, \z_{2}, \y_{1}, \y_{2}\right)-\e_{2}\left(\p, \z_{1}, \z_{2}, \y_{1}, \y_{2}\right)}\end{array}\!\!\!\!\right]\!\!\|_{2}^{2},
\end{split}
\end{equation}
where} ${\y_1^{*}\left(\p, \z_{1}, \z_{2}\right)}$ and ${\y_2^{*}\left(\p, \z_{1}, \z_{2}\right)}$ are the solution sets of dual multipliers for problem \eqref{P-define}.
\end{lemma}

\begin{proof}
See Appendix \ref{error-bound-proof}.
\end{proof}

To save space, throughout the whole proof we denote functions $\mathcal{F}$, $\mathcal{D}$ and $\mathcal{P}$ at the $kth$ iteration by
\[
\begin{split}
&\F^{k} := \F\left(\v^{k}, \e_{1}^{k}, \e_{2}^{k}, \p^{k}, \z_{1}^{k}, \z_{2}^{k}, \y_{1}^{k}, \y_{2}^{k}\right), \\
&\D^{k} := \D\left(\p^{k}, \z_{1}^{k}, \z_{2}^{k}, \y_{1}^{k}, \y_{2}^{k}\right), \\
&\P^{k} := \P\left(\p^{k}, \z_{1}^{k}, \z_{2}^{k}\right),
\end{split}
\]
respectively. Using the above abbreviations, we further introduce the following lemma.
\begin{lemma}\label{F-D-P-change}
Let $\alpha\leq\mu\lambda_{\min}(\A^T\A)$. Then, the following inequalities hold
 {\setlength\abovedisplayskip{1pt}
 \setlength\belowdisplayskip{1pt}
  \setlength\jot{1pt}
\begin{equation}\label{F-change}
\begin{split}
 \F^{k}\!-\!\F^{k+1}\!\! \geq \!& \frac{\rho\!+\!\mu}{2}\|\!\!\!\left[\!\!\!\!\begin{array}{l}{\v^{k}\!-\!\v^{k+1}} \\ {\e_{1}^{k}\!-\!\e_{1}^{k+1}} \\ {\e_{2}^{k}\!-\!\e_{2}^{k+1}}\end{array}\!\!\!\right]\!\!\!\|_{2}^{2}
\!+\!\!\frac{\rho}{2 \beta}\!\|\!\!\!\left[\!\!\!\begin{array}{l}{\p^{k}\!-\!\p^{k+1}} \\ {\z_{1}^{k}\!-\!\z_{1}^{k+1}} \\ {\z_{2}^{k}\!-\!\z_{2}^{k+1}}\end{array}\!\!\!\right]\!\!\!\|_{2}^{2} \\
&\!-\!\mu\|\!\!\left[\!\!\!\begin{array}{c}{\A \v^{k}+\e_{1}^{k}-\b} \\ {\v^{k}-\e_{2}^{k}}\end{array}\!\!\!\right]\!\!\|_{2}^{2},
\end{split}
\end{equation}
\begin{equation}\label{D-change}
\begin{split}
 & \D^{k+1}\!\!-\!\!\D^{k}
\!\geq\!  \mu\!\!\left[\!\!\!\!\begin{array}{c}{\A\v^{k}\!+\!\e_{1}^{k}\!-\!\b} \\ {\v^{k}\!-\!\e_{2}^{k}}\end{array}\!\!\!\right]^{T}\!\!\!\!\! \pmb{\phi}^{k}
\!\!+\!\!\frac{\rho}{2}\left[\!\!\!\!\begin{array}{c}{\p^{k+1}\!-\!\p^{k}} \\ {\z_{1}^{k+1}\!-\!\z_{1}^{k}} \\ {\z_{2}^{k+1}\!-\!\z_{2}^{k}}\end{array}\!\!\!\right]^{T}\!\!\!\!\! \pmb{\psi}^{k},
\end{split}
\end{equation}
\begin{equation}\label{P-change}
\begin{split}
 \P^{k}\!-\! \P^{k+1} \geq  & \rho\!\left[\!\!\!\begin{array}{c}{\p^{k+1}\!-\!\p^{k}} \\ {\z_{1}^{k+1}\!-\!\z_{1}^{k}} \\ {\z_{2}^{k+1}\!-\!\z_{2}^{k}}\end{array}\!\!\!\right]^{T}\!\!\!
\left[\!\!\!\!\begin{array}{l}{\v\left(\p^{k}, \z_{1}^{k}, \z_{2}^{k}\right)-\p^{k}} \\ {\e_{1}\left(\p^{k}, \z_{1}^{k}, \z_{2}^{k}\right)-\z_{1}^{k}}
\\ {\e_{2}\left(\p^{k}, \z_{1}^{k}, \z_{2}^{k}\right)-\z_{2}^{k}}\end{array}\!\!\!\!\right] \\
&-\frac{\rho\eta}{2}\|\!\!\!\left[\!\!\!\begin{array}{c}{\p^{k+1}-\p^{k}} \\ {\z_{1}^{k+1}-\z_{1}^{k}} \\ {\z_{2}^{k+1}-\z_{2}^{k}}\end{array}\!\!\!\right]\!\!\!\|_2^2,
\end{split}
\end{equation}
where} $\eta$ and ``$\pmb{\phi}^{k}$'' and ``$\pmb{\psi}^{k}$'' are defined as follows
\begin{subequations}\label{eta D-part1-change D-part2-change}
 \begin{align}
  &\eta=1+\frac{1}{\sqrt{\epsilon_4}}, \\
  &\pmb{\phi}^{k}\!\!=\!\!\left[\!\!\!\begin{array}{c}{\A\!\v\!\left(\p^{k}\!,\! \z_{1}^{k}\!,\!\z_{2}^{k}\!,\!\y_{1}^{k+1}\!,\! \y_{2}^{k+1}\!\right)\!+\!\e_{1}\!\left(\p^{k}\!,\!\z_{1}^{k}\!,\!\z_{2}^{k}\!,\! \y_{1}^{k+1}\!,\!\y_{2}^{k+1}\!\right)\!\!-\!\b}
\\ {\v\!\left(\p^{k}\!,\!\z_{1}^{k},\!\z_{2}^{k},\!\y_{1}^{k+1}, \y_{2}^{k+1}\right)\!-\!\e_{2}\left(\p^{k},\!\z_{1}^{k},\!\z_{2}^{k},\! \y_{1}^{k+1},\!\y_{2}^{k+1}\right)}\end{array}\!\!\!\right], \label{D-part1-change}\\
  &\pmb{\psi}^{k}\!\!=\!\!\left[\!\!\!\begin{array}{c}{\p^{k+1}\!+\!\p^{k}\!-\!2 \v\left(\p^{k+1},\!\z_{1}^{k+1},\!\z_{2}^{k+1},\!\y_{1}^{k+1},\! \y_{2}^{k+1}\!\right)}
\\ {\z_{1}^{k+1}\!+\!\z_{1}^{k}\!-\!2\e_{1}\!\left(\p^{k+1},\!\z_{1}^{k+1},\! \z_{2}^{k+1},\!\y_{1}^{k+1},\! \y_{2}^{k+1}\!\right)}
\\ {\z_{2}^{k+1}\!+\!\z_{2}^{k}\!-\!2\e_{2}\!\left(\p^{k+1},\!\z_{1}^{k+1},\! \z_{2}^{k+1},\!\y_{1}^{k+1},\!\y_{2}^{k+1}\!\right)}\end{array}\!\!\!\right].
\end{align}
\end{subequations}
\end{lemma}

\begin{proof}
See Appendix \ref{descent3-proof}.
\end{proof}

Now we are ready to prove \emph{Theorem \ref{converge-proof-theorem}}.

\begin{proof}
First, we define a potential function as follows
\begin{equation}\label{Psi-function}
\begin{split}
& \Psi=\F-2 \D+2 \P.
\end{split}
\end{equation}
The key to proving convergence of the proposed proximal-ADMM algorithm is to verify that the function $\Psi$ not only {\it decreases  sufficiently} in each iteration but also is lower-bounded.

Based on \eqref{F-change}-\eqref{P-change} in \emph{Lemma \ref{F-D-P-change}}, we obtain
\begin{equation}\label{Psi-change1}
\begin{split}
 & \Psi^{k}-\Psi^{k+1} \\
= & \left(\F^{k}-\F^{k+1}\right)+2\left(\D^{k+1}-\D^{k}\right)+2\left(\P^{k}-\P^{k+1}\right) \\
\geq & \frac{\rho\!+\!\mu}{2} \|\!\!\!\left[\!\!\begin{array}{l}{\v^{k}-\v^{k+1}} \\ {\e_{1}^{k}-\e_{1}^{k+1}} \\ {\e_{2}^{k}-\e_{2}^{k+1}}\end{array}\!\!\!\right]\!\!\!\|_2^2
+\!\left(\!\frac{\rho}{2 \beta}-\rho \eta\!\right)\|\!\!\!\left[\!\!\!\begin{array}{c}{\p^{k+1}-\p^{k}} \\ {\z_{1}^{k+1}-\z_{1}^{k}} \\ {\z_{2}^{k+1}-\z_{2}^{k}}\end{array}\!\!\!\right]\!\!\!\|_2^2\\
&\!+\!\rho\! \left[\!\!\!\begin{array}{c}{\p^{k+1}-\p^{k}} \\ {\z_{1}^{k+1}-\z_{1}^{k}} \\ {\z_{2}^{k+1}-\z_{2}^{k}}\end{array}\!\!\!\right]^{T}\!\!\!
\left(\!\!\pmb{\psi}^k+2\left[\!\!\!\!\begin{array}{c}{\v\left(\p^{k}, \z_{1}^{k}, \z_{2}^{k}\right)-\p^{k}} \\ {\e_{1}\left(\p^{k}, \z_{1}^{k}, \z_{2}^{k}\right)-\z_{1}^{k}}
\\ {\e_{2}\left(\p^{k}, \z_{1}^{k}, \z_{2}^{k}\right)-\z_{2}^{k}}\end{array}\!\!\!\!\right]\!\right) \\
&-\mu\left(\!\!\|\!\!\left[\!\!\!\begin{array}{c}{\A \v^{k}+\e_{1}^{k}-\b} \\ {\v^{k}-\e_{2}^{k}}\end{array}\!\!\!\right]\!\!\|_{2}^{2}
\!-\!2\left[\!\!\!\begin{array}{c}{\A\v^{k}\!+\!\e_{1}^{k}\!-\!\b} \\ {\v^{k}\!-\!\e_{2}^{k}}\end{array}\!\!\!\right]^{T}\!\!\!\! \pmb{\phi}^k\!\right).
\end{split}
\end{equation}
For the last term of \eqref{Psi-change1}, we have the following derivations
\[
\begin{split}
\hspace{-0.35cm}
& \|\!\!\left[\!\!\!\begin{array}{c}{\A \v^{k}+\e_{1}^{k}-\b} \\ {\v^{k}-\e_{2}^{k}}\end{array}\!\!\!\right]\!\!\|_{2}^{2}
\!-\!2\left[\!\!\!\begin{array}{c}{\A\v^{k}\!+\!\e_{1}^{k}\!-\!\b} \\ {\v^{k}\!-\!\e_{2}^{k}}\end{array}\!\!\!\right]^{T}\!\!\!\! \pmb{\phi}^k \\
\hspace{-0.35cm}
=&\|\!\!\left[\!\!\!\begin{array}{ccc}{\A} \!\!&\!\! {\I_{M}} \!\!\!&\!\!\! {\mathbf{0}} \\ {\I_{N}} \!\!&\!\! {\mathbf{0}} \!\!\!&\!\!\! {-\I_{N}}\end{array}\!\!\!\!\right]\!\!\!
\left[\!\!\!\!\begin{array}{c}{\v^{k}\!\!-\!\!\v\!\left(\!\p^{k}, \z_{1}^{k}, \z_{2}^{k}, \y_{1}^{k+1}\!\!, \y_{2}^{k+1}\!\right)}
\\ {\e_{1}^{k}\!\!-\!\!\e_{1}\!\left(\!\p^{k}, \z_{1}^{k}, \z_{2}^{k}, \y_{1}^{k+1}\!\!, \y_{2}^{k+1}\!\right)}
\\ {\e_{2}^{k}\!\!-\!\!\e_{2}\!\left(\!\p^{k}, \z_{1}^{k}, \z_{2}^{k}, \y_{1}^{k+1}\!\!, \y_{2}^{k+1}\!\right)}\end{array}\!\!\!\!\right]
\!\!\!\|_2^2 \!-\!\! \|\pmb{\phi}^k\|_2^{2} \\
\hspace{-0.35cm}
 \leq &  \delta_{\A\I}^{2}\|\!\!\left[\!\!\!\!\begin{array}{c}{\v^{k}\!\!-\!\!\v\!\left(\!\p^{k}, \z_{1}^{k}, \z_{2}^{k}, \y_{1}^{k+1}\!\!, \y_{2}^{k+1}\!\right)}
\\ {\e_{1}^{k}\!\!-\!\!\v_{1}\!\left(\!\p^{k}, \z_{1}^{k}, \z_{2}^{k}, \y_{1}^{k+1}\!\!, \y_{2}^{k+1}\!\right)}
\\ {\e_{2}^{k}\!\!-\!\!\v_{2}\!\left(\!\p^{k}, \z_{1}^{k}, \z_{2}^{k}, \y_{1}^{k+1}\!\!, \y_{2}^{k+1}\!\right)}\end{array}\!\!\!\!\right]\!\!\|_2^2\!-\! \|\pmb{\phi}^k\|_2^{2}.
\end{split}
\]
From \eqref{error-bound1} in Lemma \ref{error-bound-lemma}, we can further get
\begin{equation}\label{Psi-change11}
\begin{split}
& \|\!\!\left[\!\!\!\begin{array}{c}{\A \v^{k}+\e_{1}^{k}-\b} \\ {\v^{k}-\e_{2}^{k}}\end{array}\!\!\!\right]\!\!\|_{2}^{2}
\!-\!2\left[\!\!\!\begin{array}{c}{\A\v^{k}\!+\!\e_{1}^{k}\!-\!\b} \\ {\v^{k}\!-\!\e_{2}^{k}}\end{array}\!\!\!\right]^{T}\!\!\!\! \pmb{\phi}^k \\
\leq & \frac{\delta_{\A\I}^{2}}{\varepsilon_{1}}
\|\!\!\left[\!\!\begin{array}{l}{\v^{k}-\v^{k+1}} \\ {\e_{1}^{k}-\e_{1}^{k+1}} \\ {\e_{2}^{k}-\e_{2}^{k+1}}\end{array}\!\!\!\right]\!\!\!\|_2^2\!-\! \|\pmb{\phi}^k\|_2^{2}.
\end{split}
\end{equation}
Plugging \eqref{Psi-change11} into \eqref{Psi-change1}, the inequality can be revised as
\begin{equation}\label{Psi-change2}
\begin{split}
\hspace{-0.4cm}  & \Psi^{k}-\Psi^{k+1} \\
\hspace{-0.4cm} \geq & \left(\!\!\frac{\rho\!\!+\!\!\mu}{2}\!\!-\!\mu\!\frac{\delta_{\A\I}^{2}}{\varepsilon_{1}}\!\!\right)\!\!
\|\!\!\!\left[\!\!\!\!\begin{array}{l}{\v^{k}\!-\!\v^{k+1}\!} \\ {\e_{1}^{k}\!-\!\e_{1}^{k+1}\!} \\ {\e_{2}^{k}\!-\!\e_{2}^{k+1}\!}\end{array}\!\!\!\right]\!\!\!\|_2^2
\!+\!\!\left(\!\!\frac{\rho}{2 \beta}\!-\!\rho \eta\!\right)\!\!\|\!\!\!\left[\!\!\!\!\begin{array}{c}{\p^{k+1}\!-\!\p^{k}} \\ {\z_{1}^{k+1}\!-\!\z_{1}^{k}} \\ {\z_{2}^{k+1}\!-\!\z_{2}^{k}}\end{array}\!\!\!\right]\!\!\!\|_2^2\\
\hspace{-0.47cm} &\!+\! \mu\|\pmb{\phi}^k\|_2^{2}
\!\!+\!\!\rho\! {\left[\!\!\!\!\begin{array}{c}{\p^{k+1}\!\!-\!\mathbf{p}^{k}} \\ {\z_{1}^{k+1}\!\!-\!\z_{1}^{k}} \\ {\z_{2}^{k+1}\!-\!\z_{2}^{k}}\end{array}\!\!\!\!\right]^{T}\!\!\!\!\!
\left(\!\!\pmb{\psi}^k\!\!+\!\!2\!\left[\!\!\!\!\begin{array}{c}{\v\left(\p^{k}, \z_{1}^{k}, \z_{2}^{k}\right)\!-\!\p^{k}} \\ {\v_{1}\!\left(\p^{k}, \z_{1}^{k}, \z_{2}^{k}\right)\!-\!\z_{1}^{k}}
\\ {\v_{2}\!\left(\p^{k}, \z_{1}^{k}, \z_{2}^{k}\right)\!-\!\z_{2}^{k}}\end{array}\!\!\!\!\right]\!\right)}.
\end{split}
\end{equation}

To facilitate derivations later, we define
\begin{equation*}
\begin{split}
& \mathcal{X}^{k}\! :=\!\! \left[\!\!\!\!\begin{array}{l}{\v\left(\p^{k+1}, \z_{1}^{k+1}, \z_{2}^{k+1}, \y_{1}^{k+1}, \y_{2}^{k+1}\right)\!-\!\v\left(\p^{k}, \z_{1}^{k}, \z_{2}^{k}\right)}
\\ {\e_{1}\left(\p^{k+1}, \z_{1}^{k+1}, \z_{2}^{k+1}, \y_{1}^{k+1}, \y_{2}^{k+1}\right)\!-\!\e_{1}\left(\p^{k}, \z_{1}^{k}, \z_{2}^{k}\right)}
\\ {\e_{2}\left(\p^{k+1}, \z_{1}^{k+1}, \z_{2}^{k+1}, \y_{1}^{k+1}, \y_{2}^{k+1}\right)\!-\!\e_{2}\left(\p^{k}, \z_{1}^{k}, \z_{2}^{k}\right)}\end{array}\!\!\!\!\right].
\end{split}
\end{equation*}
Then, the last term in \eqref{Psi-change2} can be rewritten as \eqref{Psi-change22}.
Applying property $2ab\leq a^{2} / \lambda^{2}+\lambda^{2} b^{2}$, we derive ``$\eqref{Psi-change22}-(a)$'' as
\begin{figure*}
\begin{equation}\label{Psi-change22}
\begin{split}
\left[\!\!\!\!\begin{array}{c}{\p^{k+1}\!-\p^{k}} \\ {\z_{1}^{k+1}\!-\z_{1}^{k}} \\ {\z_{2}^{k+1}\!-\z_{2}^{k}}\end{array}\!\!\!\!\right]^{T}\!\!\!\!
& \left(\!\!\pmb{\psi}^k+2\!\left[\!\!\!\!\begin{array}{c}{\v\left(\p^{k}, \z_{1}^{k}, \z_{2}^{k}\right)-\p^{k}} \\ {\e_{1}\left(\p^{k}, \z_{1}^{k}, \z_{2}^{k}\right)-\z_{1}^{k}}
\\ {\e_{2}\left(\p^{k}, \z_{1}^{k}, \z_{2}^{k}\right)-\z_{2}^{k}}\end{array}\!\!\!\!\right]\!\right)
=\|\!\!\left[\!\!\!\!\begin{array}{c}{\p^{k+1}\!-\!\p^{k}} \\ {\z_{1}^{k+1}\!-\!\z_{1}^{k}} \\ {\z_{2}^{k+1}\!-\!\z_{2}^{k}}\end{array}\!\!\!\right]\!\!\!\|_2^2
-\underbrace{2\left[\!\!\!\!\begin{array}{c}{\p^{k+1}\!-\mathbf{p}^{k}} \\ {\z_{1}^{k+1}\!-\z_{1}^{k}} \\ {\z_{2}^{k+1}\!-\z_{2}^{k}}\end{array}\!\!\!\right]^{T}\!\!\!\!\!\mathcal{X}^{k}\!\!}_{\eqref{Psi-change22}-(a)}\\
&~~~~~~~~~~ -2\underbrace{\left[\!\!\!\!\begin{array}{c}{\p^{k+1}\!-\mathbf{p}^{k}} \\ {\z_{1}^{k+1}\!-\z_{1}^{k}} \\ {\z_{2}^{k+1}\!-\z_{2}^{k}}\end{array}\!\!\!\right]^{T}\!\!\!
\left[\!\!\!\begin{array}{c}{\v\left(\p^{k}, \z_{1}^{k}, \z_{2}^{k}, \y_{1}^{k+1}, \y_{2}^{k+1}\right)-\v\left(\p^{k}, \z_{1}^{k}, \z_{2}^{k}\right)}
\\ {\e_{1}\left(\p^{k}, \z_{1}^{k}, \z_{2}^{k}, \y_{1}^{k+1}, \y_{2}^{k+1}\right)-\e_{1}\left(\p^{k}, \z_{1}^{k}, \z_{2}^{k}\right)}
\\ {\e_{2}\left(\p^{k}, \z_{1}^{k}, \z_{2}^{k}, \y_{1}^{k+1}, \y_{2}^{k+1}\right)-\e_{2}\left(\p^{k}, \z_{1}^{k}, \z_{2}^{k}\right)}\end{array}\!\!\!\right]}_{\eqref{Psi-change22}-(b)}.
\end{split}
\end{equation}
\end{figure*}
\begin{equation}\label{Psi-change222}
\begin{split}
2\left[\!\!\!\!\begin{array}{c}{\p^{k+1}\!-\mathbf{p}^{k}} \\ {\z_{1}^{k+1}\!-\z_{1}^{k}} \\ {\z_{2}^{k+1}\!-\z_{2}^{k}}\end{array}\!\!\!\!\right]^{T}\!\!\!\!\mathcal{X}^{k}
&\leq \|\!\!\left[\!\!\!\!\begin{array}{c}{\p^{k+1}\!-\p^{k}} \\ {\z_{1}^{k+1}\!-\z_{1}^{k}} \\ {\z_{2}^{k+1}\!-\z_{2}^{k}}\end{array}\!\!\!\!\right]\!\!\|_2^2/\lambda^2
+ \lambda^2 \|\mathcal{X}^{k}\|_2^2.
\end{split}
\end{equation}
Moreover, according to Cauchy-Schwarz inequality and \eqref{error-bound5} in \emph{Lemma \ref{error-bound-lemma}}, inequality \eqref{Psi-change221} holds.
\begin{figure*}
\begin{equation}\label{Psi-change221}
\begin{split}
& \left[\!\!\!\!\begin{array}{c}{\p^{k+1}\!-\p^{k}} \\ {\z_{1}^{k+1}\!-\z_{1}^{k}} \\ {\z_{2}^{k+1}\!-\z_{2}^{k}}\end{array}\!\!\!\!\right]^{T}\!\!\!
 \left[\!\!\!\begin{array}{c}{\v\left(\p^{k}, \z_{1}^{k}, \z_{2}^{k}, \y_{1}^{k+1}, \y_{2}^{k+1}\right)-\v\left(\p^{k}, \z_{1}^{k}, \z_{2}^{k}\right)}
\\ {\e_{1}\left(\p^{k}, \z_{1}^{k}, \z_{2}^{k}, \y_{1}^{k+1}, \y_{2}^{k+1}\right)-\e_{1}\left(\p^{k}, \z_{1}^{k}, \z_{2}^{k}\right)}
\\ {\e_{2}\left(\p^{k}, \z_{1}^{k}, \z_{2}^{k}, \y_{1}^{k+1}, \y_{2}^{k+1}\right)-\e_{2}\left(\p^{k}, \z_{1}^{k}, \z_{2}^{k}\right)}\end{array}\!\!\!\right]
\leq \frac{1}{\sqrt{\varepsilon_{5}}}\|\!\!\!\left[\!\!\!\!\begin{array}{c}{\p^{k+1}\!-\!\p^{k}} \\ {\z_{1}^{k+1}\!-\!\z_{1}^{k}} \\ {\z_{2}^{k+1}\!-\!\z_{2}^{k}}\end{array}\!\!\!\right]\!\!\!\|_2^2.
\end{split}
\end{equation}
\hrulefill
\vspace*{4pt}
\end{figure*}
Then, plugging \eqref{Psi-change222} and \eqref{Psi-change221} into \eqref{Psi-change22}, we can obtain
\begin{equation}\label{Psi-change22-2}
\begin{split}
& \left[\!\!\!\!\begin{array}{c}{\p^{k+1}\!-\!\p^{k}} \\ {\z_{1}^{k+1}\!-\!\z_{1}^{k}} \\ {\z_{2}^{k+1}\!-\!\z_{2}^{k}}\end{array}\!\!\!\!\right]^{T}\!\!\!\!
\left(\!\!\pmb{\psi}^k\!+\!2\!\left[\!\!\!\!\begin{array}{c}{\v\left(\p^{k}, \z_{1}^{k}, \z_{2}^{k}\right)-\p^{k}} \\ {\e_{1}\left(\p^{k}, \z_{1}^{k}, \z_{2}^{k}\right)-\z_{1}^{k}}
\\ {\e_{2}\left(\p^{k}, \z_{1}^{k}, \z_{2}^{k}\right)-\z_{2}^{k}}\end{array}\!\!\!\right]\!\right) \\
\geq & \left(1-\frac{1}{\lambda^{2}}-\frac{2}{\sqrt{\varepsilon_{5}}}\right)
\|\!\!\!\left[\!\!\!\!\begin{array}{c}{\p^{k+1}\!-\p^{k}} \\ {\z_{1}^{k+1}\!-\z_{1}^{k}} \\ {\z_{2}^{k+1}\!-\z_{2}^{k}}\end{array}\!\!\!\right]\!\!\!\|_2
-\lambda^{2}\|\mathcal{X}^{k}\|_2^2.
\end{split}
\end{equation}
Furthermore, plugging \eqref{Psi-change22-2} into \eqref{Psi-change2}, we have
\begin{equation}\label{Psi-change3}
\begin{split}
\hspace{-0.3cm}  \Psi^{k}&\!-\!\Psi^{k+1}\!
 \geq\!
\rho\left(\!\!\frac{1}{2 \beta}\!-\! \eta\!+\!1\!-\!\frac{1}{\lambda^{2}}\!-\!\frac{2}{\sqrt{\varepsilon_{5}}}\right)\!\!\|\!\!\!\left[\!\!\!\!\begin{array}{c}{\p^{k+1}\!-\!\p^{k}} \\ {\z_{1}^{k+1}\!-\!\z_{1}^{k}} \\ {\z_{2}^{k+1}\!-\!\z_{2}^{k}}\end{array}\!\!\!\right]\!\!\!\|_2^2\\
&\!\!\!\!+\!\!\left(\!\!\frac{\rho\!+\!\mu}{2}\!-\!\mu\frac{\delta_{\A\I}^{2}}{\varepsilon_{1}}\!\!\right)\!\!
\|\!\!\!\left[\!\!\!\!\begin{array}{l}{\v^{k}\!-\!\v^{k+1}\!} \\ {\e_{1}^{k}\!-\!\e_{1}^{k+1}\!} \\ {\e_{2}^{k}\!-\!\e_{2}^{k+1}\!}\end{array}\!\!\right]\!\!\!\|_2^2
\!+\! \mu\|\pmb{\phi}^k\|_2^{2}\!-\!\rho\lambda^{2}\|\mathcal{X}^{k}\|_2^2.
\end{split}
\end{equation}
Letting $\lambda^{2}\!=\!\omega\beta$ and noticing  $\eta\!=\!1\!+\!\frac{1}{\sqrt{\varepsilon_{4}}}$(see \eqref{eta D-part1-change D-part2-change}), one can verify $\eta-1+\frac{1}{\lambda^{2}}+\frac{2}{\sqrt{\varepsilon_{5}}} \leq\frac{1}{3\beta}$  when
$\omega \geq \frac{6 \sqrt{\varepsilon_{4} \varepsilon_{5}}}{\sqrt{\varepsilon_{4} \varepsilon_{5}}-6 \beta\left(2\sqrt{\varepsilon_{4}}+\sqrt{\varepsilon_{5}}\right)}$.
Moreover, since $\varepsilon_{1}\! =\!\frac{(\rho-\alpha)^{2}}{(\rho+L+2)^{2}}$ (see \eqref{epsilon_varphi}) and the assumption $\mu \leq \frac{\rho(\rho-\alpha)^{2}}{4 \delta_{\A\I}^{2}(\rho+L+2)^{2}-(\rho-\alpha)^{2}}$ in \emph{Theorem \ref{converge-proof-theorem}}, one can verify that $\frac{\mu \delta_{\A\I}^{2}}{\varepsilon_{1}} \leq \frac{\rho+\mu}{4}$ holds.
Then, \eqref{Psi-change3} can be deduced as follows
\begin{equation}\label{Psi-change4}
\begin{split}
\Psi^{k}-\Psi^{k+1}
\geq & \frac{\rho}{3 \beta}\!\left[\!\!\!\!\begin{array}{c}{\p^{k+1}\!-\!\p^{k}} \\ {\z_{1}^{k+1}\!-\!\z_{1}^{k}} \\ {\z_{2}^{k+1}\!-\!\z_{2}^{k}}\end{array}\!\!\!\right]\!\!\!\|_2^2\!+\!\frac{\rho\!+\!\mu}{4}
\|\!\!\!\left[\!\!\!\!\begin{array}{l}{\v^{k}\!-\!\v^{k+1}\!} \\ {\e_{1}^{k}\!-\!\e_{1}^{k+1}\!} \\ {\e_{2}^{k}\!-\!\e_{2}^{k+1}\!}\end{array}\!\!\!\right]\!\!\!\|_2^2
 \\
&\!+\! \mu\|\pmb{\phi}^k\|_2^{2}\!-\!\rho \omega \beta\|\mathcal{X}^{k}\|_2^2.
\end{split}
\end{equation}

In the following, we show that term $\rho\omega\beta\|\mathcal{X}^{k}\|_2^2$ can be bounded by the previous three terms.

First, since $\mathbf{0}\preceq\v\preceq\mathbf{1}$, $\A\v$ is bounded.
Moreover, since $\e_1 \succeq \mathbf{0}$ and $\A\v+\e_1-\b=\mathbf{0}$, there exists some positive vector $\pmb{\theta}$ such that $\e_1\preceq \pmb{\theta}$. Then, we can define
\begin{equation}\label{V-define}
\begin{split}
& V: =\max\limits_{\mathbf{0} \preceq \v,\v^{\prime},\e_{2},\e_2^{\prime} \preceq \mathbf{1}, \atop \mathbf{0}\preceq \e_{1},\e_1^{\prime} \preceq \pmb{\theta}}
\|\!\!\left[\!\!\begin{array}{l}{\v} \\ {\e_{1}} \\ {\e_{2}}\end{array}\!\!\right]\!\!-\!\!\left[\!\!\begin{array}{c}{\v^{\prime}} \\ {\e_{1}^{\prime}} \\ {\e_{2}^{\prime}}\end{array}\!\!\right]\!\!\|_{2}.
\end{split}
\end{equation}
Moreover, we define
\begin{equation}\label{deita-function-define}
\begin{split}
&\zeta:=\min \{\Delta, \sigma(\Delta / \sqrt{6 \omega})\},
\end{split}
\end{equation}
where $\sigma(\cdot)$ is some function satisfying $\underset{\epsilon\rightarrow0}\lim \sigma(\epsilon)=0$. Since $0<\beta\leq1$ (see \eqref{proximal-ADMM update_LP} below), we can denote $\beta$'s upper-bound as
\begin{equation}\label{beita-upper-define}
\begin{split}
& \beta < \min \left\{1, \frac{(\rho+\mu) \zeta^{2}}{8 \rho \omega V^{2}}, \frac{\zeta^{2} \mu}{2 \rho \omega V^{2}}, \frac{\mu \varepsilon_{3}}{2 \rho \omega \varepsilon_{6}}\right\},
\end{split}
\end{equation}
Moreover, we define the following inequalities
\begin{subequations}
\begin{align}
 &\|\!\!\left[\begin{array}{c}{\v^{k}-\v^{k+1}} \\ {\e_{1}^{k}-\e_{1}^{k+1}} \\ {\e_{2}^{k}-\e_{2}^{k+1}}\end{array}\right]\!\rVert_{2}^{2} \leq \frac{8 \rho \omega V^{2} \beta}{\rho+\mu},   \label{supposed-inequality-1} \\
 & \|\pmb{\phi}^k\|_2^2 \leq \frac{2 \rho \omega V^{2}}{\mu} \beta, \label{supposed-inequality-2} \\
 & \|\!\!\!\left[\!\!\!\begin{array}{c}{\p^{k}}-{\p^{k+1}} \\ {\z_{1}^{k}}-{\z_{1}^{k+1}} \\ {\z_{2}^{k}}-{\z_{2}^{k+1}}\end{array}\!\!\!\right]\!\!\!\|_{2}^{2}
 \leq 6 \omega \beta^{2} \|\mathcal{X}^{k}\|_2^2. \label{supposed-inequality-3}
\end{align}
\end{subequations}

Now, we are ready to check the boundness of $\mathcal{X}^{k}$.
First, we assume all of the inequalities \eqref{supposed-inequality-1}--\eqref{supposed-inequality-3} hold. Then,
plugging \eqref{V-define}-\eqref{beita-upper-define} into \eqref{supposed-inequality-2}, we can obtain \eqref{case1-inequality2} and \eqref{case1-inequality2_2} simultaneously.
\begin{equation}\label{case1-inequality2}
\begin{split}
& \|\pmb{\phi}^k\|_2 \leq \Delta,
\end{split}
\end{equation}
\begin{equation}\label{case1-inequality2_2}
\begin{split}
& \|\pmb{\phi}^k\|_2 \leq \sigma(\Delta / \sqrt{6 \omega}).
\end{split}
\end{equation}
Then, we can obtain\footnotemark \footnotetext{See proofs in Appendix E.}
\begin{equation}\label{case1-inequality3}
\begin{split}
& \|\mathcal{X}^{k}\|_2 \leq \Delta / \sqrt{6 \omega}.
\end{split}
\end{equation}
Then, combining \eqref{proximal-z-update}, \eqref{supposed-inequality-3}, and \eqref{case1-inequality3}, we have
\begin{equation}\label{error-bound6-hold}
\begin{split}
 & \|\!\!\!\left[\!\!\!\begin{array}{c}{\v^{k+1}\!-\!\p^k} \\ {\e_{1}^{k+1}\!-\!\z_1^k} \\ {\e_{2}^{k+1}\!-\!\z_2^{k}}\end{array}\!\!\!\right]\!\!\!\|_{2}^{2}
=\!\|\!\!\!\left[\!\!\!\begin{array}{c}{\p^{k+1}\!-\p^{k}} \\ {\z_{1}^{k+1}\!-\z_{1}^{k}} \\ {\z_{2}^{k+1}\!-\z_{2}^{k}}\end{array}\!\!\!\right]\!\!\!\|_{2}^{2}/\beta^{2}
\!\leq \! 6 \omega \|\mathcal{X}^{k}\|_2^2 \! \leq\! \Delta.
\end{split}
\end{equation}
Combining it with \eqref{case1-inequality2}, we can see \eqref{error-bound6} holds.

Moreover, noticing $\mathcal{X}^k$ and $\pmb\phi^k$ are on the right side of the inequalities \eqref{error-bound6} and \eqref{error-bound3} respectively, we have the following inequality chain
\begin{equation}\label{case1-inequality4}
 \|\mathcal{X}^{k}\|_2^2
\leq \frac{1}{\varepsilon_{3}} \|\!\!\left[\!\!\begin{array}{c}{\y_{1}^{k+1}}\!-\!\y_1^{*}\left(\p^{k}, \z_{1}^{k}, \z_{2}^{k}\right) \\ {\y_{2}^{k+1}}\!\!-\!\!{\y_2^{*}\left(\p^{k}, \z_{1}^{k}, \z_{2}^{k}\right)}\end{array}\!\!\right]
\|_{2}^{2}
\!\!\leq\!\!  \frac{\varepsilon_6}{\varepsilon_{3}}\|\pmb{\phi}^k\|_2^2,
\end{equation}
where the first inequality comes from \eqref{error-bound3}, and the second inequality comes from \eqref{error-bound6}. Moreover, according to  \eqref{beita-upper-define}, \eqref{case1-inequality4} can be further derived to
\begin{equation}\label{case1-inequality5}
 \rho\omega\beta\|\mathcal{X}^{k}\|_2^2
\leq  \frac{\mu}{2}\|\pmb{\phi}^k\|_2^2.
\end{equation}

Next, we consider the case that at least one of the inequalities \eqref{supposed-inequality-1}-\eqref{supposed-inequality-3} does not hold.
There are three scenarios:
\begin{enumerate}
  \item \eqref{supposed-inequality-1} does not hold i.e.,
$
 \|\!\!\left[\begin{array}{c}{\v^{k}-\v^{k+1}} \\ {\e_{1}^{k}-\e_{1}^{k+1}} \\ {\e_{2}^{k}-\e_{2}^{k+1}}\end{array}\right]\!\!\!\!\|_{2}^{2} > \frac{8 \rho \omega V^{2} \beta}{\rho+\mu}.
$
Then, we have the following derivations
\begin{equation}\label{Psi-change5-case}
\begin{split}
\hspace{-20pt}&\frac{\rho\!+\!\mu}{4}
\|\!\!\!\left[\!\!\begin{array}{l}{\v^{k}\!-\!\v^{k+1}\!} \\ {\e_{1}^{k}\!-\!\e_{1}^{k+1}\!} \\ {\e_{2}^{k}\!-\!\e_{2}^{k+1}\!}\end{array}\!\!\right]\!\!\!\|_2^2
-\rho\omega\beta\|\mathcal{X}^{k}\|_2^2  \\
\hspace{-20pt}\geq &  \!\frac{\rho\!+\!\mu}{8}\!
\|\!\!\!\left[\!\!\!\!\begin{array}{l}{\v^{k}\!\!-\!\!\v\!^{k+1}\!} \\ {\e_{1}^{k}\!\!-\!\!\e\!_{1}^{k+1}\!} \\ {\e_{2}^{k}\!\!-\!\!\e\!_{2}^{k+1}\!}\end{array}\!\!\!\right]\!\!\!\|_2^2
\!\!+\!\! \frac{\rho\!+\!\mu}{8}\! \cdot \! \frac{8 \rho \omega V^{2} \beta}{\rho\!+\!\mu}
\!\!-\!\rho\omega\beta\|\mathcal{X}^{k}\|_2^2 \\
= & \frac{\rho\!+\!\mu}{8}
\|\!\!\!\left[\!\!\!\!\begin{array}{l}{\v^{k}\!-\!\v^{k+1}\!} \\ {\e_{1}^{k}\!-\!\e_{1}^{k+1}\!} \\ {\e_{2}^{k}\!-\!\e_{2}^{k+1}\!}\end{array}\!\!\!\right]\!\!\!\|_2^2
+ \rho\omega \beta V^{2}-\rho\omega\beta\|\mathcal{X}^{k}\|_2^2.
\end{split}
\end{equation}
Since $V\geq\|\mathcal{X}^{k}\|_2$ (see \eqref{V-define}), we can further get
\begin{equation}\label{Psi-change5-case21}
\hspace{-20pt}\frac{\rho\!+\!\mu}{4}
\|\!\!\!\left[\!\!\!\!\begin{array}{l}{\v^{k}\!-\!\v^{k+1}\!} \\ {\e_{1}^{k}\!-\!\e_{1}^{k+1}\!} \\ {\e_{2}^{k}\!-\!\e_{2}^{k+1}\!}\end{array}\!\!\right]\!\!\!\|_2^2
\!\!-\!\!\rho\omega\beta\|\mathcal{X}^{k}\|_2^2
\!\!\geq\!\!  \frac{\rho\!+\!\mu}{8}
\|\!\!\!\left[\!\!\!\!\begin{array}{l}{\v^{k}\!-\!\v^{k+1}\!} \\ {\e_{1}^{k}\!-\!\e_{1}^{k+1}\!} \\ {\e_{2}^{k}\!-\!\e_{2}^{k+1}\!}\end{array}\!\!\right]\!\!\!\|_2^2.
\end{equation}

  \item \eqref{supposed-inequality-2} does not hold, i.e.,
$
 \|\pmb{\phi}^k\|_2^2 > \frac{2 \rho \omega V^{2}}{\mu} \beta.
$
By exploiting the above inequality and $V\geq\|\mathcal{X}^{k}\|_2^2$,  we can get
\begin{equation}\label{Psi-change5-case22}
\begin{split}
&\mu\|\pmb{\phi}^k\|_2^2- \rho\omega\beta\|\mathcal{X}^{k}\|_2^2 \geq  \frac{\mu}{2}\|\pmb{\phi}^k\|_2^2.
\end{split}
\end{equation}
  \item \eqref{supposed-inequality-3} does not hold, i.e.,
$
\|\!\!\!\left[\!\!\!\begin{array}{c}{\p^{k}}\!\!-\!\!{\p^{k+1}} \\ {\z_{1}^{k}}\!\!-\!\!{\z_{1}^{k+1}} \\ {\z_{2}^{k}}\!\!-\!\!{\z_{2}^{k+1}}\end{array}\!\!\!\right]\!\!\!\|_{2}^{2}
 \!\!>\!\! 6\omega\beta^{2}\!\|\!\mathcal{X}^{k}\!\|_2^2.
$
Through similar derivations to \eqref{Psi-change5-case21} and \eqref{Psi-change5-case22}, we can obtain
\begin{equation}\label{Psi-change5-case23}
\hspace{-9pt}\frac{\rho}{3\beta}\|\!\!\!\left[\!\!\!\begin{array}{c}{\p^{k}}\!\!-\!\!{\p^{k+1}} \\ {\z_{1}^{k}}\!\!-\!\!{\z_{1}^{k+1}} \\ {\z_{2}^{k}}\!\!-\!\!{\z_{2}^{k+1}}\end{array}\!\!\!\right]\!\!\!\|_{2}^{2}
\!-\!\rho\omega\beta\|\mathcal{X}^{k}\!\|_2^2
\!\!\geq\!\!\frac{\rho}{6\beta}\|\!\!\!\left[\!\!\!\begin{array}{c}{\p^{k}}\!\!-\!\!{\p^{k+1}} \\ {\z_{1}^{k}}\!\!-\!\!{\z_{1}^{k+1}} \\ {\z_{2}^{k}}\!\!-\!\!{\z_{2}^{k+1}}\end{array}\!\!\!\right]\!\!\!\|_{2}^{2}.
\end{equation}
\end{enumerate}

Then, from \eqref{case1-inequality5}, \eqref{Psi-change5-case21}, \eqref{Psi-change5-case22}, and \eqref{Psi-change5-case23}, we can derive \eqref{Psi-change4} as
\begin{equation}\label{Psi-change5-case2}
\begin{split}
&\Psi^{k}-\Psi^{k+1} \\
\geq & \frac{\rho\!+\!\mu}{8}
\|\!\!\!\left[\!\!\!\!\begin{array}{l}{\v^{k}\!-\!\v^{k+1}\!} \\ {\e_{1}^{k}\!-\!\e_{1}^{k+1}\!} \\ {\e_{2}^{k}\!-\!\e_{2}^{k+1}\!}\end{array}\!\!\!\right]\!\!\!\|_2^2
\!+\!\frac{\rho}{6 \beta}\!\left[\!\!\!\!\begin{array}{c}{\p^{k+1}\!-\!\p^{k}} \\ {\z_{1}^{k+1}\!-\!\z_{1}^{k}} \\ {\z_{2}^{k+1}\!-\!\z_{2}^{k}}\end{array}\!\!\!\right]\!\!\!\|_2^2
 \!+\! \frac{\mu}{2}\|\pmb{\phi}^k\|_2^{2}.
\end{split}
\end{equation}

Adding both sides of the above inequality from $k = 1, 2, \ldots,$ we can get
 {\setlength\abovedisplayskip{10pt}
 \setlength\belowdisplayskip{10pt}
  \setlength\jot{8pt}
\begin{equation}\label{Psi-sum}
\begin{split}
\hspace{-10pt}\underset{k\rightarrow+\infty}\lim\Psi^{1}\!\!-\!\!\Psi^{k+1}&
\!\geq\!\frac{\rho\!+\!\mu}{8}
\sum_{k=1}^{+\infty}\|\!\!\!\left[\!\!\!\!\begin{array}{l}{\v^{k}\!-\!\v^{k+1}\!} \\ {\e_{1}^{k}\!-\!\e_{1}^{k+1}\!} \\ {\e_{2}^{k}\!-\!\e_{2}^{k+1}\!}\end{array}\!\!\!\right]\!\!\!\|_2^2 \\
&\!+\!\frac{\rho}{6 \beta}\!\sum_{k=1}^{+\infty}\left[\!\!\!\!\begin{array}{c}{\p^{k+1}\!-\!\p^{k}} \\ {\z_{1}^{k+1}\!-\!\z_{1}^{k}} \\ {\z_{2}^{k+1}\!-\!\z_{2}^{k}}\end{array}\!\!\!\right]\!\!\!\|_2^2
 \!+\! \frac{\mu}{2}\sum_{k=1}^{+\infty}\|\pmb{\phi}^k\|_2^{2}.
\end{split}
\end{equation}
According} to \emph{Definition \ref{F-D-P-define}},
we can see that $\F^{k} \geq \D^{k}$, $\P^{k} \geq \D^{k}$ and $\P^{k}$ are lower-bounded.
Therefore, $\Psi^k=(\F^{k}-\D^{k})+(\P^{k}-\D^{k})+\P^{k}$ means that $\Psi$ is also lower-bounded. Therefore, we can obtain
 {\setlength\abovedisplayskip{5pt}
 \setlength\belowdisplayskip{5pt}
  \setlength\jot{5pt}
\begin{equation}\label{xhat-limit-divided}
\begin{split}
& \lim_{k \rightarrow +\infty}\v^{k}-\v^{k+1}=\mathbf{0}, ~~ \lim _{k \rightarrow +\infty}\e_1^{k}-\e_1^{k+1}=\mathbf{0}, \\
& \lim _{k \rightarrow +\infty}\e_2^{k}-\e_2^{k+1}=\mathbf{0}, ~~ \lim _{k \rightarrow +\infty}\p^{k}-\p^{k+1}=\mathbf{0}, \\
& \lim _{k \rightarrow +\infty}\z_1^{k}-\z_1^{k+1}=\mathbf{0}, ~~~ \lim _{k \rightarrow +\infty}\z_2^{k}-\z_2^{k+1}=\mathbf{0},
\end{split}
\end{equation}
and}
\begin{equation}\label{Wx-limit}
\begin{split}
& \lim _{k \rightarrow +\infty} \pmb{\phi}^k =\mathbf{0}.
\end{split}
\end{equation}
Plugging \eqref{xhat-limit-divided} into \eqref{proximal-z-update}, we can obtain
 {\setlength\abovedisplayskip{10pt}
 \setlength\belowdisplayskip{10pt}
  \setlength\jot{8pt}
\begin{equation}\label{xvhat-limit-divided}
\begin{split}
& \lim_{k \rightarrow +\infty}\v^{k+1}-\p^{k}=\mathbf{0}, \\
& \lim _{k \rightarrow +\infty}\e_1^{k+1}-\z_1^{k}=\mathbf{0},\\
& \lim _{k \rightarrow +\infty}\e_2^{k+1}-\z_2^{k}=\mathbf{0}.
\end{split}
\end{equation}
Plugging} \eqref{xhat-limit-divided} into \eqref{error-bound1} and \eqref{error-bound5} respectively, we have
\begin{equation}\label{limit-error-bound2}
\begin{split}
& \lim_{k \rightarrow +\infty}\left[\!\!\!\begin{array}{c}{\v^{k}\! -\!\v\left(\p^{k}, \z_{1}^{k}, \z_{2}^{k}, \y_{1}^{k}, \y_{2}^{k}\right)}
 \\ {\e_{1}^{k}\!-\!\e_{1}\left(\p^{k}, \z_{1}^{k}, \z_{2}^{k}, \y_{1}^{k}, \y_{2}^{k}\right)}
 \\ {\e_{2}^{k}\!-\!\e_{2}\left(\p^{k}, \z_{1}^{k}, \z_{2}^{k}, \y_{1}^{k}, \y_{2}^{k}\right)}\end{array}\!\!\!\right]=\mathbf{0},
\end{split}
\end{equation}
and
\begin{equation}\label{limit-error-bound5}
\begin{split}
\lim_{k\!\rightarrow\!+\!\infty} &\left[\!\!\!\!\begin{array}{l}{\v\!\left(\p^{k+1}\!, \z_{1}^{k+1}\!, \z_{2}^{k+1}\!, \y_{1}^{k+1}\!, \y_{2}^{k+1}\right)}
\\ {\e_{1}\!\!\left(\p^{k+1}\!, \z_{1}^{k+1}\!, \z_{2}^{k+1}\!, \y_{1}^{k+1}\!, \y_{2}^{k+1}\right)}
\\ {\e_{2}\!\!\left(\p^{k+1}\!, \z_{1}^{k+1}\!, \z_{2}^{k+1}\!, \y_{1}^{k+1}\!, \y_{2}^{k+1}\right)}\end{array}\right. \\
&\hspace{1cm}\left.\begin{array}{l}{-\v\!\left(\u^{k}, \z_{1}^{k}, \z_{2}^{k}, \y_{1}^{k+1}, \y_{2}^{k+1}\right)}
\\ {-\e_{1}\!\!\left(\p^{k}, \z_{1}^{k}, \z_{2}^{k}, \y_{1}^{k+1}, \y_{2}^{k+1}\right)}
\\ {-\e_{2}\!\!\left(\p^{k}, \z_{1}^{k}, \z_{2}^{k}, \y_{1}^{k+1}, \y_{2}^{k+1}\right)}\end{array}\!\!\!\!\right]=\mathbf{0},
\end{split}
\end{equation}
respectively.
From \eqref{proximal-lamda-update}, we have
\begin{equation}\label{Wx-limit2}
\begin{split}
\left[\begin{array}{ccc} \mathbf{y}_1^{k+1}\!-\!\mathbf{y}_1^k\\ \mathbf{y}_2^{k+1}\!-\!\mathbf{y}_2^k \end{array}\right]\!=\!\left[\!\!\!\begin{array}{ccc}{\A} \!\!&\!\! {\I_{M}} \!\!\!&\!\!\! {\mathbf{0}} \\ {\I_{N}} \!\!&\!\! {\mathbf{0}} \!\!\!&\!\!\! {-\I_{N}}\end{array}\!\!\!\!\right]\!\!
\left[\!\!\!\!\begin{array}{c}{\v^{k+1}} \\ {\e_{1}^{k+1}}\\ {\e_{2}^{k+1}}\end{array}\!\!\!\right]\!\!-\!\!\left[\!\!\begin{array}{l}{\b} \\ {\mathbf{0}}\end{array}\!\!\right].
\end{split}
\end{equation}
Plugging \eqref{D-part1-change} into \eqref{Wx-limit2}, we have
\begin{equation}\label{Wx-limit3}
\begin{split}
&\left[\begin{array}{ccc} \mathbf{y}_1^{k+1}\!-\!\mathbf{y}_1^k\\ \mathbf{y}_2^{k+1}\!-\!\mathbf{y}_2^k \end{array}\right] \\
= & \!\!\left[\!\!\!\begin{array}{ccc}{\A} \!\!&\!\! {\I_{M}} \!\!\!&\!\!\! {\mathbf{0}} \\ \!{\I_{N}}\!\!\!&\!\!{\mathbf{0}} \!\!\!&\!\!\! {-\!\I_{N}}\end{array}\!\!\!\!\right]\!\!\!
\left[\!\!\!\!\begin{array}{c}{\v^{k\!+\!1}\!\!-\!\!\v\!\!\left(\p^{k\!+\!1}, \z_{1}^{k\!+\!1}, \z_{2}^{k\!+\!1}, \y_{1}^{k\!+\!1}, \y_{2}^{k\!+\!1}\right)}
 \\ {\e_{1}^{k\!+\!1}\!\!-\!\e_{1}\!\!\left(\p^{k\!+\!1}, \z_{1}^{k\!+\!1}, \z_{2}^{k\!+\!1}, \y_{1}^{k\!+\!1}, \y_{2}^{k\!+\!1}\right)}
 \\ {\e_{2}^{k\!+\!1}\!\!-\!\e_{2}\!\!\left(\p^{k\!+\!1}, \z_{1}^{k\!+\!1}, \z_{2}^{k\!+\!1}, \y_{1}^{k\!+\!1}, \y_{2}^{k\!+\!1}\right)}\end{array}\!\!\!\!\right] \\
&+ \!\!\left[\!\!\!\begin{array}{ccc}{\A} \!\!&\!\! {\I_{M}} \!\!\!&\!\!\! {\mathbf{0}} \\ {\I_{N}} \!\!&\!\! {\mathbf{0}} \!\!\!&\!\!\! {-\I_{N}}\end{array}\!\!\!\!\right]\left[\!\!\!\!\begin{array}{l}{\v\!\left(\p^{k+1}\!, \z_{1}^{k+1}\!, \z_{2}^{k+1}\!, \y_{1}^{k+1}\!, \y_{2}^{k+1}\right)}
\\ {\e_{1}\!\!\left(\p^{k+1}\!, \z_{1}^{k+1}\!, \z_{2}^{k+1}\!, \y_{1}^{k+1}\!, \y_{2}^{k+1}\right)}
\\ {\e_{2}\!\!\left(\p^{k+1}\!, \z_{1}^{k+1}\!, \z_{2}^{k+1}\!, \y_{1}^{k+1}\!, \y_{2}^{k+1}\right)}\end{array}\right. \\
&\hspace{2.2cm}\left.\begin{array}{l}{\!-\v\!\left(\p^{k}, \z_{1}^{k}, \z_{2}^{k}, \y_{1}^{k+1}, \y_{2}^{k+1}\right)}
\\ {\!-\e_{1}\!\!\left(\p^{k}, \z_{1}^{k}, \z_{2}^{k}, \y_{1}^{k+1}, \y_{2}^{k+1}\right)}
\\ {\!-\e_{2}\!\!\left(\p^{k}, \z_{1}^{k}, \z_{2}^{k}, \y_{1}^{k+1}, \y_{2}^{k+1}\right)}\end{array}\!\!\!\!\right]\!\! +\!\pmb{\phi}^k.
\end{split}
\end{equation}
Plugging \eqref{Wx-limit}, \eqref{limit-error-bound2}, and \eqref{limit-error-bound5} into \eqref{Wx-limit3}, we can obtain
\begin{equation}\label{dual-limit}
\lim_{k \rightarrow +\infty} \mathbf{y}_1^{k+1}\!-\!\mathbf{y}_1^k=\mathbf{0}, \ \lim_{k \rightarrow +\infty} \mathbf{y}_2^{k+1}\!-\!\mathbf{y}_2^k=\mathbf{0}.
\end{equation}

From \eqref{pADMM-frame-b} and \eqref{pADMM-frame-c}, we can see clearly that $\{\e_1^k\}$ and $\{\v^k\}$ are bounded sequences since $\{\e_2^k\}$ is bounded by $[0,1]^N$. Plugging these bounded results into \eqref{xvhat-limit-divided}, we can find that $\{\p^{k}\}$, $\{\z_1^{k}\}$, $\{\z_2^{k}\}$ are also bounded sequences.
Furthermore, based on the above bounded results, \eqref{v1-solution-component} and \eqref{v2-solution-component} imply that $\{\y_1^{k}\}$ and $\{\y_2^{k}\}$ are also bounded sequences.

Combining these bounded results with \eqref{xhat-limit-divided} and \eqref{dual-limit}, we can conclude
\begin{equation}\label{variables-limit}
\begin{split}
& \lim_{k \rightarrow +\infty}\v^{k}=\v^{*}, ~ \lim _{k \rightarrow +\infty}\e_1^{k}=\e_1^{*},~\lim _{k \rightarrow +\infty}\e_2^{k}=\e_2^{*},\\
& \lim _{k \rightarrow +\infty}\p^{k}=\p^{*}, ~ \lim _{k \rightarrow +\infty}\z_1^{k}=\z_1^{*}, ~\lim _{k \rightarrow +\infty}\z_2^{k}=\z_2^{*}, \\
&\lim_{k \rightarrow +\infty}\y_1^{k}=\y_1^{*}, ~\lim _{k \rightarrow +\infty}\y_2^{k}=\y_2^{*}.
\end{split}
\end{equation}

Plugging \eqref{variables-limit} into \eqref{pADMM-frame-b} and \eqref{xvhat-limit-divided}, we can get
\begin{equation}\label{Wx-limit-star}
\begin{split}
& \A\v^{*}+\e_1^{*}-\b = \mathbf{0}, ~~ \v^{*}=\e_2^{*}, \\
& \v^{*}=\p^{*}, ~~ \e_1^{*}=\z_1^{*}, ~~ \e_2^{*}=\z_2^{*}.
\end{split}
\end{equation}
which completes the proof of the first part of {\it Theorem 1}.

Next, we prove that $\v^*$ is a stationary point of the original problem \eqref{ML-decoding-all}. Letting $g(\v)=\pmb{\lambda}^T\v-\frac{\alpha}{2}\|\v-0.5\|_2^2$, we can obtain, $\forall \x \in X$,
\begin{equation}\label{stationary1}
\begin{split}
(\v-\v^{*})^T\nabla_{\v} g(\v^*) &= (\v-\v^{*})^T(\pmb{\lambda}-\alpha(\v-0.5)). \\
\end{split}
\end{equation}
Moreover, since $\v^{k+1}\!\!=\!\!\mathop{\rm argmin}\limits_{\v}  \F\left(\v,\! \e_{1}^{k},\!\e_{2}^{k},\!\p^{k},\!\z_{1}^{k},\!\z_{2}^{k},\!\y_{1}^{k},\!\y_{2}^{k}\right)$, then we have
\begin{equation}\label{grandient-x-star}
\begin{split}
0= &\nabla_{\v}\F(\v^{*},\e_{1}^{*},\e_2^{*},\p^*,\z_1^*,\z_2^*,\y_1^*,\y_2^*) \\
=&\pmb{\lambda}-\alpha(\v^*-0.5)+\rho(\v^*-\p^*)+\A^{T}\y_1^{*}+\y_2^* \\
& +\mu\A^T(\A\v^*+\e_1^*-\b)+\mu(\v^*-\e_2^*)\\
=&\nabla_{\v} g(\v^*)+\A^{T}\y_1^{*}+\y_2^*,
\end{split}
\end{equation}
i.e., $\nabla_{\v} g(\v^*)=-\A^{T}\y_1^{*}-\y_2^*$, where the last equality follows from \eqref{Wx-limit-star}. Then, we can further obtain
\begin{equation}\label{stationary2}
\begin{split}
\hspace{-0.4cm}(\v\!-\!\v^{*})^T\nabla_{\v} g(\v^*) & \!=\! \!-\!(\v-\v^{*})^T\!(\A^{T}\y_1^{*}\!+\!\y_2^*) \\
\hspace{-0.4cm} &\! =\! \!-\!(\y_1^{*})^T\!\A(\v\!-\!\v^{*})\!\!-\!\!(\y_2^*)^{T}\!(\v\!-\!\v^{*}).
\end{split}
\end{equation}
Obviously, if $(\y_1^{*})^T\A(\v-\v^{*})\leq 0$ and $(\y_2^{*})^{T}(\v-\v^{*})\leq 0$, then $(\v\!-\!\v^{*})^T\nabla_{\v} g(\v^*) \geq 0$.
In the following, we will prove that both of them hold.

First, we have
\begin{equation}\label{first-term1}
\begin{split}
(\y_1^{*})^T\A(\v-\v^{*}) & = (\y_1^{*})^T\left((\b-\e_1)-(\b-\e_1^*)\right) \\
& = (\y_1^{*})^T(\e_1^*-\e_1).
\end{split}
\end{equation}

Then, since $\e_1^* \succeq 0$, we have the following derivations
\begin{equation}\label{y-z-star-1}
\begin{split}
 (\y_1^{*})^T\e_1^* & \!=\! \sum_{e_{1,j}^{*}>0} y_{1,j}^*e_{1,j}^{*} \\
& \!=\! \sum_{e_{1,j}^{*}>0} \! \frac{\mu}{\rho+\mu} y_{1,j}^*(b_{j}\!-\!\mathbf{a}_{j}^{T}\v^{*}\!\!-\!\!\frac{y_{1,j}^*}{\mu}\!+\!\frac{\rho}{\mu}z_{1,j}^{*}) \\
&\! =\! \sum_{e_{1,j}^{*}>0}y_{1,j}^*\big(b_{j}-\mathbf{a}_{j}^{T}\v^{*}-\frac{y_{1,j}^*}{\rho+\mu}\big),
\end{split}
\end{equation}
where the last equality holds since $\z_1^*=\b-\A\v^{*}$, which follows from $\e_1^*=\z_1^*$ and $\A\v^{*}+\e_1^{*}-\b=0$.
Moreover, since $\A\v^{*}+\e_1^{*}-\b=0$, we also have
\begin{equation}\label{y-z-star-2}
\begin{split}
 (\y_1^{*})^T\e_1^* & \!=\! (\y_1^{*})^T(\b-\A\v^{*}) \\
& \!=\!\!\! \sum_{e_{1,j}^{*}>0}\!\! y_{1,j}^*(b_{j}\!-\!\!\mathbf{a}_{j}^{T}\x^{*})\!+\!\!\!\!\sum_{e_{1,j}^{*}=0}\!\! y_{1,j}^*(b_{j}\!-\!\!\mathbf{a}_{j}^{T}\x^{*}) \\
& \!=\! \sum_{e_{1,j}^{*}>0} y_{1,j}^*(b_{j}-\mathbf{a}_{j}^{T}\v^{*}).
\end{split}
\end{equation}

Comparing \eqref{y-z-star-1} and \eqref{y-z-star-2}, we can see that when $e_{1,j}^{*}>0$,
\begin{equation}\label{z-star-dayu0-1}
\begin{split}
& y_{1,j}^* =0.
\end{split}
\end{equation}
Therefore, we can obtain
\begin{equation}\label{y-z-star-0}
(\y_1^{*})^T\e_1^* =0,
\end{equation}
since $\e_1^*\succeq\mathbf{0}$.

On the other hand, if $e_{1,j}^{*}=0$, there exists $\frac{\mu}{\rho+\mu}(b_{j}-\mathbf{a}_{j}^{T}\v^{*}-\frac{y_{1,j}^*}{\mu}+\frac{\rho}{\mu}z_{1,j}^{*})\leq0$ (see \eqref{v1-solution-component}).
Moreover, since $\e_1^*=\z_1^*=\b-\A\v^{*}$, one can see that $\frac{y_{1,j}^*}{\mu} \geq (1+\frac{\rho}{\mu})e_{1,j}^{*}\geq0$.
Besides, since $\rho >0$ and $\mu>0$, thus we have
\begin{equation}\label{z-star-dayu0-2}
\begin{split}
& y_{1,j}^* \geq 0,
\end{split}
\end{equation}
when $e_{1,j}^{*}=0$.
From \eqref{z-star-dayu0-1} and \eqref{z-star-dayu0-2}, we conclude
\begin{equation}\label{y-star-dayu0}
\y_1^{*} \succeq 0.
\end{equation}
Plugging \eqref{y-z-star-0} and \eqref{y-star-dayu0} into \eqref{first-term1},
we obtain $(\y_1^{*})^T(\e_1^*-\e_1)=-(\y_1^{*})^T \e_1 \leq 0$,
which means
\begin{equation}\label{first-term-proof}
(\y_1^{*})^T\A(\v-\v^{*}) \leq 0.
\end{equation}

Similar to the above derivations for \eqref{first-term-proof}, we can also have
\begin{equation}\label{second-term-proof}
(\y_2^*)^{T}(\v-\v^{*})\leq 0.
\end{equation}

Therefore, we can conclude
\begin{equation}\label{stationary3}
(\v-\v^{*})^T\nabla_{\v} g(\v^*) \geq 0, \ \forall \v \in \mathcal{X}.
\end{equation}
This completes the proof.
\end{proof}

\section{Proof of \eqref{case1-inequality3}}
We prove by contradiction. Suppose \eqref{case1-inequality3} does not hold. Since $\underset{\epsilon\rightarrow0}\lim \sigma(\epsilon)=0$, and suppose there exists a sequence of tuples
\begin{equation}\label{con_hat}
\begin{split}
\underset{k\rightarrow+\infty}\lim\{\p^k,\v_1^k,\v_2^k,\y_1^{k+1},\y_2^{k+1}\}= \{\hat{\p},\hat{\v}_1,\hat{\v}_2,\hat{\y}_1,\hat{\y}_2\},
\end{split}
\end{equation}
such that
\begin{equation}\label{contrary1}
\begin{split}
\underset{k\rightarrow+\infty}\lim\|\pmb{\phi}^k\|_2 = 0.
\end{split}
\end{equation}
\begin{equation}\label{contrary2}
\begin{split}
\underset{k\rightarrow+\infty}\lim\|\mathcal{X}^k\|_2 > 0.
\end{split}
\end{equation}
Next we prove \eqref{contrary1} and \eqref{contrary2} cannot hold simultaneously.

First, plugging the convergence result \eqref{con_hat} into \eqref{contrary1} and noticing
$\left[\!\!\!\!\begin{array}{l}{\v(\p^{k},\z_{1}^{k},\z_{2}^{k},\y_{1}^{k+1}\!, \y_{2}^{k+1}\!)}
\\ {\e_{1}\!(\p^{k}, \z_{1}^{k}, \z_{2}^{k}, \y_{1}^{k+1}\!, \y_{2}^{k+1}\!)}
\\ {\e_{2}\!(\p^{k}, \z_{1}^{k}, \z_{2}^{k}, \y_{1}^{k+1}\!, \y_{2}^{k+1}\!)}\end{array}\!\!\!\!\right] $ is continuous (see \eqref{x-D-define} and $\pmb\phi^k$ in \eqref{D-part1-change}), we can see
\begin{equation}\label{phi-0}
\begin{split}
\hspace{-10pt}\left[\!\!\!\begin{array}{c}{\A \v\left(\hat{\p},\hat{\z}_1,\hat{\z}_2,\hat{\y}_1,\hat{\y}_2\right)\!+\!\e_{1}\left(\hat{\p},\hat{\z}_1,\hat{\z}_2,\hat{\y}_1,\hat{\y}_2\right)\!-\!\b}
\\ {\v\left(\hat{\p},\hat{\z}_1,\hat{\z}_2,\hat{\y}_1,\hat{\y}_2\right)\!-\!\e_{2}\left(\hat{\p},\hat{\z}_1,\hat{\z}_2,\hat{\y}_1,\hat{\y}_2\right)}\end{array}\!\!\!\right]\!=\!0.
\end{split}
\end{equation}
Plugging \eqref{phi-0} into KKT equations of problem \eqref{x-D-define} and notice $\hat{\y}_1$ and $\hat{\y}_2$ are the corresponding optimal Lagrangian multipliers (problem \eqref{x-D-define} is strongly convex), we can see
$$\left[\!\!\!\!\begin{array}{l}{\v(\hat{\p},\hat{\z}_{1},\hat{\z}_{2},\hat{\y}_{1}\!, \hat{\y}_{2}\!)}
\\ {\e_{1}(\hat{\p},\hat{\z}_{1},\hat{\z}_{2},\hat{\y}_{1}\!, \hat{\y}_{2}\!)}
\\ {\e_{2}(\hat{\p},\hat{\z}_{1},\hat{\z}_{2},\hat{\y}_{1}\!, \hat{\y}_{2}\!)}\end{array}\!\!\!\!\right]
= \left[\!\!\!\!\begin{array}{l}{\v(\hat{\p},\hat{\z}_{1},\hat{\z}_{2})}
\\ {\e_{1}(\hat{\p},\hat{\z}_{1},\hat{\z}_{2})}
\\ {\e_{2}(\hat{\p},\hat{\z}_{1},\hat{\z}_{2})}\end{array}\!\!\!\!\right]$$
which indicates $\underset{k\rightarrow+\infty}\lim\|\mathcal{X}^k\|_2 = 0$. This is a contradiction.

\section{Proof of \emph{Lemma \ref{error-bound-lemma}}}\label{error-bound-proof}
\begin{proof}
First, we prove \eqref{error-bound1}. To simplify the proof, $\F(\v,\e_{1},\e_{2},\bullet^{k})$ is used
to denote function $\F(\v,\e_{1},\e_{2},\p^k,\z_1^k,\z_2^k,\y_1^k,\y_2^k)$.
Moreover, we give the following definitions relative to $\v$, $\e_{1}$ and $\e_{2}$ respectively
\begin{equation}\label{error3-defi}
\begin{split}
&\mathbf{r}_{\v}^{k}\!= \!\! \nabla_{\v}\F(\v^{k}\!,\e_{1}^{k},\e_{2}^{k},\bullet^{k})
     \!-\!\!\nabla_{\v}\F(\v^{k+1}\!\!,\e_{1}^{k},\e_{2}^{k},\bullet^{k})\!\!+\!\v^{k+1}\!\!\!-\!\v^{k},\\
&\mathbf{r}_{\e_{1}}^{k}\!\!\!=\!\!\nabla_{\e_1}\!\F(\!\v^{k},\e_{1}^{k},\e_{2}^{k},\!\bullet^{k})
     \!\!-\!\!\!\nabla_{\e_1}\!\F(\!\v^{k+1}\!\!,\e_{1}^{k+1}\!\!,\e_{2}^{k},\!\bullet^{k})\!\!+\!\e_{1}^{k+1}\!\!\!-\!\e_{1}^{k},\\
&\mathbf{r}_{\mathbf{e}_{2}}^{k}\!\!\!=\!\!\!\nabla_{\e_{2}}\!\F(\!\v^{k},\!\e_{1}^{k},\!\e_{2}^{k},\!\bullet^{k})
     \!\!-\!\!\nabla_{\mathbf{v}_{2}}\!\F(\!\v^{k+1}\!,\!\e_{1}^{k+1}\!,\!\e_{2}^{k+1}\!,\!\bullet^{k})\!\!+\!\!\e_{2}^{k+1}\!\!\!-\!\!\e_{2}^{k}.
\end{split}
\end{equation}

Applying the triangle inequality to $\mathbf{r}_{\v}^{k}$, we obtain
\begin{equation}\label{error1-derivative}
\begin{split}
 \|\mathbf{r}_{\v}^{k}\|_{2}
& \!\leq\! \|\nabla_{\v}\F(\!\v^{k},\!\e_{1}^{k},\!\e_{2}^{k},\!\bullet^{k}\!)\!\!-\!\!\nabla_{\v}\F(\!\v^{k+1}\!\!,\e_{1}^{k},\!\e_{2}^{k},\!\bullet^{k}\!)\|_{2} \\
&~~~  +\|\v^{k+1}-\v^{k}\|_{2},\\
&  \!\leq\! (\rho+L+1)\|\v^{k+1}-\v^{k}\|_2,
\end{split}
\end{equation}
where $L$ is the Lipschitz constant defined in \eqref{x-Lipschitz-defi}.
Following similar derivations to \eqref{error1-derivative}, we can get
\begin{equation}\label{error2-derivative}
\begin{split}
\hspace{-0.15cm} \|\mathbf{r}_{\e_{1}}^{k}\!\|_{2} & \!\!\leq\!\! \|\nabla\!_{\e_{1}}\!\F\!(\v^{k}\!,\e_{1}^{k},\e_{2}^{k},\!\bullet^{k}\!)\!\!-\!\!\nabla\!_{\e_{1}}\!\F\!(\v^{k\!+\!1}\!\!,\e_{1}^{k\!+\!1}\!\!,\e_{2}^{k},\!\bullet^{k}\!)\|_{2} \\
&~~  \!+\!\|\e_{1}^{k+1}\!-\!\e_{1}^{k}\|_{2},\\
& \!\!=\!\! (\rho+\mu+1)\|\e_{1}^{k+1}-\e_{1}^{k}\|_{2},
\end{split}
\end{equation}
\begin{equation}\label{error3-derivative}
\begin{split}
\hspace{-0.2cm} \|\mathbf{r}_{\e_{2}}^{k}\!\|_{2} & \!\!\leq\!\! \|\nabla\!_{\e_{1}}\!\F\!(\v^{k}\!,\!\e_{1}^{k},\!\e_{2}^{k},\!\bullet^{k}\!)
     \!\!-\!\!\nabla\!_{\e_{1}}\!\F\!(\v^{k\!+\!1}\!\!,\e_{1}^{k\!+\!1}\!\!,\e_{2}^{k\!+\!1}\!\!,\!\bullet^{k}\!)\|_{2} \\
&~~ \!+\!\|\e_{2}^{k+1}\!-\!\e_{2}^{k}\|_{2},\\
& = (\rho+\mu+1)\|\e_{2}^{k+1}-\e_{2}^{k}\|_{2},
\end{split}
\end{equation}
Then, through \eqref{error1-derivative}-\eqref{error3-derivative}, we can obtain
 {\setlength\abovedisplayskip{2pt}
 \setlength\belowdisplayskip{2pt}
  \setlength\jot{1pt}
\begin{equation}\label{error4-derivative1}
\|\!\!\left[\!\!\begin{array}{c} {\mathbf{r}_{\v}^k} \\ {\mathbf{r}_{\e_{1}}^k} \\ {\mathbf{r}_{\e_{2}}^k}\end{array}\!\!\right]\!\!\|_2^2 \\
\leq  (\rho\!\!+\!\!L\!\!+\!\!1)^2\!\|\!\!\!\left[\!\!\!\begin{array}{c}{\v^{k\!+\!1}\!\!-\!\!\v^{k}} \\ {\e_{1}^{k\!+\!1}\!\!-\!\!\e_{1}^{k}} \\ {\e_{2}^{k\!+\!1}\!\!-\!\!\e_{2}^{k}}\end{array}\!\!\!\!\right]\!\!\!\|_{2}.
\end{equation}}

Moreover, since $\v^{k+1}$, $\e_{1}^{k+1}$, and $\e_{2}^{k+1}$ are minimizers of convex quadratic problems \eqref{proximal-x-update}, \eqref{proximal-v1-update}, and \eqref{proximal-v2-update} respectively, according to the fixed point theorem, we have
 {\setlength\abovedisplayskip{2pt}
 \setlength\belowdisplayskip{2pt}
  \setlength\jot{3pt}
\begin{equation}\label{x-v-other-expression}
\begin{split}
& \v^{k+1} = [\v^{k+1}-\nabla_{\v}\F(\v^{k+1},\e_{1}^{k},\e_{2}^{k},\bullet^{k})]_{+},\\
&\e_{1}^{k+1} = [\e_{1}^{k+1}-\nabla_{\e_{1}}\F(\v^{k+1},\e_{1}^{k+1},\e_{2}^{k},\bullet^{k})]_{+},\\
&\e_{2}^{k+1} = [\e_{2}^{k+1}-\nabla_{\e_{2}}\F(\v^{k+1},\e_{1}^{k+1},\e_{2}^{k+1},\bullet^{k})]_{+}.
\end{split}
\end{equation}
Plugging} \eqref{x-v-other-expression} into \eqref{error3-defi}, we have
\begin{equation}\label{hatx-other-expression}
\begin{split}
\left[\!\!\begin{array}{c}{\v^{k+1}} \\ {\e_{1}^{k+1}} \\ {\e_{2}^{k+1}}\end{array}\!\!\!\right]
\!\!=\!\!\left[\!\!\!\begin{array}{c}{\v^{k}-\nabla_{\v}\F(\v^{k},\e_{1}^{k},\e_{2}^{k},\bullet^{k})+\mathbf{r}_{\v}^k}
\\ {\e_{1}^{k}-\nabla_{\e_{1}}\F(\v^{k},\e_{1}^{k},\e_{2}^{k},\bullet^{k})+\mathbf{r}_{\e_{1}}^k}
\\ {\e_{2}^{k}-\nabla_{\e_{2}}\F(\v^{k},\e_{1}^{k},\e_{2}^{k},\bullet^{k})+\mathbf{r}_{\e_{2}}^k}\end{array}\!\!\!\right]_{+}.
\end{split}
\end{equation}
Then, we have the following derivations
 {\setlength\abovedisplayskip{1pt}
 \setlength\belowdisplayskip{1pt}
  \setlength\jot{1pt}
\begin{equation}\label{errorsum-derivative-2}
\begin{split}
& \|\!\!\left[\!\!\begin{array}{c}{\v^{k}}-\v^{k+1} \\ {\e_{1}^{k}}-{\e_{1}^{k+1}} \\ {\e_{2}^{k}}-{\e_{2}^{k+1}}\end{array}\!\!\right]\!\|_{2} \\
= & \|\!\!\left[\!\!\begin{array}{c}{\v^{k}} \\ {\e_{1}^{k}} \\ {\e_{2}^{k}}\end{array}\!\!\right]
\!\!-\!\!\left[\!\!\!\begin{array}{c}{\v^{k}\!-\!\nabla_{\v}\F(\v^{k},\e_{1}^{k},\e_{2}^{k},\bullet^{k})+\mathbf{r}_{\v}^k}
\\ {\e_{1}^{k}\!-\!\nabla_{\e_{1}}\F(\v^{k},\e_{1}^{k},\e_{2}^{k},\bullet^{k})+\mathbf{r}_{\e_{1}}^k}
\\ {\e_{2}^{k}\!-\!\nabla_{\e_{2}}\F(\v^{k},\e_{1}^{k},\e_{2}^{k},\bullet^{k})+\mathbf{r}_{\e_{2}}^k}\end{array}\!\!\!\right]_{+}\!\!\|_{2}\\
\geq & \|\!\!\left[\!\!\begin{array}{c}{\v^{k}} \\ {\e_{1}^{k}} \\ {\e_{2}^{k}}\end{array}\!\!\right]
-\!\left[\!\!\!\begin{array}{c}{\v^{k}-\nabla_{\v}\F(\v^{k},\e_{1}^{k},\e_{2}^{k},\bullet^{k})}
\\ {\e_{1}^{k}-\nabla_{\e_{1}}\F(\v^{k},\e_{1}^{k},\e_{2}^{k},\bullet^{k})}
\\ {\e_{2}^{k}-\nabla_{\e_{2}}\F(\v^{k},\e_{1}^{k},\e_{2}^{k},\bullet^{k})}\end{array}\!\!\!\right]_{+}\!\!\|_{2} \\
& - \|\!\left[\!\!\!\begin{array}{c}{\v^{k}-\nabla_{\v}\F(\v^{k},\e_{1}^{k},\e_{2}^{k},\bullet^{k})+\mathbf{r}_{\v}^k}
\\ {\e_{1}^{k}-\nabla_{\e_{1}}\F(\v^{k},\e_{1}^{k},\e_{2}^{k},\bullet^{k})+\mathbf{r}_{\e_{1}}^k}
\\ {\e_{2}^{k}-\nabla_{\e_{2}}\F(\v^{k},\e_{1}^{k},\e_{2}^{k},\bullet^{k})+\mathbf{r}_{\e_{2}}^k}\end{array}\!\!\!\right]_{+} \\
&~~~~ -\!\left[\!\!\!\begin{array}{c}{\v^{k}-\nabla_{\v}\F(\v^{k},\e_{1}^{k},\e_{2}^{k},\bullet^{k})}
\\ {\e_{1}^{k}-\nabla_{\e_{1}}\F(\v^{k},\e_{1}^{k},\e_{2}^{k},\bullet^{k})}
\\ {\e_{2}^{k}-\nabla_{\e_{2}}\F(\v^{k},\e_{1}^{k},\e_{2}^{k},\bullet^{k})}\end{array}\!\!\!\right]_{+}\!\!\|_{2},
\end{split}
\end{equation}
where} the inequality comes from the triangle inequality. Following the non-expansiveness property of projection operations, \eqref{errorsum-derivative-2} can be deduced to
 {\setlength\abovedisplayskip{3pt}
 \setlength\belowdisplayskip{3pt}
  \setlength\jot{1pt}
\begin{equation}\label{errorsum-derivative-3}
\begin{split}
\hspace{-0.5cm} &\|\!\!\left[\!\!\begin{array}{c}{\v^{k}}-\v^{k+1} \\ {\e_{1}^{k}}-{\e_{1}^{k+1}} \\ {\e_{2}^{k}}-{\e_{2}^{k+1}}\end{array}\!\!\right]\!\|_{2} \\
\overset{b} \geq & \|\!\!\left[\!\!\!\begin{array}{c}{\v^{k}} \\ {\e_{1}^{k}} \\ {\e_{2}^{k}}\end{array}\!\!\right]
\!\!\!-\!\!\!\left[\!\!\!\!\begin{array}{c}{\v^{k}\!\!-\!\!\nabla_{\v}\F(\v^{k},\!\e_{1}^{k},\!\e_{2}^{k},\bullet^{k})}
\\ {\e_{1}^{k}\!\!-\!\!\nabla_{\e_{1}}\!\F(\v^{k},\!\e_{1}^{k},\!\e_{2}^{k},\bullet^{k})}
\\ {\e_{2}^{k}\!\!-\!\!\nabla_{\e_{2}}\!\F(\v^{k},\!\e_{1}^{k},\!\e_{2}^{k},\bullet^{k})}\end{array}\!\!\!\!\right]_{+}\!\|_{2}
\!-\!\|\!\!\left[\!\!\!\!\begin{array}{c} {\mathbf{r}_{\v}^k} \\ {\mathbf{r}_{\e_{1}}^k} \\ {\mathbf{r}_{\e_{2}}^k}\end{array}\!\!\!\right]\!\!\|_2.
\end{split}
\end{equation}
Moreover,} \eqref{errorsum-derivative-3} can be further derived as
 {\setlength\abovedisplayskip{3pt}
 \setlength\belowdisplayskip{2pt}
  \setlength\jot{1pt}
\begin{equation}\label{errorsum-derivative-4}
\begin{split}
& \|\!\!\left[\!\!\begin{array}{c}{\v^{k}}-\v^{k+1} \\ {\e_{1}^{k}}-{\e_{1}^{k+1}} \\ {\e_{2}^{k}}-{\e_{2}^{k+1}}\end{array}\!\!\right]\!\|_{2} \\
\overset{a} \geq & (\rho-\alpha) \|\!\!\left[\!\!\begin{array}{c}{\v^{k}}-{\v(\p^{k},\z_1^{k},\z_2^{k},\y_1^k,\y_2^k)} \\ {\e_{1}^{k}}-{\e_{1}(\p^{k},\z_1^{k},\z_2^{k},\y_1^k,\y_2^k)} \\ {\e_{2}^{k}}-{\e_{2}(\p^{k},\z_1^{k},\z_2^{k},\y_1^k,\y_2^k)}\end{array}\!\!\right]
\!\|_{2}
-\|\!\!\left[\!\!\begin{array}{c} {\mathbf{r}_{\v}^k} \\ {\mathbf{r}_{\e_{1}}^k} \\ {\mathbf{r}_{\e_{2}}^k}\end{array}\!\!\right]\!\!\|_2 \\
\overset{b} \geq & (\rho\!\!-\!\!\alpha) \|\!\!\!\left[\!\!\!\begin{array}{c}{\v^{k}}\!-\!{\v(\p^{k}\!,\z_1^{k}\!,\z_2^{k}\!,\y_1^k\!,\y_2^k)} \\ {\e_{1}^{k}}\!-\!{\e_{1}(\p^{k}\!,\z_1^{k}\!,\z_2^{k}\!,\y_1^k\!,\y_2^k)} \\ {\e_{2}^{k}}\!-\!{\e_{2}(\p^{k}\!,\z_1^{k}\!,\z_2^{k}\!,\y_1^k,\y_2^k)}\end{array}\!\!\!\right]
 \!\!\!\|_{2}\!\!-\!\!(\rho\!\!+\!\!L\!\!+\!\!1)\!\|\!\!\!\left[\!\!\!\begin{array}{c}{\v^{k\!+\!1}\!\!-\!\!\v^{k}} \\ {\e_{1}^{k\!+\!1}\!\!-\!\!\e_{1}^{k}} \\ {\e_{2}^{k\!+\!1}\!\!-\!\!\e_{2}^{k}}\end{array}\!\!\!\!\right]\!\!\!\|_{2},
\end{split}
\end{equation}
where} ``$\overset{a} \geq$'' holds since $\mathcal{F}(\mathbf{v},\mathbf{e}_1,\mathbf{e}_2, \bullet^k)$ is a strongly convex function with modulus $(\rho-\alpha)>0$ \cite{global-error-bound} and ``$\overset{b} \geq $'' follows from \eqref{error4-derivative1}.
Then, we can obtain \eqref{error-bound1} as follows
\begin{equation*}\label{errorsum-derivative-5}
\begin{split}
& \|\!\!\!\left[\!\!\!\begin{array}{l}{\v^{k}\!-\!\v^{k+1}} \\ {\e_{1}^{k}\!-\!\e_{1}^{k+1}} \\ {\e_{2}^{k}\!-\!\e_{2}^{k+1}}\end{array}\!\!\!\right]\!\!\!\|_{2}^{2}
\!\geq\! \varepsilon_{1}\|\!\!\!\left[\!\!\!\begin{array}{c}{\v^{k}\! -\!\v\left(\p^{k}, \z_{1}^{k}, \z_{2}^{k}, \y_{1}^{k}, \y_{2}^{k}\right)}
 \\ {\e_{1}^{k}\!-\!\e_{1}\left(\p^{k}, \z_{1}^{k}, \z_{2}^{k}, \y_{1}^{k}, \y_{2}^{k}\right)}
 \\ {\e_{2}^{k}\!-\!\e_{2}\left(\p^{k}, \z_{1}^{k}, \z_{2}^{k}, \y_{1}^{k}, \y_{2}^{k}\right)}\end{array}\!\!\!\right]\!\!\!\|_{2}^{2},
\end{split}
\end{equation*}
where $\varepsilon_1=\frac{(\rho-\alpha)^2}{(\rho+L+2)^2}$.

Next, we prove that \eqref{error-bound2} holds.
Based on the triangle inequality and \eqref{error-bound1}, we have
\begin{equation}\label{error-bound2-proof}
\begin{split}
& \|\!\!\!\left[\!\!\!\begin{array}{c}{\v^{k+1}\! -\!\v\left(\p^{k}, \z_{1}^{k}, \z_{2}^{k}, \y_{1}^{k}, \y_{2}^{k}\right)}
 \\ {\e_{1}^{k+1}\!-\!\e_{1}\left(\p^{k}, \z_{1}^{k}, \z_{2}^{k}, \y_{1}^{k}, \y_{2}^{k}\right)}
 \\ {\e_{2}^{k+1}\!-\!\e_{2}\left(\p^{k}, \z_{1}^{k}, \z_{2}^{k}, \y_{1}^{k}, \y_{2}^{k}\right)}\end{array}\!\!\!\right]\!\!\!\|_{2}^{2} \\
\leq &  \|\!\!\!\left[\!\!\!\begin{array}{l}{\v^{k+1}\!-\!\v^{k}} \\ {\e_{1}^{k+1}\!-\!\e_{1}^{k}} \\ {\e_{2}^{k+1}\!-\!\e_{2}^{k}}\end{array}\!\!\!\right]\!\!\!\|_{2}^{2}
+ \|\!\!\!\left[\!\!\!\begin{array}{c}{\v_{k}\! -\!\v\left(\p^{k}, \z_{1}^{k}, \z_{2}^{k}, \y_{1}^{k}, \y_{2}^{k}\right)}
 \\ {\e_{1}^{k}\!-\!\e_{1}\left(\p^{k}, \z_{1}^{k}, \z_{2}^{k}, \y_{1}^{k}, \y_{2}^{k}\right)}
 \\ {\e_{2}^{k}\!-\!\e_{2}\left(\p^{k}, \z_{1}^{k}, \z_{2}^{k}, \y_{1}^{k}, \y_{2}^{k}\right)}\end{array}\!\!\!\right]\!\!\!\|_{2}^{2}  \\
\leq & \left(\!1+\!\frac{1}{\sqrt{\varepsilon_{1}}}\!\right)\!\!\|\!\!\!\left[\!\!\!\begin{array}{l}{\v^{k+1}\!-\!\v^{k}} \\ {\e_{1}^{k+1}\!-\!\e_{1}^{k}} \\ {\e_{2}^{k+1}\!-\!\e_{2}^{k}}\end{array}\!\!\!\right]\!\!\!\|_{2}^{2},
\end{split}
\end{equation}
i.e.,
\begin{equation*}\label{error-bound2-proof}
\begin{split}
& \|\!\!\!\left[\!\!\!\!\begin{array}{l}{\v^{k}\!\!-\!\!\v^{k+1}} \\ {\e_{1}^{k}\!\!-\!\!\e_{1}^{k+1}} \\ {\e_{2}^{k}\!\!-\!\!\e_{2}^{k+1}}\end{array}\!\!\!\right]\!\!\!\|_{2}^{2}
\!\geq\! \!\varepsilon_{2}\|\!\!\!\left[\!\!\!\!\begin{array}{c}{\v^{k+1}\!\! -\!\v\left(\p^{k}, \z_{1}^{k}, \z_{2}^{k}, \y_{1}^{k}, \y_{2}^{k}\right)}
 \\ {\e_{1}^{k+1}\!\!-\!\e_{1}\left(\p^{k}, \z_{1}^{k}, \z_{2}^{k}, \y_{1}^{k}, \y_{2}^{k}\right)}
 \\ {\e_{2}^{k+1}\!\!-\!\e_{2}\left(\p^{k}, \z_{1}^{k}, \z_{2}^{k}, \y_{1}^{k}, \y_{2}^{k}\right)}\end{array}\!\!\!\right]\!\!\!\|_{2}^{2},
\end{split}
\end{equation*}
where $\varepsilon_{2}=\frac{(\rho-\alpha)^2}{(2\rho+L+2-\alpha)^2}$.

Moreover, through similar proofs for (3.6)-(3.8) in \cite{proximal-admm}, we can verify that inequalities \eqref{error-bound3}-\eqref{error-bound5} hold.
In addition, since $g(\v, \e_{1}, \e_{2})=g(\v)=\pmb{\lambda}^T\v-\frac{\alpha}{2}\|\v-0.5\|_2^2$ is Lipschitz differentiable corresponding to variables $\v$, $\e_{1}$, and $\e_{2}$ with constant $L_{g}>\alpha$, based on \emph{Proposition~2.3} in \cite{proximal-admm}, we can see that problem \eqref{pADMM-frame-problem} satisfies the strict complementary condition. Then, we can prove that inequality \eqref{error-bound6} holds through similar derivations to (3.9) in \cite{proximal-admm}. This ends the proof.
\end{proof}

\section{Proof of \emph{Lemma \ref{F-D-P-change}}}\label{descent3-proof}
\begin{proof}
First, we define the following quantities
\begin{subequations}\label{define-divide-F}
\begin{align}
& \mathcal{F}_\v^k=\F^k-\F(\v^{k+1},\e_{1}^{k},\e_{2}^{k},\p^k,\z_1^k,\z_2^k,\y_1^k,\y_2^k),\label{define-divide-Fv} \\
& \mathcal{F}_{\e_1}^k= \mathcal{F}(\v^{k+1},\e_{1}^{k},\e_{2}^{k},\p^k,\z_1^k,\z_2^k,\y_1^k,\y_2^k) \nonumber \\
& ~~~~~~~~ -\mathcal{F}(\v^{k+1},\e_{1}^{k+1},\e_{2}^{k},\p^k,\z_1^k,\z_2^k,\y_1^k,\y_2^k),\label{define-divide-Fe1} \\
& \mathcal{F}_{\e_2}^k= \mathcal{F}(\v^{k+1},\e_{1}^{k+1},\e_{2}^{k},\p^k,\z_1^k,\z_2^k,\y_1^k,\y_2^k) \nonumber \\
& ~~~~~~~~ -\mathcal{F}(\v^{k+1},\e_{1}^{k+1},\e_{2}^{k+1},\p^k,\z_1^k,\z_2^k,\y_1^k,\y_2^k),\label{define-divide-Fe2} \\
& \mathcal{F}_{\p\z}^k= \mathcal{F}(\v^{k+1},\e_{1}^{k+1},\e_{2}^{k+1},\p^k,\z_1^k,\z_2^k,\y_1^k,\y_2^k) \nonumber \\
& ~~~~~~~~ - \F(\v^{k\!+\!1}\!,\e_{1}^{k\!+\!1}\!,\e_{2}^{k\!+\!1}\!,\p^{k\!+\!1}\!,\z_1^{k\!+\!1}\!,\z_2^{k\!+\!1}\!,\y_1^k,\y_2^k), \label{define-divide-Fpz} \\
& \mathcal{F}_{\y}^k\!=\! \F(\v^{k\!+\!1}\!\!,\e_{1}^{k\!+\!1}\!\!,\e_{2}^{k\!+\!1}\!\!,\p^{k\!+\!1}\!\!,\z_1^{k\!+\!1}\!\!,\z_2^{k\!+\!1}\!\!,\y_1^k,\y_2^k)\!-\!\F^{k+1}. \label{define-divide-Fy}
\end{align}
\end{subequations}
It is easy to see
\begin{equation}\label{primal-descent1}
\F^k-\F^{k+1}  =  \mathcal{F}_\v^k+\mathcal{F}_{\e_1}^k+\mathcal{F}_{\e_2}^k+\mathcal{F}_{\p\z}^k+\mathcal{F}_{\y}^k.
\end{equation}
From \eqref{proximal-x-update} and \eqref{F-define}, we can find
\[\v^{k+1}=\underset{\v}{\rm argmin}\F(\v, \e_{1}^{k}\!,\e_{2}^{k\!+\!1},\p^{k},\z_1^{k},\z_2^k,\y_1^k,\y_2^k).
\]
Since $\F(\v, \e_{1}^{k}\!,\e_{2}^{k},\p^{k},\z_1^{k},\z_2^k,\y_1^k,\y_2^k)$ is a strongly convex quadratic function, we have
\[
\begin{split}
\hspace{-0.34cm}\mathcal{F}_\v^k=& \frac{\rho\!+\!\mu\!-\!\alpha}{2}\|\v^{k}\!-\!\v^{k\!+\!1}\|_2^2 \!+\!\frac{\mu}{2}\|\A(\v^{k}\!-\!\v^{k\!+\!1})\|_2^2 \\
\geq & \frac{\rho+\mu-\alpha+\mu\lambda_{\min}(\A^T\A)}{2}\|\v^{k}-\v^{k+1}\|_2^2.
\end{split}
\]
Under the assumption of $\mu\lambda_{\min}(\A^T\A)\geq\alpha$, we can obtain
\begin{equation}\label{primal-descent1-a1}
\mathcal{F}_{\v}^k
\geq  \frac{\rho+\mu}{2}\|\v^{k}-\v^{k+1}\|_2^2.
\end{equation}

Since $\F(\v^{k+1}\!, \e_{1}\!,\e_{2}^{k},\p^{k},\z_1^{k},\z_2^k,\y_1^k,\y_2^k)$ is a strongly convex quadratic function, $\e_1\succeq\mathbf{0}$, and $\e_1^{k+1}$ is the minimizer of problem \eqref{proximal-v1-update}, we have
\begin{equation}\label{primal-descent1-e1}
  \mathcal{F}_{\e_1}^k \geq \frac{\rho+\mu}{2}\|\e_{1}^{k}-\e_1^{k+1}\|_2^2.
\end{equation}
Similarly, we can also obtain
\begin{equation}\label{primal-descent1-e2}
  \mathcal{F}_{\e_2}^k \geq \frac{\rho+\mu}{2}\|\e_{2}^{k}-\e_2^{k+1}\|_2^2.
\end{equation}

Plugging \eqref{proximal-z-update} into \eqref{define-divide-Fpz}, we can derive it as \eqref{primal-descent1-b}, where the last inequality holds since $0<\beta \leq 1$.
\begin{figure*}
\begin{equation}\label{primal-descent1-b}
\begin{split}
\mathcal{F}_{\p\z}^k =&\frac{\rho}{2}\!\Big(\!\big(\p^{k\!+\!1}\!-\!\p^k\big)\!^T\!\big(2\v^{k\!+\!1}\!-\!\p^{k\!+\!1}\!-\!\p^k\!\big)\!\!+\!\!\big(\z_1^{k\!+\!1}\!\!-\!\z_1^k\!\big)\!^T\!\big(2\e_1^{k\!+\!1}\! \!-\!\z_1^{k+1}\!-\!\z_1^k\big)\!\!+\!\!\big(\z_2^{k+1}\!\!-\!\z_2^k\big)^T\!\big(2\e_2^{k+1}\!-\!\z_2^{k+1}\!-\!\z_2^k\big)\!\Big) \\
= & \frac{p}{2}\!\Big(\!\frac{2}{\beta}\!\!-\!\!1\!\Big)\!\Big(\!\|\p^{k\!+\!1}\!\!\!-\!\p^k\|_2^2\!+\!\!\|\z_1^{k\!+\!1}\!\!-\!\z_1^k\|_2^2\!+\!\!\|\z_2^{k\!+\!1}\!\!-\!\z_2^k\|_2^2\!\Big) \\
\geq & \frac{p}{2\beta}\Big(\!\|\p^{k+1}\!\!-\!\p^k\|_2^2\!+\!\|\z_1^{k+1}\!\!-\!\z_1^k\|_2^2\!+\!\|\z_2^{k+1}\!\!-\!\z_2^k\|_2^2\Big)
\end{split}
\end{equation}
\hrulefill
\vspace*{4pt}
\end{figure*}

Plugging \eqref{proximal-lamda-update} into \eqref{define-divide-Fy}, we can obtain
\begin{equation}\label{primal-descent1-c}
\begin{split}
\mathcal{F}_{\y}^k =& \big(\y_1^k-\y_1^{k+1}\big)^T\big(\A\v^{k+1}+\e_{1}^{k+1}-\b\big) \\
& +\big(\y_2^k-\y_2^{k+1}\big)^T\big(\v^{k+1}-\e_{2}^{k+1}\big) \\
=& \!-\mu\|\A\v^{k+1}+\e_{1}^{k+1}-\b\|_2^2\!\! -\!\!\mu\|\v^{k+1}\!-\!\e_{2}^{k+1}\|_2^2.
\end{split}
\end{equation}

Then, plugging \eqref{primal-descent1-a1}--\eqref{primal-descent1-c} into \eqref{primal-descent1}, we have
\begin{equation}\label{F-change-proof}
\begin{split}
 \F^{k}\!-\!\F^{k+1}\!\! \geq \!& \frac{\rho\!+\!\mu}{2}\|\!\!\!\left[\!\!\!\!\begin{array}{l}{\v^{k}\!-\!\v^{k+1}} \\ {\e_{1}^{k}\!-\!\e_{1}^{k+1}} \\ {\e_{2}^{k}\!-\!\e_{2}^{k+1}}\end{array}\!\!\!\right]\!\!\!\|_{2}^{2}
\!+\!\!\frac{\rho}{2 \beta}\!\|\!\!\!\left[\!\!\!\begin{array}{l}{\p^{k}\!-\!\p^{k+1}} \\ {\z_{1}^{k}\!-\!\z_{1}^{k+1}} \\ {\z_{2}^{k}\!-\!\z_{2}^{k+1}}\end{array}\!\!\!\right]\!\!\!\|_{2}^{2} \\
&\!-\!\mu\|\!\!\left[\!\!\!\begin{array}{c}{\A \v^{k}+\e_{1}^{k}-\b} \\ {\v^{k}-\e_{2}^{k}}\end{array}\!\!\!\right]\!\!\|_{2}^{2}.
\end{split}
\end{equation}
That completes the proof of \eqref{F-change}.

Next, we consider to prove inequality \eqref{D-change}. To facilitate discussions later, we define
\begin{equation}\label{D-divide}
\begin{split}
& \D_{\p\z}^k=\D^{k+1}-\D(\p^k,\z_1^k,\z_2^k,\y_1^{k+1},\y_2^{k+1}), \\
& \D_{\y}^k= \D(\p^k,\z_1^k,\z_2^k,\y_1^{k+1},\y_2^{k+1})-\D^k.
\end{split}
\end{equation}
Then, we have
\begin{equation}\label{dual-descent-divide}
\begin{split}
& \D^{k+1}-\D^{k}=\D_{\p\z}^k+\D_{\y}^k.
\end{split}
\end{equation}
According to \eqref{D-define}, we can write $\D_{\p\z}^k$ as
\begin{equation}\label{dual-descent-a1-1}
\begin{split}
\hspace{-0.35cm} \D_{\p\z}^k \!\!=& \F\!\Big(\!\v(\!\p^{k\!+\!1}\!\!, \z_{1}^{k\!+\!1}\!\!, \z_{2}^{k+1}\!\!, \y_{1}^{k\!+\!1}\!\!, \y_{2}^{k\!+\!1}\!),\! \e_1(\!\p^{k\!+\!1}\!\!, \z_{1}^{k\!+\!1}\!\!, \z_{2}^{k\!+\!1}\!\!, \\
&~~~ \y_{1}^{k+1}\!\!, \y_{2}^{k+1}), \e_2(\p^{k+1}\!\!, \z_{1}^{k+1}\!\!, \z_{2}^{k+1}\!\!, \y_{1}^{k+1}\!\!, \y_{2}^{k+1}), \\
&~~~ \p^{k+1}\!\!, \z_{1}^{k+1}\!\!, \z_{2}^{k+1}\!\!, \y_{1}^{k+1}\!\!, \y_{2}^{k+1}\Big) \\
-\!& \F\!\Big(\!\v(\p^{k}\!\!, \z_{1}^{k}\!, \z_{2}^{k}\!, \y_{1}^{k\!+\!1}\!\!, \y_{2}^{k\!+\!1}\!),\! \e_1\!(\p^{k}\!, \z_{1}^{k}\!, \z_{2}^{k}\!, \y_{1}^{k\!+\!1}\!\!, \y_{2}^{k\!+\!1}\!), \\
&~~ \e_2(\p^{k}\!, \z_{1}^{k}, \z_{2}^{k}, \y_{1}^{k\!+\!1}\!\!, \y_{2}^{k\!+\!1}), \p^{k}, \z_{1}^{k}, \z_{2}^{k}, \y_{1}^{k\!+\!1}\!\!, \y_{2}^{k\!+\!1}\!\Big). \\
\end{split}
\end{equation}
According to \eqref{x-D-define}, \eqref{dual-descent-a1-1} can be rewritten as
 {\setlength\abovedisplayskip{8pt}
 \setlength\belowdisplayskip{8pt}
  \setlength\jot{5pt}
\begin{equation}\label{dual-descent-a1-2}
\begin{split}
 \D_{\p\z}^k \!\!
\geq& \F\!\Big(\!\v(\!\p^{k\!+\!1}\!\!, \z_{1}^{k\!+\!1}\!\!, \z_{2}^{k\!+\!1}\!\!, \y_{1}^{k\!+\!1}\!\!, \y_{2}^{k\!+\!1}\!),\! \e_1(\!\p^{k\!+\!1}\!\!, \z_{1}^{k\!+\!1}\!\!, \z_{2}^{k\!+\!1}\!\!, \\
&~~~ \y_{1}^{k+1}\!\!, \y_{2}^{k+1}), \e_2(\p^{k+1}\!\!, \z_{1}^{k+1}\!\!, \z_{2}^{k+1}\!\!, \y_{1}^{k+1}\!\!, \y_{2}^{k+1}), \\
&~~~ \p^{k+1}\!\!, \z_{1}^{k+1}\!\!, \z_{2}^{k+1}\!\!, \y_{1}^{k+1}\!\!, \y_{2}^{k+1}\Big) \\
-\!\!& \F\!\Big(\!\v(\p^{k\!+\!1}\!\!, \z_{1}^{k\!+\!1}\!\!, \z_{2}^{k\!+\!1}\!\!, \y_{1}^{k\!+\!1}\!\!, \y_{2}^{k\!+\!1}\!),\! \e_1(\p^{k\!+\!1}\!\!, \z_{1}^{k\!+\!1}\!\!, \z_{2}^{k\!+\!1}\!\!, \\
&~~~ \y_{1}^{k+1}\!\!, \y_{2}^{k+1}), \e_2(\p^{k+1}\!\!, \z_{1}^{k+1}\!\!, \z_{2}^{k+1}\!\!, \y_{1}^{k+1}\!\!, \y_{2}^{k+1}), \\
&~~~ \p^{k}, \z_{1}^{k}, \z_{2}^{k}, \y_{1}^{k+1}, \y_{2}^{k+1}\Big).
\end{split}
\end{equation}
Plugging} \eqref{F-define} into \eqref{dual-descent-a1-2}, we can obtain
 {\setlength\abovedisplayskip{8pt}
 \setlength\belowdisplayskip{8pt}
  \setlength\jot{5pt}
\begin{equation}\label{dual-descent-a1-3}
\begin{split}
\hspace{-0.5cm}\D_{\p\z}^k \!\!
\geq &\frac{\rho}{2}\!(\!\p^{k\!+\!1}\!\!\!-\!\!\p^{k})\!^T\!\!\!\left(\p^{k\!+\!1}\!\!\!+\!\!\p^{k}\!\!\!-\!\!2\v\!\left(\p^{k\!+\!1}\!\!, \z_{1}^{k\!+\!1}\!\!, \z_{2}^{k\!+\!1}\!\!, \y_{1}^{k\!+\!1}\!\!, \y_{2}^{k\!+\!1}\!\right)\!\right) \\
\!\!+ &  \frac{\rho}{2}\!(\!\z_{1}^{k\!+\!1}\!\!-\!\!\z_{1}^{k})\!^T\!\!\left(\z_{1}^{k\!+\!1}\!\!\!+\!\!\z_{1}^{k}\!\!-\!\!2 \e_{1}\!\!\left(\p^{k\!+\!1}\!\!, \z_{1}^{k\!+\!1}\!\!, \z_{2}^{k\!+\!1}\!\!, \y_{1}^{k\!+\!1}\!\!, \y_{2}^{k\!+\!1}\!\right)\!\right) \\
\!\!+ &  \frac{\rho}{2}\!(\!\z_{2}^{k\!+\!1}\!\!-\!\!\z_{2}^{k})\!^T\!\!\left(\z_{2}^{k\!+\!1}\!\!\!+\!\!\z_{2}^{k}\!\!-\!\!2 \e_{2}\!\!\left(\p^{k\!+\!1}\!\!, \z_{1}^{k\!+\!1}\!\!, \z_{2}^{k\!+\!1}\!\!, \y_{1}^{k\!+\!1}\!\!, \y_{2}^{k\!+\!1}\!\right)\!\right),
\end{split}
\end{equation}
which} can be further written as
\begin{equation}\label{dual-descent-a1-4}
\D_{\p\z}^k\geq\frac{\rho}{2}\!\left[\!\!\begin{array}{l}{\p^{k+1}-\p^{k}} \\ {\z_{1}^{k+1}-\z_{1}^{k}} \\ {\z_{2}^{k+1}-\z_{2}^{k}}\end{array}\!\!\right]^{T}\!\!\!\pmb{\psi}^k.
\end{equation}
Through similar derivations, we can also get
\begin{equation}\label{dual-descent-a2}
\begin{split}
\D_{\y}^k
\!\geq\!  \mu\! \left[\!\!\!\!\begin{array}{c}{\A\v^{k}\!+\!\e_{1}^{k}\!-\!\b} \\ {\v^{k}\!-\!\e_{2}^{k}}\end{array}\!\!\!\!\right]^{T}\!\!\!\!\pmb{\phi}^k.
\end{split}
\end{equation}

Thus, combining \eqref{dual-descent-a1-4} and \eqref{dual-descent-a2}, we obtain \eqref{D-change}.

Finally, we consider to prove \eqref{P-change}.
According to Danskin's theorem \cite{convex-analysis}, we have
\begin{equation}\label{P-gradient-formula}
  \begin{split}
& \nabla \P\left(\p^{k}, \z_{1}^{k}, \z_{2}^{k}\right)
= \rho\left[\!\!\begin{array}{l}{\p^{k}}-\v(\p^k,\z_1^k,\z_2^k) \\ {\z_{1}^{k}}-{\e_{1}(\p^k,\z_1^k,\z_2^k)} \\ {\z_{2}^{k}}-{\e_{2}(\p^k,\z_1^k,\z_2^k)}\end{array}\!\!\!\right].
  \end{split}
\end{equation}
Applying the triangle inequality property and \eqref{error-bound4} to \eqref{P-gradient-formula}, we have the following derivations
 {\setlength\abovedisplayskip{10pt}
 \setlength\belowdisplayskip{10pt}
  \setlength\jot{8pt}
\begin{equation}\label{P-gradient-cha}
  \begin{split}
\hspace{-0.4cm} & \left\|\nabla \P\left(\p^{k}, \z_{1}^{k}, \z_{2}^{k}\right)\!-\!\nabla \P\left(\p^{k+1}, \z_{1}^{k+1}, \z_{2}^{k+1}\right)\right\|_{2} \\
\hspace{-0.4cm} =&\rho\|\!\!\left[\!\!\!\begin{array}{l}{\p^{k}\!\!-\!\p^{k\!+\!1}} \\ {\z_{1}^{k}\!-\!\z_{1}^{k\!+\!1}} \\ {\z_{2}^{k}\!-\!\z_{2}^{k\!+\!1}}\end{array}\!\!\!\!\right]
\!\!\!-\!\!\!\left[\!\!\!\!\begin{array}{l}{\v(\p^{k\!+\!1}\!\!,\z_1^{k\!+\!1}\!\!,\z_2^{k\!+\!1})\!\!-\!\!\v(\p^{k}\!,\z_1^{k}\!,\z_2^{k})}
\\ {\e_1\!(\p^{k\!+\!1}\!\!,\z_1^{k\!+\!1}\!\!,\z_2^{k\!+\!1})\!\!-\!\!\e_1\!(\p^{k}\!,\z_1^{k}\!,\z_2^{k})}
\\ {\e_2\!(\p^{k\!+\!1}\!\!,\z_1^{k\!+\!1}\!\!,\z_2^{k\!+\!1})\!\!-\!\!\e_2\!(\p^{k}\!,\z_1^{k}\!,\z_2^{k})}\end{array}\!\!\!\!\right]\!\!\|_2 \\
\hspace{-0.4cm} \leq & \rho\|\!\!\!\left[\!\!\!\begin{array}{l}{\p^{k}\!\!\!-\!\!\p^{k\!+\!1}} \\ {\z_{1}^{k}\!\!-\!\!\z_{1}^{k\!+\!1}} \\ {\z_{2}^{k}\!\!-\!\!\z_{2}^{k\!+\!1}}\end{array}\!\!\!\!\right]\!\!\!\|_2
\!\!+\!\!\rho\|\!\!\left[\!\!\!\!\begin{array}{l}{\v(\p^{k\!+\!1}\!\!,\z_1^{k\!+\!1}\!\!,\z_2^{k\!+\!1})\!\!-\!\!\v(\p^{k}\!,\z_1^{k}\!,\z_2^{k})}
\\ {\e_1\!(\p^{k\!+\!1}\!\!,\z_1^{k\!+\!1}\!\!,\z_2^{k\!+\!1})\!\!-\!\!\e_1\!(\p^{k}\!,\z_1^{k}\!,\z_2^{k})}
\\ {\e_2\!(\p^{k\!+\!1}\!\!,\z_1^{k\!+\!1}\!\!,\z_2^{k\!+\!1})\!\!-\!\!\e_2\!(\p^{k}\!,\z_1^{k}\!,\z_2^{k})}\end{array}\!\!\!\!\right]\!\!\!\|_2 \\
\hspace{-0.4cm} \leq & \rho\left(1+\frac{1}{\sqrt{\varepsilon_{4}}}\right)
\|\!\!\left[\!\!\!\begin{array}{l}{\p^{k}\!\!-\!\p^{k\!+\!1}} \\ {\z_{1}^{k}\!-\!\z_{1}^{k\!+\!1}} \\ {\z_{2}^{k}\!-\!\z_{2}^{k\!+\!1}}\end{array}\!\!\!\!\right]\!\!\|_2,
  \end{split}
\end{equation}
which} means the gradient of function $\P$ is Lipschitz continuous with respect to variables $\p$, $\z_1$, and $\z_2$. The corresponding Lipschitz constant is $\rho\left(1+\frac{1}{\sqrt{\varepsilon_{4}}}\right)$.
Then, according to property of the Lipschitz continuous function \cite{nonlinear-programm}, we have
\begin{equation}\label{proximal-descent-result}
  \begin{split}
\hspace{-0.28cm} &   \P^{k+1}-\P^k  \\
\hspace{-0.28cm} \leq & \left(\!\left[\!\!\!\begin{array}{c}{\p^{k+1}} \\ {\z_{1}^{k+1}} \\ {\z_{2}^{k+1}}\end{array}\!\!\!\right]
\!\!-\!\!\left[\!\!\!\begin{array}{c}{\p^{k}} \\ {\z_{1}^{k}} \\ {\z_{2}^{k}}\end{array}\!\!\!\right]\!\right)^{T}\!\!\!
\left(\!\left[\!\!\!\begin{array}{c}{\p^{k+1}} \\ {\z_{1}^{k+1}} \\ {\z_{2}^{k+1}}\end{array}\!\!\!\right]
\!\!\!-\!\!\!\left[\!\!\!\begin{array}{c}{\v(\p^{k},\z_1^k,\z_2^k)} \\ {\e_{1}(\p^{k},\z_1^k,\z_2^k)} \\ {\e_{2}(\p^{k},\z_1^k,\z_2^k)}\end{array}\!\!\!\right]\!\right) \\
\hspace{-0.28cm} & +\frac{\rho}{2}\left(1+\frac{1}{\sqrt{\varepsilon_{4}}}\right)\|\!\!\left[\!\!\begin{array}{c}{\p^{k+1}} \\ {\z_{1}^{k+1}} \\ {\z_{2}^{k+1}}\end{array}\!\!\right]
\!\!-\!\!\left[\!\!\begin{array}{c}{\p^{k}} \\ {\z_{1}^{k}} \\ {\z_{2}^{k}}\end{array}\!\!\right]\!\!\|_2^2.
  \end{split}
\end{equation}
Letting $\eta = \left(1+\frac{1}{\sqrt{\varepsilon_{4}}}\right)$, we can reach  \eqref{P-change}. This ends the proof.
\end{proof}

\section{Proof of Fact \ref{W-property}}\label{W-property-proof}
\begin{proof}
To be clear, we rewrite $\mathbf{A}$ in \eqref{Ab-construct_A} as follows
\[
{\mathbf{A}}\!=\![\hat{\mathbf{W}}_1(\mathbf{Q}_1\!\otimes\!\mathbf{I}); \cdots;\!\hat{\mathbf{W}}_{\tau} \!(\mathbf{Q}_\tau\!\otimes\!\mathbf{I});\cdots;\hat{\mathbf{W}}_{\Gamma_c}(\mathbf{Q}_{\Gamma_c}\!\otimes\!\mathbf{I});\mathbf{S}].
\]
Moreover, $\hat{\W}_\tau$ in \eqref{W-hat-equal} can be written equivalently as
\[
\hat{\mathbf{W}}_{\tau}\!=\begin{bmatrix} \mathbf{P}\displaystyle\bigg(\sum_{i\in\mathcal{K}_1}\mathbf{T}_{i}\oplus\! 2\bigg)\mathbf{D}_\tau\oplus\! 2;\dotsb; \mathbf{P}\displaystyle\bigg(\!\sum_{i\in\mathcal{K}_{2^q-1}}\!\!\!\!\mathbf{T}_{i}\oplus\! 2\bigg)\mathbf{D}_\tau
\end{bmatrix}.
\]
 Based on the definitions of $\mathbf{T}_i$ and $\mathbf{D}_{\tau}$ (see \eqref{Ti} and \eqref{D}), we can derive each term in matrix $\hat{\mathbf{W}}_{\tau}$ as
\begin{equation}\label{W-hat-eachpart-equal}
\begin{split}
\hspace{-0.35cm}\mathbf{P}\!\displaystyle\bigg(\sum_{i\in\mathcal{K}_\ell}&\mathbf{T}_{i}\bigg)\mathbf{D}_\tau
=\mathbf{P}\textrm{diag}\!\bigg(\displaystyle\Big(\!\!\sum_{i\in\mathcal{K}_{\ell}}\!\!\mathbf{b}_{i}^T\!\!\oplus\! 2\!\Big)\mathbf{D}(2^q,h_{\tau_1}),\\
&\!\!\displaystyle\Big(\!\!\sum_{i\in\mathcal{K}_{\ell}}\!\!\mathbf{b}_{i}^T\!\!\oplus\! 2\!\Big)\mathbf{D}(2^q,h_{\tau_2}),
\displaystyle\Big(\!\!\sum_{i\in\mathcal{K}_\ell}\!\!\mathbf{b}_{i}^T\!\!\oplus\! 2\!\Big)\mathbf{D}(2^q,h_{\tau_3})\!\!\bigg).
\end{split}
\end{equation}

Since there is only one ``$1$'' in either column/row of matrix $\mathbf{D}(2^q,h_{\tau_j})$
and nonzero elements in $\displaystyle\sum_{i\in\mathcal{K}_\ell}\!\mathbf{b}_{i}^T\!\oplus 2$ are  ``$1$'',
row vector $\big(\displaystyle\sum_{i\in\mathcal{K}_\ell}\mathbf{b}_{i}^T\oplus 2\big)\mathbf{D}(2^q,h_{\tau_j})$ only includes one nonzero element ``$1$'', where $j=1,2,3$.
Moreover, since elements in matrix $\mathbf{P}$ are $1$ or $-1$, elements in $\mathbf{P}\!\displaystyle\bigg(\sum_{i\in\mathcal{K}_\ell}\mathbf{T}_{i}\bigg)\mathbf{D}_\tau$ are 0, $1$, or $-1$.
Therefore, we can conclude that element $\hat{\mathbf{W}}_{\tau}$ is also  either 0, $1$, or $-1$.

Moreover, since variable-selecting matrix $\mathbf{Q}_{\tau}$ has only one nonzero element ``1'' in its each row/column, $\mathbf{Q}_{\tau}\otimes \mathbf{I}$ should have the same property. Therefore, it is obvious that elements in $\hat{\mathbf{W}}_{\tau}(\mathbf{Q}_{\tau}\otimes \mathbf{I})$ are 1, -1, or 0. Besides, since matrix $\mathbf{S}$ only includes one nonzero element ``1'', we can conclude that matrix $\A$ consists of elements 1, -1, and 0. This completes the proof.
\end{proof}

%
%

\ifCLASSOPTIONcaptionsoff
  \newpage
\fi



%
%
%




\end{document}